\documentclass[11pt]{article}
\usepackage{amsgen,amsmath,amsthm,amstext,amsbsy,amsopn,amssymb, bbm, algorithm, algorithmic, enumerate, color, graphicx, tikz, natbib, pdflscape, multirow, epsfig}
\usepackage[titletoc]{appendix}
\usepackage[colorlinks,citecolor=blue,urlcolor=blue]{hyperref}
\usepackage[left=2.54cm,top=2.54cm,right=2.54cm,bottom=2.54cm]{geometry}

\definecolor{ZurichBlue}{rgb}{.255,.41,.884} % RoyalBlue of svgnames
\definecolor{ZurichRed}{rgb}{1, 0, 0} % Red of svgnames

\newtheorem{Corollary}{Corollary}
\newtheorem{Proposition}{Proposition}
\newtheorem{Lemma}{Lemma}

\newtheorem{Theorem}{Theorem}

\newtheorem{Remark}{Remark}
\newtheorem{Assumption}{Assumption}

\newcommand{\ba}{\mathbf a}

\newcommand{\bc}{\mathbf c}
\newcommand{\bd}{\mathbf d}
\newcommand{\f}{\mathbf f}
\newcommand{\hf}{{\hat{\mathbf f}}^{(n)}}

\newcommand{\hfnob}{{\hat{f}}^{(n)}}
\newcommand{\bu}{\mathbf u}

\newcommand{\x}{\mathbf x}
\newcommand{\y}{\mathbf y}
\newcommand{\hy}{\hat{y}}
\newcommand{\hby}{\hat{{\mathbf y}}}

\newcommand{\B}{{\mathbf{B}}}
\newcommand{\R}{\mathbb{R}}

\newcommand{\0}{\mathbf 0}

\newcommand{\bB}{\mathbf B}

\newcommand{\C}{\mathbf C}
\newcommand{\Chat}{\hat{\mathbf C}}
\newcommand{\hatC}{\hat{\mathbf C}^{(n)}}

\newcommand{\F}{\mathbf F}

\newcommand{\G}{\mathbf G}

\newcommand{\bH}{\mathbf H}
\newcommand{\Hsp}{\mathcal H}
\newcommand{\bHsp}{\boldsymbol{\mathcal H}}
\newcommand{\I}{\mathbf I}
\newcommand{\J}{\mathbf J}

\newcommand{\bM}{\mathbf M}
\newcommand{\N}{\mathbb{N}}
\newcommand{\Norm}{\mathcal{N}}

\newcommand{\mN}{\mathcal{N}}
\newcommand{\bP}{\mathbf P}
\newcommand{\Q}{\mathbf Q}

\newcommand{\bR}{\mathbf R}
\newcommand{\mR}{\mathcal{R}}
\newcommand{\br}{\mathbf r}

\newcommand{\bS}{\mathbf S}
\newcommand{\tS}{\tilde{\mathbf{S}}}
\newcommand{\bT}{\mathbf{T}}

\newcommand{\bV}{\mathbf{V}}

\newcommand{\hV}{\hat{\mathbf{V}}^{(n)}}
\newcommand{\tV}{\tilde{\mathbf{V}}^{(n)}}

\newcommand{\mX}{\mathcal{X}}

\newcommand{\hlambda}{{\hat{\lambda}^{(n)}}}

\newcommand{\bphi}{{\boldsymbol{\phi}}}

\newcommand{\tiphi}{\tilde{\phi}}
\newcommand{\hphi}{\hat{\phi}^{(n)}}
\newcommand{\bPhi}{{\boldsymbol{\Phi}}}

\newcommand{\hPhi}{{\hat{\boldsymbol{\Phi}}^{(n)}}}
\newcommand{\bPsi}{{\boldsymbol{\Psi}}}

\newcommand{\hsigma}{{\hat{\sigma}}}

\newcommand{\bbeta}{{\boldsymbol{\beta}}}

\newcommand{\bgamma}{{\boldsymbol{\gamma}}}

\newcommand{\cont}{C(\mX)}
\newcommand{\kxx}{k:\mX\times\mX\longrightarrow\R}
\newcommand{\Vleft}{(\C_{ii}+\epsilon\bP_i^1)^{-1/2}}
\newcommand{\Vright}{(\C_{jj}+\epsilon\bP_j^1)^{-1/2}}
\newcommand{\Vmid}{\C_{ij}}
\newcommand{\Vleftno}{\C_{ii}+\epsilon\bP_i^1}
\newcommand{\Vrightno}{\C_{jj}+\epsilon\bP_j^1}
\newcommand{\Vhatleft}{(\hat{\C}_{ii}+\epsilon\bP_i^1)^{-1/2}}
\newcommand{\Vhatright}{(\hat{\C}_{jj}+\epsilon\bP_j^1)^{-1/2}}
\newcommand{\Vhatmid}{\hat{\C}_{ij}}
\newcommand{\Vhatleftno}{\hat{\C}_{ii}+\epsilon\bP_i^1}
\newcommand{\Vhatrightno}{\hat{\C}_{jj}+\epsilon\bP_j^1}
\newcommand{\Var}{{\rm Var\,}}
\newcommand{\Varhat}{\widehat{\rm Var\,}}
\newcommand{\Cov}{{\rm Cov}}
\newcommand{\Span}{{\rm span}}
\newcommand{\Supp}{{\rm supp}}
\newcommand{\Dim}{{\rm dim}}
\newcommand{\half}{\frac{1}{2}}

\newcommand{\minmaxst}{\operatornamewithlimits{min/max/stationary}}
\newcommand{\argmin}{\operatornamewithlimits{argmin}}

\newcommand{\reals}{\rm I\!R}
\newcommand{\Phat}{{\hat{P}}}
\newcommand{\Hhat}{{\hat{H}}}
\newcommand{\tr}{{\mathrm {tr}}}

\begin{document}
\title{Kernel Additive Principal Components}
\author{Xin Lu Tan\footnote{xtan@wharton.upenn.edu}, Andreas Buja\footnote{buja.at.wharton@gmail.com}, and Zongming Ma\footnote{zongming@wharton.upenn.edu}
\vspace{.3cm} \\
Department of Statistics \\
The Wharton School \\
University of Pennsylvania \\
Philadelphia, PA 19104
}
\maketitle

\begin{abstract}
Additive principal components (APCs for short) are a nonlinear
  generalization of linear principal components.  We focus on smallest
  APCs to describe additive nonlinear constraints that are
  approximately satisfied by the data.  Thus APCs fit data with
  implicit equations that treat the variables symmetrically, as
  opposed to regression analyses which fit data with explicit
  equations that treat the data asymmetrically by singling out a
  response variable.  We propose a regularized data-analytic procedure
  for APC estimation using kernel methods.  In contrast to existing
  approaches to APCs that are based on regularization through subspace
  restriction, kernel methods achieve regularization through shrinkage
  and therefore grant distinctive flexibility in APC estimation by
  allowing the use of infinite-dimensional functions spaces for
  searching APC transformation while retaining computational
  feasibility.  To connect population APCs and kernelized
  finite-sample APCs, we study kernelized population APCs and their
  associated eigenproblems, which eventually lead to the establishment
  of consistency of the estimated APCs.  Lastly, we discuss an iterative 
  algorithm for computing kernelized finite-sample APCs.
\smallskip \\
\noindent{\bf Keywords:} \emph{Additive models, kernel methods, nonlinear multivariate analysis, principal
components, reproducing kernel Hilbert space.} \\
\end{abstract}

%%%%%%%%%%%%%%%%%%%%%%%%%%%%%%%%%%%%%%%%%%%%%%%%%%%

%!TEX root = Paper.tex

% Introduction
\section{Introduction}
\label{APCintro}

Principal component analysis (PCA) is a tool commonly used to reduce
the dimensionality of data sets consisting of a large number of
interrelated variables $X_1,X_2, \ldots,X_p$.  PCA amounts to finding
linear functions of the variables, $\sum a_jX_j$, whose variances are
maximal or, more generally, large and stationary under a unit norm
constraint, $\sum a_j^2=1$.  These linear combinations, called
\emph{largest linear principal components}, are thought to represent
low-dimensional linear structure of the data.  The reader is referred
to \cite{Jolliffe2002} for a comprehensive review of PCA.

One can similarly define the \emph{smallest linear principal
  component} as linear functions of the variables whose variances are
minimal or small and stationary subject to a unit norm constraint on
the coefficients.  If these variances are near zero,
$\Var(\sum a_jX_j)\approx 0$, the interpretation is that the data lie
near the hyperplane defined by the linear constraint $\sum a_jX_j = 0$
(assuming that the variables $X_j$ are centered).  Thus the purpose of
performing PCA on the lower end of the principal components spectrum
is quite different from that of performing it on the upper end:
largest principal components are concerned with structure of \emph{low
  dimension}, whereas smallest principal components are concerned with
structure of \emph{low codimension}.

Smallest additive principal components (``APCs'' for short) are a
nonlinear generalization of smallest linear principal components
(``LPCs'' for short), initially proposed in \cite{DonnellBuja1994}.
The \emph{smallest APC} is defined as an additive function of the
variables, $\sum\phi_j(X_j)$, with smallest variance subject to a
normalizing constraint $\sum \|\phi_j(X_j)\|^2 = 1$.  The
interpretation of a smallest APC is that the additive constraint
represented by the implicit additive equation $\sum\phi_j(X_j) = 0$
defines a nonlinear manifold that approximates the data.  The focus of
the current paper will be on smallest APCs, and we will therefore
sometimes drop the adjective ``smallest''.  Smallest APCs can be
motivated in several ways:
\begin{itemize}
\item APCs can be used as a generalized collinearity diagnostic for
  additive regression models.  Just as approximate collinearities
  $\sum \beta_j X_j \approx 0$ destabilize inference in linear
  regression, additive approximate ``concurvities''
  \citep{DonnellBuja1994} of the form $\sum\phi_j(X_j) \approx 0$
  destabilize inference in additive regression.  Such concurvities can
  be found by applying APC analysis to the predictors of an additive
  regression.
\item APCs can also be used as a symmetric alternative to additive
  regression as well as to ACE regression \citep{Breiman1985} when it
  is not possible or not desirable to single out any one of the
  variables as a response.  Additive implicit equations estimated with
  APCs will then freely identify the variables that have strong
  additive associations with each other.
\item Even when there is a specific response variable of interest in
  the context of an additive regression, an APC analysis of all
  variables, predictors as well as response, can serve as an indicator
  of the strength of the regression, depending on whether the response
  variable has a strong presence in the smallest APC.  If the response
  shows up only weakly, it follows that the predictors have stronger
  additive associations among each other than with the response.
\end{itemize}
%%AB: I can't see the following.
%   Thirdly, APCs contains additive canonical
%   correlation analysis (CCA) as a special case when the number of
%   variables $p=2$.  Therefore, estimating APCs in this case is
%   equivalent to finding transformations that result in maximal
%   correlation between two variables (see Section~\ref{kernelCCA}).

Estimation of APCs and their transforms $\phi_j(X_j)$ from finite data
requires some form of regularization.  There exist two broad classes
of regularization in nonparametric function estimation, namely,
subspace regularization and shrinkage regularization.  Subspace
regularization restricts the function estimates $\hat{\phi}_j$ to
finite-dimensional function spaces on $X_j$.  Shrinkage regularization
produces function estimates by adding a penalty to the goodness-of-fit
measure in order to impose the spatial structure of $X_j$ on
$\hat{\phi}_j$.  Commonly used are generalized ridge-type quadratic
penalties (also called the ``kernelizing approach'') and lasso-type
$\ell_1$-penalties.  The original APC proposal in
\cite{DonnellBuja1994} uses subspace regularization for estimation
without providing asymptotic theory for it.  In the present article we
propose APCs based on shrinkage/kernelizing regularization and provide
some asymptotic consistency theory.

It should be pointed out that introducing a shrinkage/kernelizing
approach into a multivariate method is not a mechanical exercise.  It
is not a priori clear where and how the penalties should be inserted
into a criterion of multivariate analysis, which in the case of PCA is
variance subject to a constraint.  The situation differs from
regression where there is no conceptual difficulty in adding a
regularization penalty to a goodness-of-fit measure.  In a PCA-like
method such as APC analysis, however, it is not clear whether
penalties should be added to, or subtracted from, the variance, or
somehow added to the constraint, or both.  An interesting and related
situation occurred in functional multivariate analysis where the same
author (B.~Silverman) co-authored two different approaches to the same
PCA regularization problem \citep{RS1991, S1994}, differing in
where and how the penalty is inserted.  Our approach, if transposed to
functional multivariate analysis, agrees with neither of them.  One
reason for our third way is that neither of the approaches in
\cite{RS1991} or \cite{S1994} generalize to the low end of the PCA
spectrum.  In contrast, the regularized criterion proposed in this
article can be applied to the high and the low end of the spectrum,
and hence to the discovery of low dimension as well as low
co-dimension.  Our more specific interest is in the latter.

An immediate benefit of injecting penalty regularization into
multivariate analysis stems from recent methodological innovations in
kernelizing.  These include the possibility of using
infinite-dimensional function spaces, the interpretation of
regularization kernels as positive definite similarity measures, and
the kernel algebra with the freedom of modeling it engenders.  Two
decades ago, when \cite{DonnellBuja1994} was written, it would have
been harder to make the case for penalty regularization.

% The contribution of the shrinkage regularization approach over the
% subspace regularization approach adopted in existing
% transformational multivariate analysis literature is that the
% penalty term as a measure of complexity intrinsically calibrates the
% associated functions space as a hierarchy of progressively more
% complex functions.  \textcolor{ZurichRed}{This is preferable as the
% underlying true transforms is assumed to be simple.  Moreover, the
% use of kernels method allows the search space for APCs to be
% infinite dimensional, yet computational feasibility is retained.
% This provides a distinguishing flexibility for estimating
% constraints in comparison to the subspace restriction approach.  In
% addition, feature extraction is replaced with the notion of
% similarity as conferred by the choice of kernel.}

In what follows we first describe the mathematical structure of APCs
which sets up the optimization problem that will be studied
subsequently (Section~\ref{APCmath}).  Section \ref{sec:penalty-parameters}
discusses methodology for selecting penalty parameters in kernelized finite-sample 
APCs.  In Section~\ref{simdata} we demonstrate the application of kernelized APC 
analysis in simulated and real data sets.  Section~\ref{RKHS} poses the APC problem 
in the framework of reproducing kernel Hilbert spaces.  Section~\ref{Consistency} 
establishes the existence of  population APCs and the consistency of 
kernelized finite-sample APCs.  Section~\ref{compute} presents the power 
method for computing APCs, together with its supporting theoretical framework.  
We conclude with a discussion in Section~\ref{conclusion}.  
Proofs of results stated in Section~\ref{Consistency} are in Appendix~\ref{consproof},
whereas proofs related to the power method of Section~\ref{compute}
are in Appendix~\ref{powerproof}.  Appendix~\ref{sec:power-details} contains 
implementation details for the power method, while Appendix~\ref{linAlg} contains an 
alternative linear algebra method for computing APCs.

%%%%%%%%%%%%%%%%%%%%%%%%%%%%%%%%%%%%%%%%%%%%%%%%%%%

%!TEX root = Paper.tex

% Mathematical Formulation of APCs
\section{Detailed Statement  of the APC Problem}
\label{APCmath}

\subsection{Transformations and Their Interpretations}

The raw material for LPCs are real-valued random variables $X_1,
\ldots , X_p$ with a joint distribution and finite second moments.
LPCs are then defined as linear combinations $\sum a_j X_j$ with
extremal or stationary variance subject to a constraint $\sum a_j^2
= 1$.

For APCs one replaces linear combinations with additive combinations
of the form $\sum \phi_j(X_j)$ where $\phi_j(X_j)$ are real-valued
functions defined on arbitrary types of random observations $X_j$.
That is, $X_1, \ldots, X_p$ can be random observations with values in
arbitrary measurable spaces $\mX_1, \ldots , \mX_p$, each of which can
be continuous or discrete, temporal or spatial, high- or
low-dimensional.  The only assumption at this point is that these
random observations have a joint distribution
$P_{1:p}(dx_1, \ldots, dx_p)$ on $\mX_1 \times \cdots \times \mX_p$.
Random variables are obtained by forming real-valued functions
$\phi_j(X_j)$ of the arbitrarily-valued $X_j$.  The functions $\phi_j$
are often interpreted as ``quantifications'' or ``scorings'' or
``scalings'' of the underlying spaces $\mX_j$.  If $X_j$ is already
real-valued, then $\phi_j(X_j)$ is interpreted as a variable
transformation.

We will only consider functions $\phi_j(X_j)$ that have finite
variance and belong to some closed subspace of square-integrable
functions with regard to their marginal distributions $P_j(dx_j)$:
\[
\phi_j\in H_j \subset L^2(\mX_j, dP_j) := \{\phi_j: \Var(\phi_j(X_j)) <\infty\}.
\]
In what follows we will write $\phi_j$ or $\phi_j(X_j)$
interchangeably.  The role of the coefficient vector $\ba =
(a_1,\ldots,a_p)^T$ in LPCs is taken on by a vector of
transformations:
\[
\bPhi:= (\phi_1, \ldots, \phi_p)\in\bH:= H_1\times H_2\times\cdots\times H_p.
\]
Similarly, the role of the linear combination $\sum a_j X_j$ in LPCs
is taken on by an additive function $\sum \phi_j(X_j)$.  APCs
contain LPCs as a special case when all $X_j$ are real-valued and
$H_j = \{ \phi_j :\, \phi_j(X_j) = a_j X_j,\, a_j \in \reals \}$.

\subsection{Population Inner Products and Irrelevant Constants}

The inner product will initially be population covariance.  For this
reason all functions are defined only modulo sets of measure zero and
also modulo constants.  Constants are a particular nuisance in
additive functions $\sum\phi_j$ because they are non-identifiable
across the transformations: for example, $\tiphi_k=\phi_k+c$,
$\tiphi_l=\phi_l-c$ for some $k \neq l$ and $\tiphi_j=\phi_j$ else
result in the same additive function, $\sum\tiphi_j = \sum\phi_j$.

To deal with unidentifiable constants in additive functions, we think
of $L^2(\mX_j, dP_j)$ as consisting of equivalence classes of
functions where two functions are equivalent if they differ by a
function that is constant almost everywhere (rather than requiring the
two functions to be equal almost everywhere).  We can then endow $H_j$
with the covariance as inner product and the variance as the squared norm:
\[
\langle\phi_j, \psi_j\rangle_{P_j} :=
\Cov(\phi_j, \psi_j)
\qquad\text{and}\qquad
\|\phi_j\|_{P_j}^2 := \Var(\phi_j).
\]
The natural inner product and squared norm on the product space $\bH$
are therefore
\[
\langle\bPhi, \bPsi\rangle_{P_{1:p}}:=\sum_{j=1}^p\langle\phi_j,
\psi_j\rangle_{P_j}\qquad\text{and}\qquad\|\bPhi\|_{P_{1:p}}^2:=\sum_{j=1}^p\|\phi_j\|_{P_j}^2.
\]
To avoid unidentifiable constants, \cite{DonnellBuja1994} take $H_j$
to be a closed subspace of centered transformations,
$L^2_c(\mX_j, dP_j) := \{\phi_j: E\phi_j = 0, \Var(\phi_j) <\infty\}$.
This solution to the non-identifiability problem of constants is
essentially equivalent to ours, but ours has the benefit that it does
not raise unnecessary questions when estimates $\hat{\phi}_j$ of the
transformations $\phi_j$ cannot be centered at the population (which
is not known) and hence strictly speaking cannot be in $H_j$ as defined
in~\cite{DonnellBuja1994}.  Our framework says that differences by
constants are irrelevant and should be ignored.

\subsection{Criterion and Constraint --- A Null Comparison Principle}
\label{sec:null-comparison}

In the introduction we stated that the optimization criterion for APCs
is variance, $\Var(\sum\phi_j)$, but we did not specify the
normalization constraint other than writing it as
$\sum \| \phi_j \|^2$ $= 1$.  From the previous subsections it is
clear that we will choose the constraint norms to be
$\| \phi_j \|_{P_j}^2 = \Var(\phi_j)$.  This was the choice made by
\cite{DonnellBuja1994}, their justification being that it generalizes
LPCs: for $H_j \!=\! \{ a_j X_j \!: a_j \in \reals \}$ we have
$\Var(\phi_j) = a_j^2$ for real-valued standardized $X_j$,
$\Var(X_j) = 1$, hence the constraint becomes $\sum a_j^2 = 1$.

To kernelize APCs, we will need a deeper justification for the
constraint.  Even for LPCs, however, we may ask: what is it that makes
$\sum a_j^2$ ``natural'' as a quadratic constraint form?  The answer
we propose is in the following:

\smallskip

\begin{itemize}
\item[] {\bf\em Null comparison principle for multivariate
    analysis:} {\em The quadratic form to be used for the constraint is
    the optimization criterion evaluated under the null assumption of
    vanishing correlations.}
\end{itemize}

This principle ties the constraint form in a unique way to the
criterion: There is no longer a choice of the constraint because it
derives directly from the criterion.  We arrive at a powerful and
principled way of devising generalizations of multivariate methods, as
will be exemplified when we kernelize APCs.  For the familiar cases
the principle works out as follows: For LPCs the null assumption is
\[
\Cov(X_j,X_k) = 0 \;\;\;\;\;\;\; \forall j \neq k. % Why are ~ ignored?
\]
The evaluation of the criterion, assuming also standardized real
variables $X_j$, results in the familiar form
\[
\Var(\sum a_j X_j) = \sum \Var(a_j X_j) = \sum a_j^2 .
\]
For APCs the null assumption is
\[
\Cov(\phi_j(X_j), \phi_k(X_k))=0  \;\;\;\;\;\;\; \forall \phi_j \in H_j, \; \phi_k \in H_k, \; j \neq k,
\]
that is, pairwise independence of the $X_j$.  The evaluation of the
criterion results in
\[
\Var(\sum \phi_j(X_j)) = \sum \Var(\phi_j(X_j)), 
\]
agreeing with our choice for the constraint form.  In retrospect this
justifies what we defined above to be the ``natural'' squared norm on
$\bH = H_1 \times \cdots \times H_p$.  Associated with it is the
``natural'' inner product also defined above:
\[
\langle \bPhi, \bPsi \rangle_{P_{1:p}} = \sum \Cov(\phi_j(X_j),\psi_j(X_j)) .
\]
Its utility is in defining hierarchies of APCs whereby the constrained
optimization problem is solved repeatedly under the additional
constraint of orthogonality to all previously obtained solutions.

\subsection{Kernelized APCs}
\label{kernelAPC}
The APC estimation procedure of \cite{DonnellBuja1994} 
can be characterized as using the above APC
framework where one replaces the population distribution $P_{1:p}$
with the empirical distribution $\Phat_{1:p} = \frac{1}{n} \sum_i \delta_{\x_i}$
of the data $\{\x_i: 1\leq i\leq n\}$, and $H_j$ by
finite-dimensional Hilbert spaces $\Hhat_j$ whose dimension is low
compared to $n$ (they may be spanned by dummy variables for discrete
$X_j$ or by low-degree monomials or spline basis functions with few
knots for quantitative $X_j$).  The finite dimensionality of the
spaces $\Hhat_j$ achieves the regularization necessary for estimation.
The empirically solvable optimization problem is therefore
\begin{equation} \label{sampleVers} 
  \min_{\phi_j \in \Hhat_j} \Varhat(\sum_{j=1}^p \phi_j) 
  \qquad\text{subject to}\qquad
  \sum_{j=1}^p \Varhat(\phi_j) = 1 ,
\end{equation}
where $\Varhat$ is the empirical variance obtained from data.  This
reduces to a generalized finite-dimensional eigenvalue/eigenvector
problem \citep{DonnellBuja1994}.

By contrast, we will consider here APC estimation based on
kernelizing, where regularization is achieved through additive
quadratic penalties $J_j(\phi_j)$ that induce Reproducing Kernel
Hilbert Spaces (RKHS) $\Hsp_j$ which become the natural choice 
for $H_j$.  While estimation is again based on the empirical
distribution $\Phat_{1:p}$, a regularized population version based on
the actual distribution $P_{1:p}$ exists also and is useful for
bias-variance calculations.  For simplicity of notation we continue
the discussion for the population case.  The kernelized optimization
criterion we choose is the penalized variance:
\begin{equation} \label{eq:criterion-penalized}
  \Var(\sum_{j=1}^p \phi_j) + \sum_{j=1}^p J_j(\phi_j) .
\end{equation}
This is a natural choice for minimization because it forces the
transformations $\phi_j$ not only to generate small variance but also
regularity in the sense of the penalties.  For concreteness the reader
may use as a leading example a cubic spline penalty
$J_j(\phi_j) = \alpha_j \int \phi_j''(x_j) dx_j$ for a quantitative
variable $X_j$ (where we absorbed the tuning constant $\alpha_j$ in
$J_j$), but the reader versed in kernelizing will recognize the
generality of modeling offered by positive definite quadratic forms
that generate RKHSs.

The question is next what the natural constraint should be.  Informed
by the null comparison principle of Section \ref{sec:null-comparison},
we will not naively carry $\sum \Var(\phi_j) = 1$ over to the
kernelized problem.  Instead we evaluate the criterion
\eqref{eq:criterion-penalized} under the assumption of absent
correlations between the transformations $\phi_j$ in the spaces $\Hsp_j$:
\begin{equation} \label{eq:constraint-penalized}
  \sum_{j=1}^p \Var(\phi_j) + \sum_{j=1}^p J_j(\phi_j) ~=~ 1.    
\end{equation}
As it turns out, this formulation produces meaningful results both for
minimization {\em and} maximization, hence both for estimating
implicit additive equations for discovery of structure of {\em low
  co-dimension}, and for estimating additive dimension reduction for
discovery of structure of {\em low dimension}.  In the present article
we pursue the former goal.  ---~An equivalent unconstrained problem
is in terms of the associated Rayleigh quotient:
\begin{equation} \label{eq:Rayleigh}
 \minmaxst_{\phi_1, \ldots,\phi_p} ~ \frac{  \Var( \sum \phi_j ) + \sum J_j(\phi_j) }
       {  \sum \Var(\phi_j)   + \sum J_j(\phi_j) }
\end{equation}

On data we will replace the population quantities in equations
\eqref{eq:criterion-penalized} and \eqref{eq:constraint-penalized}
with their sample counterparts.  As is usual, the penalties will be
expressed in terms of quadratic forms of certain kernel matrices.

\subsection{Alternative Approaches to Kernelized APCs}
\label{sec:alterntatives}

A brief historic digression is useful to indicate the conceptual
problem solved by the null comparison principle: As mentioned in the
introduction, in the related but different field of functional
multivariate analysis, Silverman co-authored two different approaches
to the same PCA regularization problem where largest principal
components are sought for dimension reduction.  These can be
transposed to the APC problems as follows:
\begin{eqnarray}
  \label{eq:FPCA-Rice-Silverman-91}
  \max_{\phi_1, \ldots,\phi_p} ~ \Var( \sum \phi_j ) - \sum J_j(\phi_j)
  \qquad
  \text{subject to}
  \qquad
  \sum \Var(\phi_j ) = 1 ,
  \\
  \label{eq:FPCA-Silverman-94}
  \max_{\phi_1, \ldots,\phi_p} ~ \Var( \sum \phi_j ) 
  \qquad
  \text{subject to}
  \qquad
  \sum \Var(\phi_j ) + \sum J_j(\phi_j) = 1 ,
\end{eqnarray}
where \eqref{eq:FPCA-Rice-Silverman-91} is due to \cite{RS1991} and
\eqref{eq:FPCA-Silverman-94} is due to \cite{S1994}.  The first
approach \eqref{eq:FPCA-Rice-Silverman-91} substracts the penalty from
the criterion, which does what it should do for regularized variance
{\em maximization}.  It is unsatisfactory for reasons of mathematical
aesthetics: a difference of two quadratic forms can result in negative
values, which may not be a practical problem but ``does not seem
right'.'  The second approach \eqref{eq:FPCA-Silverman-94} solves this
isssue by adding a penalty to the constraint rather than substracting
it from the criterion, which again does what it should do for variance
maximization.  Both approaches can be criticized for resulting in
non-sense when the goal is regularized variance {\em minimization}.
Here the first approach \eqref{eq:FPCA-Rice-Silverman-91} is more
satisfying because it is immediately clear how to modify it to work
for regularized variance minimization:
\begin{equation*}
  \min_{\phi_1, \ldots,\phi_p} ~ \Var( \sum \phi_j ) + \sum J_j(\phi_j)
  \qquad
  \text{subject to}
  \qquad
  \sum \Var(\phi_j) = 1 ,
\end{equation*}
whereas for the approach \eqref{eq:FPCA-Silverman-94} it is not clear
how it could be modified to work in this case.  Subtracting the
penalty from the constraint variance,
$\sum \Var(\phi_j) - \sum J_j(\phi_j) = 1$, is clearly not going to
work.

Escewing these problems, we propose
\begin{equation} \label{eq:APC}
  \min_{\phi_1, \ldots,\phi_p} ~ \Var( \sum \phi_j ) + \sum J_j(\phi_j)
  \qquad
  \text{subject to}
  \qquad
  \sum \Var(\phi_j)  + \sum J_j(\phi_j) = 1 .
\end{equation}
The merits of this proposal are that (1)~it has no aesthetic issues,
(2)~it works for both ends of the variance spectrum, and (3)~it
derives from a more fundamental principle rather than mathematical
experimentation.

\subsection{The Kernelizing View of APCs}
\label{sec:kernel-view}

The major benefit of formulating APCs in the kernelizing framework is
the flexibility of embedding the information contained in data objects
in $p$ different $n \times n$ kernel matrices as opposed to an
$n \times p$ feature matrix.  Kernel matrices have an interpretation
as similarity measures between pairs of data objects.  It is therefore
possible to directly design similarity matrices (instead of features)
for non-Euclidean data for use as kernels.  Thus topological
information between data objects captured by kernels can be used to
directly estimate APC transforms of non-quantitative data.  Just as
one extracts multiple features from data objects, one similarly
extracts multiple similarity matrices to capture different topological
information in data objects.  APC then helps us find associations
between these kernels in terms of ``implicit'' redundancies.
Following the discussion at the end of
Section~\ref{kernelAPC}, on data the APC variance is
evaluated on the sum of transforms, and the penalties are obtained
from the constructed kernel matrices.

\subsection{Relation of APCs to Other Kernelized Multivariate Methods}

The focus on the lower end of the spectrum seems to have found little
attention in the literature, but the criterion we use can be related
to existing proposals even if their focus is on the upper end of the
spectrum.

A special situation with precedent in the literature occurs for $p=2$,
in which case the kernelized APC problem \eqref{eq:APC} reduces to the
kernelized canonical correlation analysis (CCA) problem discussed by
\cite{Fukumizu2007}.  To see the equivalence, one may start with the
simplified Rayleigh problem \eqref{eq:Rayleigh},
\begin{equation} \label{eq:APC2}
  \minmaxst_{\phi_1, \phi_2} ~ \frac{ \Var( \phi_1 + \phi_2 ) + J_1(\phi_1) + J_2(\phi_2) }
       { \Var(\phi_1) + \Var(\phi_2) + J_1(\phi_1) + J_2(\phi_2) }.
\end{equation}
It can be shown that stationary solutions satisfy 
\begin{equation} \label{eq:CCA-EqualNorms}
  \Var(\phi_1) + J_1(\phi_1) ~=~ \Var(\phi_2) + J_2(\phi_2) ,
\end{equation}
and it follows that the problem \eqref{eq:APC2} is equivalent to
\[\minmaxst_{\phi_1, \phi_2} ~
  \frac{ \Cov( \phi_1, \phi_2 ) }
       { \left(\Var(\phi_1) + J_1(\phi_1)\right)^{1/2} \left(\Var(\phi_2) + J_2(\phi_2)\right)^{1/2} },
\]
where the normalization \eqref{eq:CCA-EqualNorms} can be enforced
without loss of generality.  This is recognized as a penalized form of CCA.  It has been
rediscovered several times over, in the machine learning literature by
\cite{Bach2003}, and earlier in the context of functional multivariate
analysis by \cite{Leurgans1993}.
%, and probably first, using a simple Ridge penalty by \cite{Pito197x}.

Interesting is the work of \cite{Bach2003} which generalizes CCA to the
case $p>2$ but shows no interest in the results of such an analysis
other than this becoming the building block in a method for
independent components analysis (ICA), where the input variables $X_j$
are projections of multivariate data onto frames of orthogonal unit
vectors.  \cite{Bach2003} correctly build up a finite-sample version
of what amounts to APCs for $p>2$ without a guiding principle other
than the appearance of it being a ``natural generalization''.  A
population version and associated consistency theory is missing as
their focus is on ICA and associated computational problems.

Finally it would be natural to discuss a relationship between kernel
APCs and kernel principal components (KPCA, \cite{Scholkopf1998},
\cite{Scholkopf2002}).  However, we do not see a natural connection at
this point.

% KPCA can be interpreted as maximizing
% variance under a kernel norm constraint: $\Var(\phi(X)) = \max_\phi$
% subject to $J(\phi(X)) \le 1$.  This can be turned into an equivalent
% Rayleigh quotient in several ways, the most natural being
% \[
%   \frac{\Var(\phi(X))} {J(\phi(X))} ~=~ \max_\phi .
% \]
% In this interpretation KPCA appears to be a nonlinear 1-block version
% that is analogous to \cite{S1994}'s functional PCA as
% described in Section~\ref{sec:alterntatives},
% \eqref{eq:FPCA-Silverman-94}.

%%%%%%%%%%%%%%%%%%%%%%%%%%%%%%%%%%%%%%%%%%%%%%%%%%%

%!TEX root = Paper.tex

% Methodologies for Penalty Parameters
\section{Methodologies for Choosing Penalty Parameters}
\label{sec:penalty-parameters}

Any kernel calls implicitly for a multiplicative penalty parameter
that controls the amount of regularization to balance bias and
variance against each other.  Methods that use multiple kernels will
have as many penalty parameters as kernels.  Choosing these parameters
in a given problem requires some principles for systematically
selecting the values for these parameters.  Such principles have been
discussed at least as long as there have existed additive models
\citep{HT1990}, and APCs pose new problems only in so far as they use
Rayleigh quotients as their optimization criteria rather than residual
sums of squares or other regression loss functions as their
minimization criteria.  An initial division of principles for penalty
parameter selection is into a priori choice and data-driven choice.
% {\red (Kernel-based estimation of APCs for fixed penalties is developed in
% Section~\ref{compute}.)}

\subsection{A Priori Choice of Penalty Parameters} 

In order to make an informed a priori choice of a penalty parameter it
must be translated into an interpretable form.  The most common such
form is in terms of a notion of ``degrees of freedom'' which can be
heuristically rendered as ``equivalent number of observations invested
in estimating a transformation.''  To define degrees of freedom for
kernelizing one makes use of the fact that for a fixed penalty
parameter a fit $\hphi(x)$ produced by a kernel on regressor-response
data $(x_i,y_i)$ of size $n$ is a linear operation
$\y = (y_i)_{i=1...n} \mapsto \hphi$, $\reals^n \rightarrow$~$H$,
hence the evaluation map $\y \mapsto \hby = (\hphi(x_i))_{i=1..n}$,
$\reals^n \rightarrow \reals^n$ can be represented by a matrix
operation $\bS \y = \hby$, where the $n \times n$ ``smoother matrix''
$\bS$ is symmetric and non-negative definite, and all its eigenvalues
are $\le 1$.  The matrix $\bS$ depends on the penalty parameter
$\alpha$, $\bS = \bS_{\alpha}$, and serves as the basis for defining
notions of degrees of freedom.  Several definitions exist, three of
which are as follows~\citep{BHT1989}:
\begin{itemize}
\item $df = \tr(\bS^2)$: This derives from the total variance in
  $\hby$, which under homoskedasticity is
  $\sum_i \Var(\hy_i) = \tr(\bS \bS') \sigma^2$.  Variance of fitted
  values is a measure of how much response variation has been invested
  in the fits.
\item $df = \tr(2 \bS - \bS^2)$: This derives from the total residual
  variance in $\br = \y - \hby$ under a homoskedasticity assumption:
  $\sum_i \Var(r_i) = \tr(\I - \bS - \bS' + \bS \bS') \sigma^2$.
  Variance of residuals, when substracted from $n \sigma^2$, is a
  measure of how much of the error variance has been lost to the
  fitted values.
\item $df = \tr(\bS)$: This derives from a Bayesian interpretation of
  kernelizing under a natural Bayes prior that results in
  $\bS \sigma^2$ as the posterior covariance matrix of $\hby$.  A
  frequentist derivation is obtained by generalizing Mallows' $C_p$
  statistic which corrects the residual sum of squares with a term
  $2 (df) \hsigma^2$ to make it unbiased for the predictive MSE; the
  appropriate generalization for smoothers is $df = \tr(\bS)$.
\end{itemize}
Among these, the third is the most popular version.  If $\bS$ is a
projection, all three definitions result in the same value, which is
the projection dimension, but for kernels whose $\bS$ contains
eigenvalues strictly between 0 and 1 the three definitions are
measures of different concepts.  For general kernels the calculation
of degrees of freedom for a ladder of penalty parameter values
$\alpha$ may result in considerable computational expense, which is
compounded by the fact that in practice for a prescribed degree of
freedom several values of $\alpha$ need to be tried in a bisection
search.  Yet the translation of $\alpha$ to a degree of freedom may be
the most natural device for deciding a priori on an approximate value
of the penalty parameter.  Selecting degrees of freedom separately for
each transformation $\phi_j$ is of course a heuristic for APCs, as it
is for additive regression models, because what matters effectively is
the total degrees of freedom in the additive function
$\sum_{j=1}^p \hphi_j$.  Summing up the individual degrees of freedom
of $\hphi_j$ is only an approximation to the degrees of freedom
of~$\sum_{j=1}^p \hphi_j$~\citep{BHT1989}.

In practice one often decides on identical degrees of freedom $df$ for
all transforms $\hphi_j$ and chooses the sum $p \cdot df$ to be a
fraction of $n$, such as $p \cdot df = n/10$.

\subsection{Data-driven Choice of Penalty Parameters} 
\label{sec:CV}

% --- linear formula for leave-one-out
% --- weaving CV into power iterations, no convergence guarantee
%     CV even variable-wise inside power iterations
The most popular data-driven method is based on cross-validation.  A
first question is what the criterion should be that is being
cross-validated.  We use as the relevant criterion the empirical,
unpenalized sample eigenvalue:
\[
\hat{\lambda}_1 := \frac{\Varhat(\sum\hat{\phi}_j)}{\sum\Varhat(\hat{\phi}_j)}.
\]
It is an estimate of $\lambda_1 := \min_{\phi_1, \ldots, \phi_p}\Var(\sum\phi_j)/\sum\Var(\phi_j)$ which, when small ($\ll \! 1$),
suggests the existence of concurvity in the data.  Of course, the
criterion that is actually being minimized is the penalized sample
eigenvalue:
\[
\frac{\Varhat(\sum\hat{\phi}_j) + \sum J_j(\hat{\phi}_j)}{\sum\Varhat(\hat{\phi}_j) + \sum J_j(\hat{\phi}_j)}.
\]
This we treat as a surrogate quantity that is not of substantive
interest.  (The distinction between quantity of interest and surrogate
quantity is familiar from supervised classification where interest
focuses on misclassification rates but minimization is carried out on
surrogate loss functions such as logistic or exponential loss;
accordingly it is misclassification rates that are used in
cross-validation.)

To choose the penalty parameters in the simplest possible way, one
often makes them identical for all variables and then searches their
common value $\alpha$ on a grid, minimizing the $k$-fold cross-validation
criterion
\[
\text{CV}(\alpha) =
 \frac{1}{k}\sum_{i=1}^k\
 \frac{\Varhat(\sum_{j=1}^p\hat{\phi}_{\{i\}j})}
 {\sum_{j=1}^p\Varhat(\hat{\phi}_{\{i\}j})}.
\]
The variances $\Varhat$ are evaluated on the holdout
sets while the transforms $\hat{\phi}_{\{i\}j}$ are estimated from the
training sets.

Here, however, attention must be paid to the question of what ``equal
value of the penalty parameters'' means.  The issue is that the
meaning of a penalty parameter $\alpha$ is very much scale dependent.
For example, a standard Gaussian kernel
$k(x, x') = \exp\{-\half (x-x')^2\}$ is very different when a variable
measured in miles is converted to a variable measured in feet.  One
approach to equalizing the effect of scale on the penalties and
kernels is to standardize all variables involved.  Another approach is
to calibrate all penalty parameters to produce the same degrees of
freedom.

%%%%%%%%%%%%%%%%%%%%%%%%%%%%%%%%%%%%%%%%%%%%%%%%%%%

%!TEX root = Paper.tex

% Simulation and Real Datasets
\section{Simulation and Application to Real Data}
\label{simdata}

In this section we demonstrate the use of APC in identifying
additive degeneracy in data.  We first apply APC analysis to a university 
webpages data and an air pollution data.  We then evaluate 
the finite-sample performance of APC on a simulated data for which 
the optimal transformations are known.  
% We also present a heuristic
% cross-validation procedure for tuning the parameters $\alpha_j$'s
% associated with the penalty term in the kernelized sample APC problem
% \eqref{kernSampleAPC}.  

\subsection{Real Data}

% University Webpages
\subsubsection{University Webpages}

In this section, we demonstrate the kernelizing use of kernel APC in
estimating additive implicit equation.  Instead of extracting features
and estimating APC transforms of the features, we start with
similarity measures between data points and use them as elements in
the kernel matrices directly.  In this case, one can view APC
transformation as an embedding of the original data points in the
Euclidean space into an RKHS, and the estimated additive constraints
represent an implicit equation in the embedded space.  Such
flexibility is only possible with kernel APC (which only requires
kernel matrices as input) but not with the original subspace
restriction APC as considered in \cite{DonnellBuja1994} (which
requires a specification of basis functions as input).

\begin{table}[htp]
\centering
\begin{tabular}{| ll | ll | ll | ll |}
  \hline
 group 1 &  & group 2 &  & group 3 &  & group 4 &  \\ 
  \hline
activ & student & address &  & acm & languag & advanc & receiv \\ 
area & teach & contact &  & algorithm & method & assist & scienc \\ 
book & work & cours &  & analysi & model & associ & softwar \\ 
build & year & depart &  & applic & network & center & state \\ 
california &  & email &  & architectur & parallel & colleg & technolog \\ 
chair &  & fall &  & base & problem & degre & univers \\ 
class &  & fax &  & comput & process & director &  \\ 
current &  & hall &  & confer & program & educ &  \\ 
faculti &  & home &  & data & public & electr &  \\ 
graduat &  & inform &  & design & research & engin &  \\ 
group &  & link &  & develop & select & institut &  \\ 
includ &  & list &  & distribut & structur & intellig &  \\ 
interest &  & mail &  & gener & studi & laboratori &  \\ 
introduct &  & offic &  & high & system & mathemat &  \\ 
paper &  & page &  & ieee & techniqu & member &  \\ 
project &  & phone &  & implement & theori & number &  \\ 
recent &  & updat &  & investig & time & profession &  \\ 
special &  & web &  & journal & tool & professor &  \\ 
   \hline
\end{tabular}
\caption{Keywords in group 1 to group 4.}
\label{keyword}
\end{table}

We consider the university webpages data set from the ``World Wide
Knowledge Base" project at Carnegie Mellon University.  This data set
was preprocessed by \cite{Cardoso} and previously studied in
\cite{Guo2011} and \cite{Tan2015}.  It includes webpages from computer
science departments at Cornell, University of Texas, University of
Washington, and University of Wisconsin. In this analysis, we consider
only the faculty webpages --- resulting in $n$ = 374 faculty webpages
and $d$ = 3901 keywords that appear on these webpages.

We now discuss how we constructed for this data set four similarity
matrices to be used as kernels.  First we reduced the number of
keywords from 3901 to 100 by thresholding the $\log$-entropy.  Let
$f_{ij}$ be the number of times the $j^{th}$ keyword appears in the
$i^{th}$ webpage.  The log-entropy of the $j^{th}$ keyword is defined
as $1+\sum_{i=1}^ng_{ij}\log(g_{ij})/\log(n)$, where
$g_{ij} = f_{ij}/\sum_{i=1}^nf_{ij}$.  We then selected the 100
keywords with the largest $\log$-entropy values and constructed an
$n\times 100$ matrix $H$ whose $(i, j)$ element is $\log(1+f_{ij})$.
We further standardized each column to have zero mean and unit
variance.  In order to obtain four different kernels, we applied the
$k$-means algorithm to cluster the keywords in $H$ into $k=p=4$
groups.  Each group of keywords is represented as an $n\times m_i$
submatrix $H_i$, and we obtained the final $n\times n$ kernel matrix
$K_i = H_iH_i^T / \tr(H_iH_i^T)$, where the normalization
is to account for different group sizes.  Table~\ref{keyword} shows the
keywords in each group.  Roughly, group 1 contains keywords related to
teaching and current projects, group 2 contains keywords related to
contact information, group 3 contains keywords related to research
area, and group 4 contains keywords related to biography of a faculty.

\begin{figure}[htp]
\centering
\includegraphics[angle=0,width=4.5in]{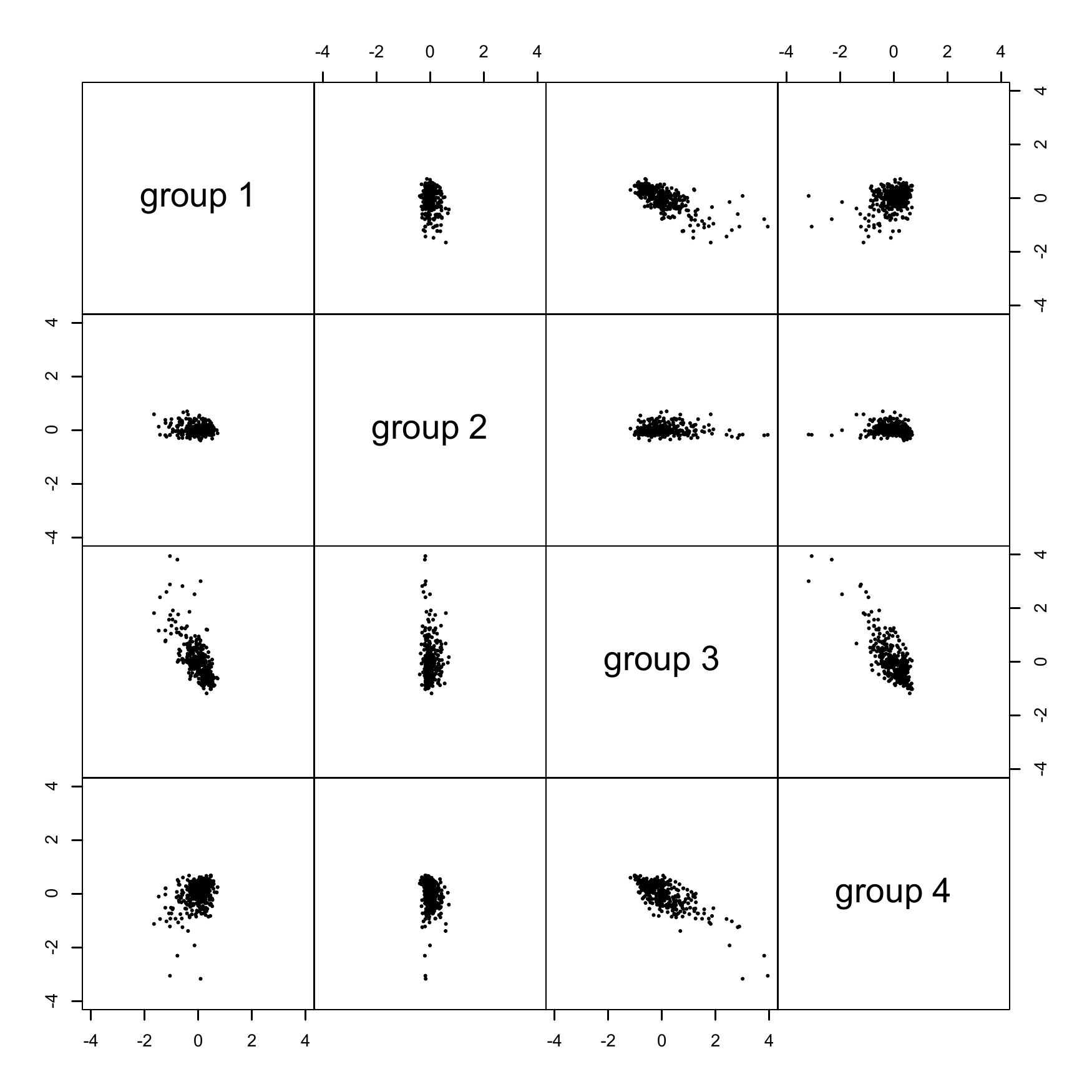}
\caption{Pairwise scatterplot of the smallest kernelized sample APC 
for the university webpages data.  The sample eigenvalue for the 
estimated APC is 0.0910.}
\label{Webpage1}
\end{figure}

Upon an APC analysis using the kernel matrices $K_1, \ldots, K_4$
constructed above, we obtain the transformations
$\hat{\phi}_1, \ldots, \hat{\phi}_4$.  Figure~\ref{Webpage1} shows the pairwise
scatterplot of the transformed data points.  We see that $\hat{\phi}_3$ and
$\hat{\phi}_4$ have strong negative correlation, $\hat{\phi}_1$ and 
$\hat{\phi}_3$ have moderate negative correlation, and $\hat{\phi}_1$ and 
$\hat{\phi}_4$ have weak positive correlation.  For ease of interpretation, 
the transformed data points are centered to zero mean and normalized 
to $\sum_{j=1}^4\Varhat\hat{\phi}_j = 1$.  The normalization permits
us to interpret $\Varhat\hat{\phi}_j$ as relative importance of
the $j^{th}$ group in the estimated APCs.  The variance of each group in 
the smallest APC are: 0.1662 (group 1), 0.0305 (group 2),
0.5562 (group 3), 0.2471 (group 4).  Ignoring group 2 which has the
smallest weight, we see that, roughly, this means that
\[
\hat{\phi}_1 + \hat{\phi}_3+\hat{\phi}_4 \approx 0, \qquad\text{or, equivalently,}\qquad\hat{\phi}_4\approx -\hat{\phi}_1-\hat{\phi}_3.
\]
If we plot $\hat{\phi}_4$ against $\hat{\phi}_1+\hat{\phi}_3$, we get the scatterplot in
Figure~\ref{Webpage2}.

\begin{figure}[htp]
\centering
\includegraphics[angle=0,width=2.6in]{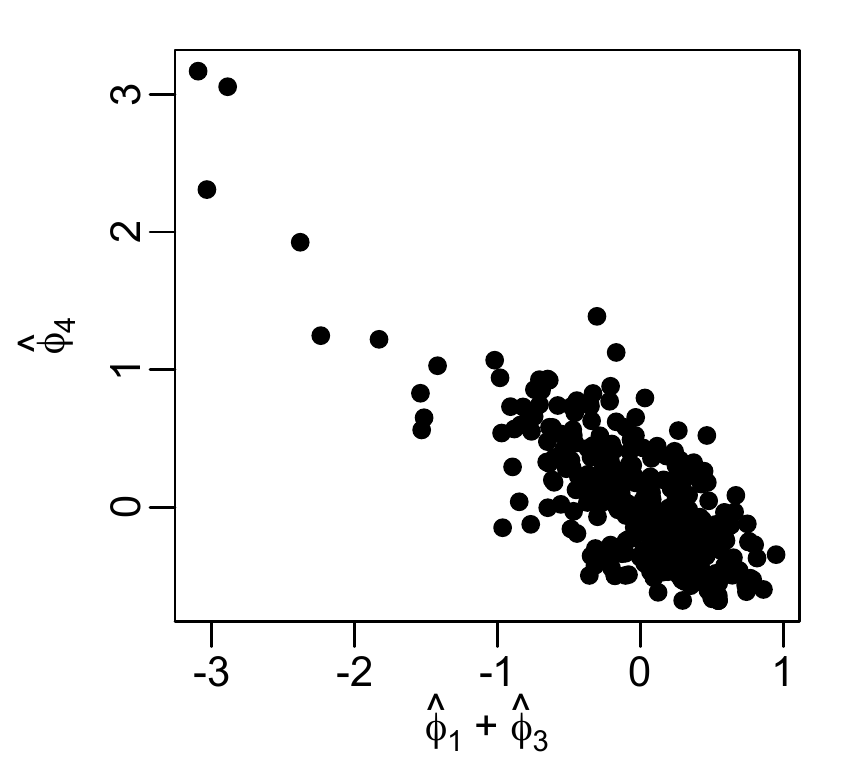}
\caption{Plot of $\hat{\phi}_4$ against $\hat{\phi}_1 + \hat{\phi}_3$ in the smallest kernelized 
sample APC for the university webpages data.}
\label{Webpage2}
\end{figure}

Recalling that the kernels were constructed to reflect similarity in
terms of keywords related to (1)~teaching and current projects,
(2)~contact information, (3)~research area, and (4)~biography, we
obtain two results: contact information is related to neither of
teaching and projects nor research, whereas biography is well
predicted by teaching, projects and research.  This is of course
highly plausible for faculty webpages and academic biographies.

% Air Pollution
\subsubsection{Air Pollution}
In this section, we apply APC analysis to a data set consisting of
quantitative variables, where the purpose is to find nonlinear
transformations that reflect redundancies among the variables.  To
this end we analyze the $NO_2$ data that is publicly available on the
StatLib data sets archive
\textit{http://lib.stat.cmu.edu/datasets/NO2.dat}.  It contains a
subsample of 500 observations from a data set collected by the
Norwegian Public Roads Administration for studying the dependence of
air pollution at a road on traffic volume and meteorological
condition.  The response variable consists of hourly values of the
log-concentration of $NO_2$ particles, measured at Alnabru in Oslo,
Norway, between October 2001 and August 2003.  Because the posted data
is only a subset of the original data, the middle chunk of
observations is missing.  To avoid nonsensical transformations of the
variables, only the second half of the data (roughly November 2002 to
May 2003) is used in the APC analysis.  Given below are descriptions
for individual variables in the data:

\begin{center}
\begin{tabular}{ll}
\texttt{NO2}: & hourly values of the logarithm of the concentration of $NO_2$ particles; \\  
\texttt{Cars}: & logarithm of the number of cars per hour; \\   
\texttt{TempAbove}: & temperature $2$ meter above ground (degree C); \\   
\texttt{Wind}: & wind speed (meters/second); \\   
\texttt{TempDiff}: & temperature difference between $25$ and $2$ meters above ground (degree C); \\   
\texttt{WindDir}: & wind direction (degrees between 0 and 360); \\
\texttt{HourOfDay}: & hour of day; \\   
\texttt{DayNumber}: & day number from October 1, 2001. \\   
\end{tabular}
\end{center}

Figure~\ref{NO21} shows the corresponding transformations for
individual variables in the smallest APC.  The variables \texttt{Cars}
and \texttt{HourOfDay} are the primary variables involved (with
respective variance $0.51$ and $0.304$).  Holding other variables
fixed, the estimated constraint says that
$\hat{\phi}_2(\texttt{Cars})+\hat{\phi}_7(\texttt{HourOfDay}) \approx 0$.  Since
$\hat{\phi}_2$ is monotone decreasing and the transformation of
\texttt{HourOfDay} peaks around 4pm, we infer that the largest number
of cars on the roads is found in the late afternoon, which is
consistent with the daily experience of commuters.

\begin{figure}[htp]
\centering
\includegraphics[angle=0,width=6in]{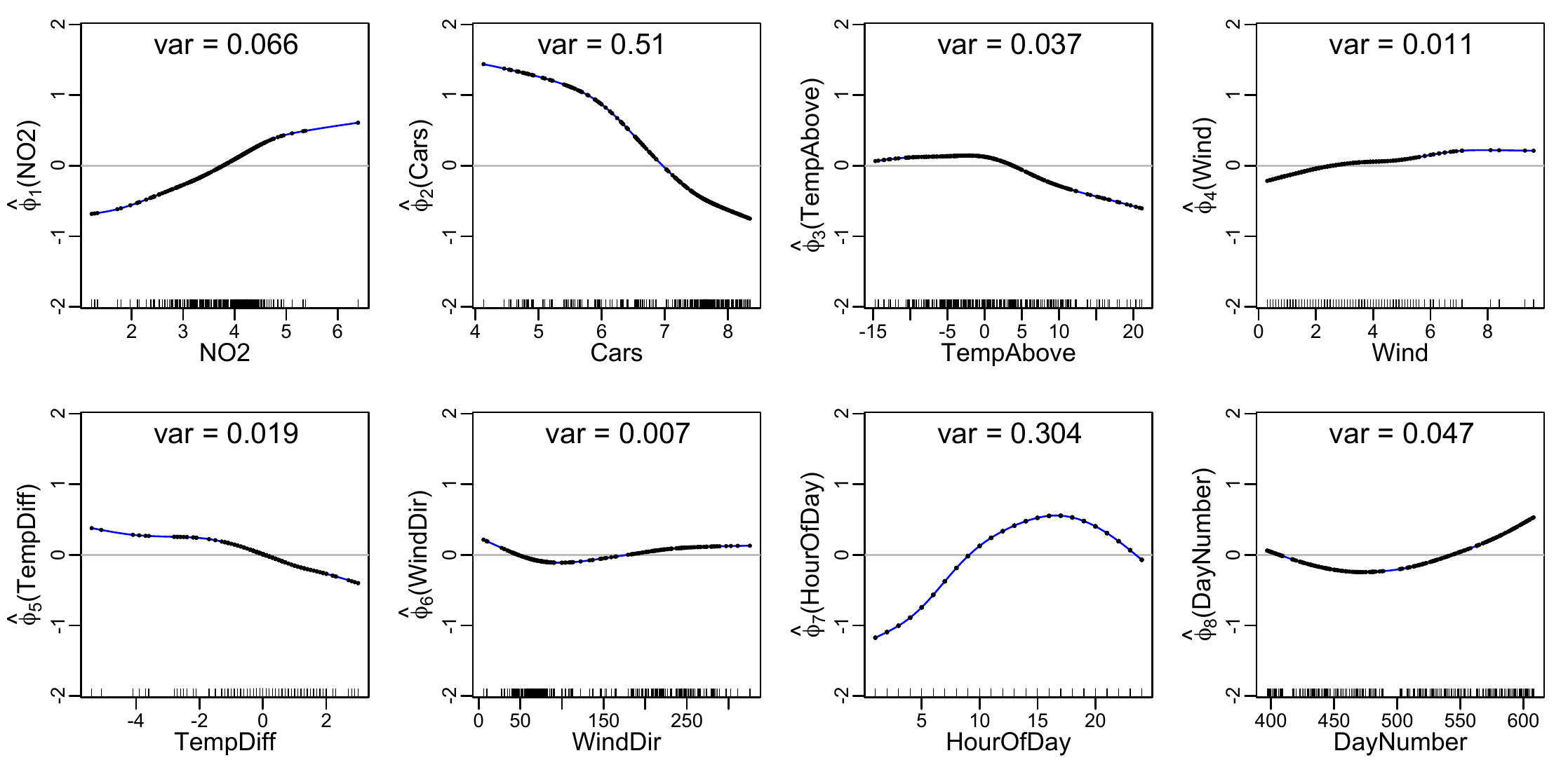}
\caption{The smallest kernelized sample APC for the $NO_2$ data obtained from power algorithm using Sobolev kernel for individual variables.  The sample eigenvalue for the estimated APC is 0.0621.  The black bars at the bottom of each plot indicate the location of data points for that variable.}
\label{NO21}
\end{figure}

\begin{figure}[htp]
\centering
\includegraphics[angle=0,width=6in]{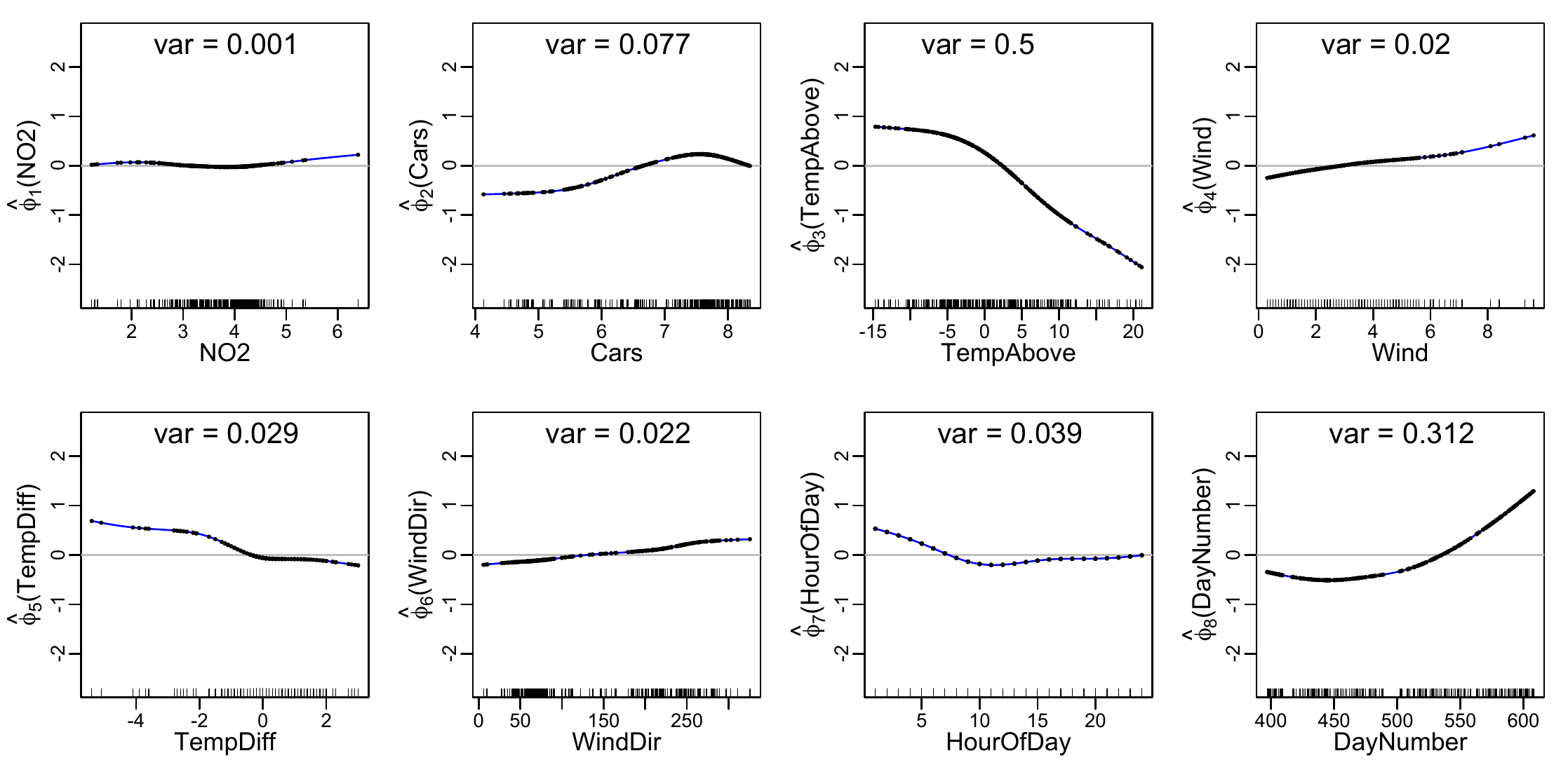}
\caption{The second smallest kernelized sample APC for the $NO_2$ data obtained from power algorithm using Sobolev kernel for individual variables.  The sample eigenvalue for the estimated APC is 0.0827.  The black bars at the bottom of each plot indicate the location of data points for that variable.}
\label{NO22}
\end{figure}

\begin{figure}[htp]
\centering
\includegraphics[angle=0,width=6in]{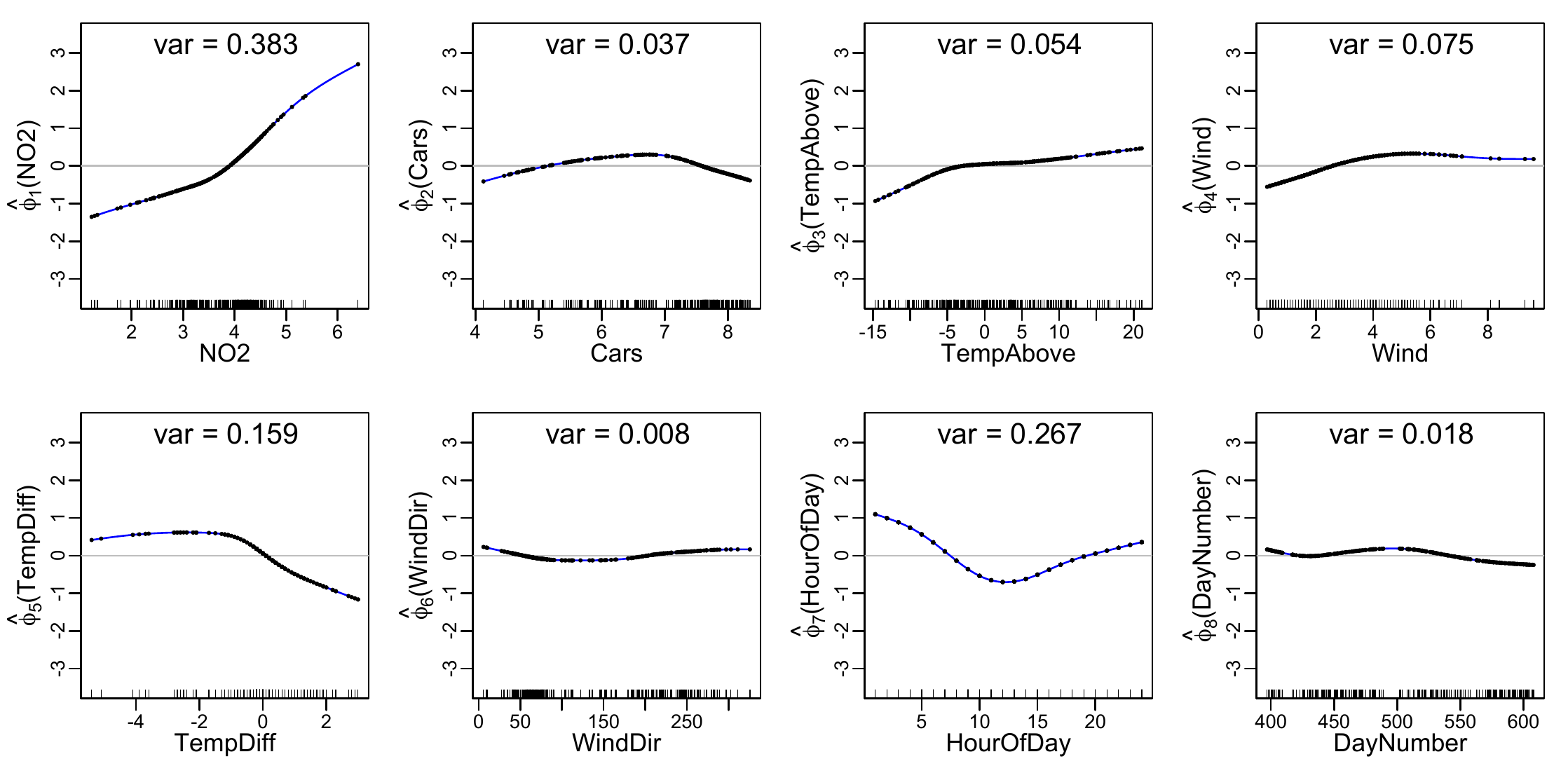}
\caption{The third smallest kernelized sample APC for the $NO_2$ data obtained from power algorithm using Sobolev kernel for individual variables.  The sample eigenvalue for the estimated APC is 0.189.  The black bars at the bottom of each plot indicate the location of data points for that variable.}
\label{NO23}
\end{figure}

In the second smallest APC in Figure~\ref{NO22}, the variables
\texttt{TempAbove} and \texttt{DayNumber} play the dominant roles, and
we have
$\hat{\phi}_3(\texttt{TempAbove})+\hat{\phi}_8(\texttt{DayNumber})\approx 0$.
Since $\hat{\phi}_3$ is monotone decreasing it follows that
\texttt{TempAbove} decreases and then increases with respect to
\texttt{DayNumber}.  This relationship makes sense since our data
spans over November 2002 to May 2003, hence carries the transition
from fall to summer.

The response variable of interest in the original study, \texttt{NO2},
does not appear until the third smallest APC, which is given in
Figure~\ref{NO23}.  We have
$\hat{\phi}_1(\texttt{NO2}) + \hat{\phi}_5(\texttt{TempDiff}) +
\hat{\phi}_7(\texttt{HourOfDay}) \approx 0$.
From the shape of $\hat{\phi}_7$ we see that the highest \texttt{NO2} occurs
during lunch time, which makes sense as this is the time of greatest
sun exposure.  Note that surprisingly there is no interpretable
association with \texttt{Cars} as its transformation has little
variance and is not monotone (more cars should create more
\texttt{NO2}).  However, the strong association between \texttt{Cars}
and \texttt{HourOfDay} in the smallest APC creates an approximate
non-identifiability of association of \texttt{Cars} and
\texttt{HourOfDay} vis-\`{a}-vis other variables such as \texttt{NO2},
which may partly explain the absence of association between
\texttt{Cars} and \texttt{NO2}.

APC analysis suggests some interesting relationship among the
variables in the data.  It also suggests that had an additive model
been fitted to the data with \texttt{NO2} as the response and all
other variables as the predictors, the estimated transforms for
individual variables will not be interpretable due to the presence of
concurvity (as represented by the smallest and second smallest APC) in
the data.

% Simulated Data
\subsection{Simulated Data}

We construct a simulated example consisting of four univariate random
variables $X_1, \ldots, X_4$ with known APC transformations
$\phi_1(X_1), \ldots$, $\phi_4(X_4)$.  This will be achieved by
constructing them in such a way that the joint distribution of these
transformations will be multivariate normal and highly collinear.  The
reason for this construction is that the extremal APCs of multivariate
normal distributions are linear.  (They also have APCs with
non-extremal eigenvalues consisting of systems of Hermite polynomials;
see \cite{DonnellBuja1994}.)  This implies that if transformations
$\phi_j(X_j)$ exist that result in a jointly multivariate normal
distribution, they will constitute an APC.

A simple procedure for simulating a situation with well-defined APCs
is to first construct a multivariate normal distribution and transform
its variables with the inverses of the desired transformations.  APC
estimation is then supposed to find approximations of these
transformations from data simulated in this manner.  

We start by constructing a multivariate normal distribution by using
two independent variables $W_1, W_2 \sim \Norm(0,1)$ to generate the
underlying collinearity and four independent variables
$Z_1, Z_2, Z_3, Z_4 \sim \Norm(0,0.1^2)$ to generate noise:
\begin{equation*}
  Y_1 = W_1 + Z_1      , \;\;\;\;\;\;
  Y_2 = W_2 + Z_2      , \;\;\;\;\;\;
  Y_3 = W_1 + W_2 + Z_3, \;\;\;\;\;\;
  Y_4 = Z_4 .
\end{equation*}
Thus the joint distribution features a collinearity of co-dimension 1
in the first three variables, and the fourth variable is independent
of the rest.  The correlation matrix of these four variables has a
smallest eigenvalue of $0.007441113...$, which will be the smallest
population APC eigenvalue.  The associated eigenvector is
$(1/2,1/2,1/\sqrt{2},0)$, which indicates that the fourth transform
will be zero, whereas the first three transforms will have variances
$1/4$, $1/4$ and $1/2$, respectively.
% R code for the covariance matrix of the first 3 variables:
% sig <- 0.1
% cv <- cbind(c(1+sig^2,0,1,0), c(0,1+sig^2,1,0),c(1,1,2+sig^2,0),c(0,0,0,1))
% cr <- cv / sqrt(outer(c(1+sig^2,1+sig^2,2+sig^2,1), c(1+sig^2,1+sig^2,2+sig^2,1)))
% eigen(cr)
The ``observed'' variables are constructed as marginal transformations
$X_j = f_j(Y_j)$ using the following choices:
\begin{equation*}
X_1 = \exp(Y_1)            , \;\;\;\;\;
X_2 = -Y_2^{1/3}            , \;\;\;\;\;
X_3 = \exp(Y_3)/(1+\exp(Y_3) , \;\;\;\;\;
X_4 = Y_4  ,
\end{equation*}
hence the APC transformations are
\begin{equation*}
  \phi_1^*(x) \sim \log(x)       , \;\;\;\;\;\;
  \phi_2^*(x) \sim -x^3          , \;\;\;\;\;\;
  \phi_3^*(x) \sim \log(x/(1-x)) , \;\;\;\;\;\;
  \phi_4^*(x) = 0 .
\end{equation*}
As noted above the last transformation vanishes, and the other
transformations are given only up to irrelevant additive constants as
well as scales to achieve $\Var(\phi_1)$ $=$ $\Var(\phi_2)$
$= 1/4$ and $\Var(\phi_3) = 1/2$.  

Figure~\ref{Sim1} shows the kernelized sample APC for this data set
($n=250$), with a common penalty parameter chosen by 5-fold
cross-validation.  As discussed at the end of Section~\ref{sec:CV}, we
standardized all variables to have unit variance before applying a
standard Gaussian kernel $k(x, x') = \exp\{-\half (x-x')^2\}$ for each
variable $X_j$.  The solid red line denotes the true transform
$\phi_j^*$, while the dashed blue line denotes estimated transform
$\hat{\phi}_j$.  We see that for each variable, the two lines are
almost indistinguishable, though estimation accuracy worsens near the
boundaries and on regions with few data points (location of data
points are indicated by the black bars at the bottom of each plot).
The transformed data points are centered to zero mean and normalized
to $\sum_{j=1}^4\Varhat\hat{\phi}_j = 1$, so that
$\Varhat\hat{\phi}_j$ indicates the relative importance of
$\hat{\phi}_j$ in the estimated APCs.  In fact, we see that
$\Varhat\hat{\phi}_j$ is close to $\Var Y_j/(\sum_{i=1}^4\Var Y_i)$ in
the data generating steps.

\begin{figure}[htp]
\centering
\includegraphics[angle=0,width=5.5in]{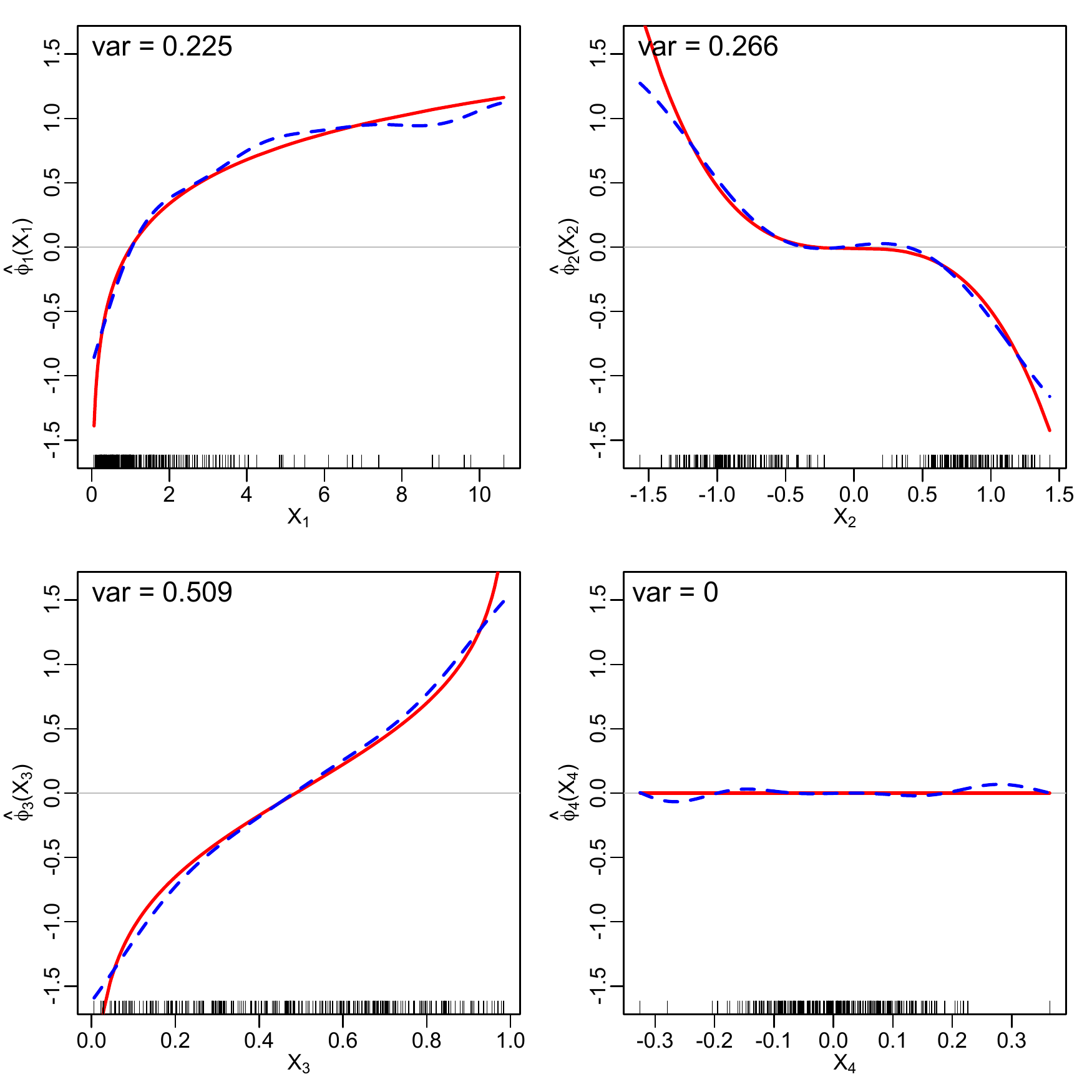}
\caption{Plot of transformations in population APC (\protect\includegraphics[height=0.22em]{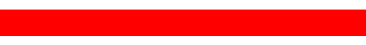}) and transformations in kernelized sample APC (\protect\includegraphics[height=0.22em]{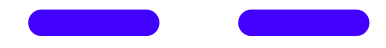}).  The black bars at the bottom of each plot indicate the location of data points for that variable.  The kernelized sample APC is obtained from power algorithm using gaussian kernel for individual variables, with penalty parameters chosen by 5-fold cross-validation.  The sample eigenvalue for the estimated APC is 0.0035.}
\label{Sim1}
\end{figure}

Next, we study the effects that different values of penalty parameters 
have on the resulting kernelized sample APCs, and we examine the 
performance of the cross-validation procedure described in 
Section~\ref{sec:CV}.  We consider a range of sample sizes 
$n \in \{20, 50, 100, 250, 500\}$ and penalty parameters 
$\alpha \in\{1.5^{-29}, 1.5^{-28}, \ldots, 1.5^5\}$.  
For each combination of $n$ and $\alpha$, we compute the true 
estimation error $\Var[\hat{\phi}_j(X_j) - \phi_j^*(X_j)]$, by evaluating
$\hat{\phi}_j$ and $\phi_j^*$ on $n=10,000$ samples generated from the
true distribution, for $1\leq j\leq 4$.  Figure~\ref{Sim2} displays
the average estimation error curves computed from 100 simulated data sets.  
The crosses denote the locations where $\alpha$ achieves the minimum 
estimation error (which can be different from variable 
to variable), one for each sample size, while the circles represent the
average value of the tuning parameters selected by the proposed
cross-validation procedure (which is set to be identical for all variables) and
its average estimation error.  Figure~\ref{Sim2} suggests that the proposed
cross-validation procedure performs reasonably well when $n$ is 
sufficiently large --- although it tends to overestimate the optimal tuning
parameters, the resulting estimation error is generally close to the
best achievable value.

% \textcolor{ZurichRed}{Therefore, we apply the cross-validation
%   procedure to the real data as well?}

\begin{figure}[htp]
\centering
\includegraphics[angle=0,width=5.8in]{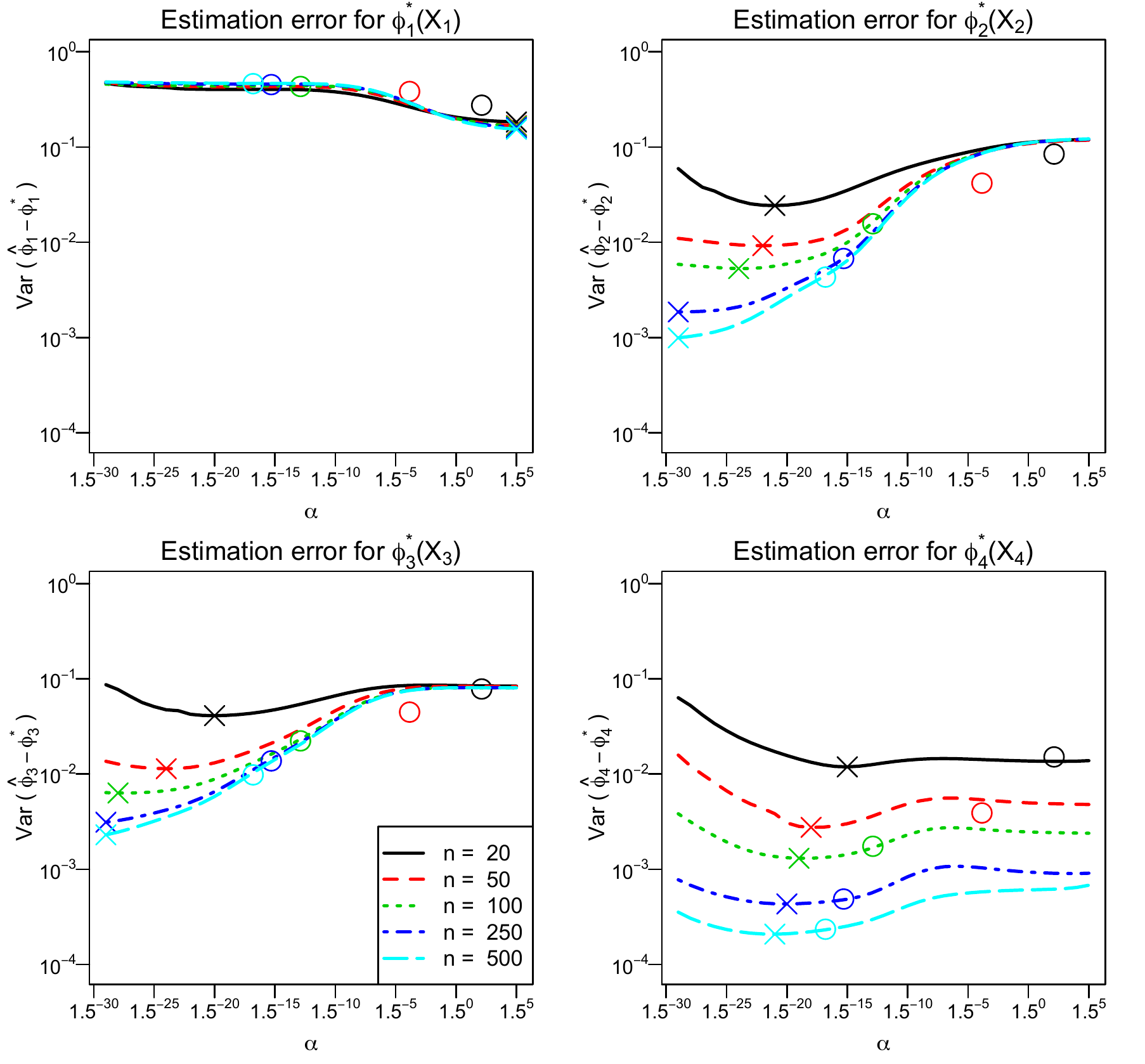}
\caption{Variance of the difference between the true function $\phi_j^*$ and its estimate $\hat{\phi}_j$ obtained from power algorithm using gaussian kernel, for $1\leq j\leq 4$. Crosses: points with lowest estimation error; circles: average penalty parameter values chosen by cross-validation and the average estimation error, the average being taken over 100 simulated data sets.}
\label{Sim2}
\end{figure}

%%%%%%%%%%%%%%%%%%%%%%%%%%%%%%%%%%%%%%%%%%%%%%%%%%%

%!TEX root = Paper.tex

\newpage
\section{Regularization of APCs via Reproducing Kernel Hilbert Spaces}
\label{RKHS} 

In this section, we formalize the statement of the APC problem in 
reproducing kernel Hilbert spaces.  

\subsection{Fundamentals of Reproducing Kernel Hilbert Spaces}
Let $\mX$ be a nonempty set, and let $\Hsp$ be a Hilbert space of
functions $f: \mX\longrightarrow\R$, endowed with the inner product
$\langle\cdot, \cdot\rangle_\Hsp$.  The space $\Hsp$ is a
\emph{reproducing kernel Hilbert space} (RKHS) if all evaluation functionals
(the maps $\delta_x: f\longrightarrow f(x)$, where $x\in\mX$) are
bounded.  Equivalently, $\Hsp$ is an RKHS if there exists a symmetric
function $\kxx$ that satisfies (a) $\forall x\in\mX$,
$k_x = k(x, \cdot)\in\Hsp$, (b) the reproducing property:
$\forall x\in\mX, \forall f\in\Hsp$,
$\langle f, k_x\rangle_\Hsp = f(x)$.  We call such a $k$ the
\emph{reproducing kernel} of $\Hsp$.  There is a one-to-one
correspondence between an RKHS $\Hsp$ and its reproducing kernel $k$.
Thus, specifying $k$ is equivalent to specifying $\Hsp$, and we will
write $\langle\cdot, \cdot\rangle_k$ for
$\langle\cdot, \cdot\rangle_\Hsp$.  Also, $\|k_x\|_k^2 = k(x, x)$.

A distinctive characteristic of an RKHS is that each element $f$ is a
function whose values at any $x\in\mathcal{X}$ is well-defined,
whereas an element of a Hilbert space (i.e. $L^2(\mX, dP)$) is usually
an equivalence class of functions that equal almost everywhere.  Such
a characteristic of RKHS is crucial as it allows us to evaluate the
function $f$ at each point $x\in\mathcal{X}$.

\subsection{APCs in Reproducing Kernel Hilbert Spaces}
We are now ready to define the search space for kernelized APCs.  
Let $X_1, \ldots, X_p$ be random observations taking values in
arbitrary measurable spaces $\mX_1, \ldots , \mX_p$.  Let $P_j(dx_j)$
be the marginal probability measure of $X_j$, and let
$P_{1:p}(dx_1, \ldots, dx_p)$ be the joint probability measure of
$X_1, \ldots, X_p$.  We define the search space for kernelized APCs as
$\bHsp = \Hsp_1\times\cdots\times\Hsp_p$, where the individual
transformation $\phi_j\in\Hsp_j$, and $\Hsp_j$ is an RKHS with
reproducing kernel $k_j: \mX_j\times\mX_j\longrightarrow\R$.  We
suppose that $\Hsp_j$ bears the decomposition
\begin{equation}
\Hsp_j = \Hsp_j^0\oplus\Hsp_j^1, \label{decomp}
\end{equation}
where $\Hsp_j^0$ is a finite-dimensional linear subspace of $\Hsp_j$
with basis $\{q_{1j}, \ldots, q_{m_jj}\}$,
$m_j=\Dim(\Hsp_j^0)<\infty$, and $\Hsp_j^1$ is the orthogonal
complement of $\Hsp_j^0$.  We will refer to $\Hsp_j^0$ as the null
space of $\Hsp_j$, as it contains the set of functions in $\Hsp_j$
that we do not wish to ``penalize" (to be explained further).  With
the decomposition \eqref{decomp}, the reproducing kernel $k_j$ can
also be uniquely decomposed as $k_j = k_j^0 + k_j^1$, where
$k_j^0(x, \cdot) = \bP_j^0k_j(x, \cdot)$,
$k_j^1(x, \cdot) = \bP_j^1k_j(x, \cdot)$, and $\bP_j^0$ and $\bP_j^1$
denotes the orthogonal projection onto $\Hsp_j^0$ and $\Hsp_j^1$,
respectively.  Furthermore, the inner product
$\langle f, g\rangle_{k_j}$ on $\Hsp_j$ can be decomposed as
\begin{equation}
\label{decompinnerprod}
\langle f, g\rangle_{k_j} = \langle f^0, g^0\rangle_{k_j^0} + \langle f^1, g^1\rangle_{k_j^1}, \qquad\forall f, g\in\Hsp_j, 
\end{equation}
where $f = f^0 + f^1$, $g = g^0+g^1$, with $f^0, g^0\in\Hsp_j^0, f^1, g^1\in\Hsp_j^1$, and the decomposition is again unique.

We define the penalty functional $J_j$ as a squared semi-norm on $\Hsp_j$, by letting
\begin{equation}
J_j(f) = \|f^1\|_{k_j^1}^2 = \|\bP_j^1f\|_{k_j^1}^2, \qquad\forall f\in\Hsp_j.
\end{equation}
With slight abuse of notation, we will write
$J_j(f) = \|f\|_{k_j^1}^2$ hereafter.  By definition of $J_j(f)$, we see that $f_j^0 =
\bP_j^0f_j$ does not play any role in the regularization of APCs.  Hence, 
we can also write $\Hsp_j^0 = \{f\in\Hsp_j: J_j(f) = 0\}$, the set of functions in $\Hsp_j$ that are not
``penalized''.  

The kernelized APC problem can be stated in the RKHS framework as follows:
\begin{equation}
\min_{\bPhi\in\bHsp}\Var(\sum_{j=1}^p\phi_j) + \sum_{j=1}^p\alpha_j\|\phi_j\|_{k_j^1}^2 \\
\quad\text{subject to}\quad\sum_{j=1}^p\Var(\phi_j) + \sum_{j=1}^p\alpha_j\|\phi_j\|_{k_j^1}^2 = 1,  \label{statePenSample}
\end{equation}
where $\alpha_j>0$ is the penalty parameter for $X_j$, $j = 1, \ldots, p$.

The second-smallest and other higher order kernelized APCs can be obtained as the solution of \eqref{statePenSample} subject to additional orthogonality constraints
\begin{equation}
\sum_{j=1}^p\Cov(\phi_{\ell, j}, \phi_j)  + \sum_{j=1}^p\alpha_j\langle\phi_{\ell, j}, \phi_j\rangle_{k_j^1} = 0, \label{orthogonal}
\end{equation}
where $\bPhi_\ell = (\phi_{\ell, 1}, \ldots, \phi_{\ell, p})$ encompasses all the previous kernelized APCs.

\begin{Remark}
{\rm
A canonical example of RKHS with a null space is the Sobolev space of order $m$:
\[
\Hsp = \big\{f: [a, b]\longrightarrow\R\ \big|\ f, f^{(1)}, \ldots, f^{(m-1)}\text{ absolutely continuous}, f^{(m)}\in L^2[a, b]\big\},
\]
with $J(f) := \int_a^b (f^{(m)}(t))^2\ dt$.  In this case, the null space $\Hsp^0 = \Span\{1,x, \ldots, x^{m-1}\}$.  
}
\end{Remark}

\begin{Remark}
\label{RemarkConstructRKHS}
{\rm
More generally, one can construct $\Hsp$ with the decomposition 
\eqref{decomp} and \eqref{decompinnerprod}, starting from an RKHS $\Hsp^1$
with reproducing kernel $k^1$.  The idea is to augment
$\Hsp^1$ by a finite-dimensional vector space $\Hsp^0 = \Span\{q_{1}, \ldots, q_{m}\}$,
where $\{q_{1}, \ldots, q_{m}\}$ are linearly
independent and $\Hsp^0\cap\Hsp^1 = \{0\}$, and then define 
\[\Hsp := \big\{f: f = f^0+f^1, f^0\in\Hsp^0, f^1\in\Hsp^1\big\}.\]
Since $\Hsp^0$ is finite-dimensional, any inner product $\langle\cdot, \cdot\rangle$ that induces
a (strict) norm on $\Hsp^0$ turns it into an RKHS.  To see this, first apply the Gram-schmidt procedure on $\{q_1, \ldots, q_m\}$ to obtain an orthonormal basis $\{r_1, \ldots, r_m\}$ of $\Hsp^0$.  Given $f^0=\sum a_ir_i$ and $g^0=\sum b_ir_i$, define $\langle f^0, g^0\rangle_{k^0} = \sum a_ib_i$ and $k^0(x, x') = \sum r_i(x)r_i(x')$.  Then $\Hsp^0$ is an RKHS with inner product $\langle\cdot, \cdot\rangle_{k^0} \equiv \langle\cdot, \cdot\rangle$.  By \cite{Aronszajn1950} (page 352-354), $\Hsp$ as defined is an RKHS with reproducing kernel $k = k^0 + k^1$.
Note, however, that the reproducing structure of $\Hsp^0$ actually does not
come into play in the kernelized APC problem, since functions in this
space go unpenalized. 
}
\end{Remark}

%%%%%%%%%%%%%%%%%%%%%%%%%%%%%%%%%%%%%%%%%%%%%%%%%%%

%!TEX root = Paper.tex

%%%%%%%%%%%%%%%%%%%%%%%%%%%%%%%%%%%%%%%%%%%%%%%%%%%
% Consistency
\section{Consistency}
\label{Consistency}

We define the population APC $\bPhi^* = (\phi_1^*, \ldots, \phi_p^*)$, when it exists, as the solution to
\begin{equation} 
\label{popCons}
  \min_{\bPhi\in\bHsp}\Var(\sum_{j=1}^p\phi_j) 
  \quad\text{subject to}\quad
  \sum_{j=1}^p\Var(\phi_j) = 1.
\end{equation}
On the other hand, the kernelized sample APC $\hPhi = (\hphi_1, \ldots, \hphi_p)$, when it exists, solves
\begin{equation} 
\label{kernSampleAPC}
  \min_{\bPhi\in\bHsp}\Varhat(\sum_{j=1}^p\phi_j) +
  \sum_{j=1}^p\alpha_j\|\phi_j\|_{k_j^1}^2
  \quad\text{subject to}\quad
  \sum_{j=1}^p\Varhat(\phi_j) +
  \sum_{j=1}^p\alpha_j\|\phi_j\|_{k_j^1}^2 = 1.
\end{equation}
The superscript $(n)$ is to emphasize the dependence of $\hPhi$ on the number of observed vectors $n$.

In this section, we establish the existence and uniqueness of
$\bPhi^*$ and also the consistency of $\hPhi$ as an estimator of
$\bPhi^*$ under mild conditions.

% In this section, we will study the APC problem in three versions, namely, the population APCs, the kernelized sample APCs, and the regularized population APCs which bridge the gap between the two.  We will introduce the cross-covariance operators on RKHS's, which enables the reformulation of the APC problems in terms of quadratic form in RKHS's.  We will then see that such formulation allows the proof of convergence of the kernelized sample APCs to population APCs, as $n\longrightarrow\infty$ and $\alpha_j\longrightarrow 0$ for $1\leq j\leq p$. 

\subsection{Preliminaries}
% In this paper, the symbol $\Hsp$ always denotes a \textcolor{ZurichRed}{separable} Hilbert space (usually an RKHS).
Let $\Hsp$, $\Hsp_1$, $\Hsp_2$ be separable Hilbert spaces.  We denote
the norm of a bounded linear operator $\bT$ by
% $\|\bT\|$, where 
$\|\bT\| := \sup_{\|f\|\leq 1}\|\bT f\|$.  The null space and the
range of an operator $\bT: \Hsp_1\longrightarrow\Hsp_2$ are denoted by
$\mN(\bT)$ and $\mR(\bT)$, respectively, where $\mN(\bT) =
\{f\in\Hsp_1: \bT f = 0\}$ and $\mR(\bT) = \{\bT f\in\Hsp_2:
f\in\Hsp_1\}$.  We denote by $\bT^*$ the Hilbert space adjoint of
$\bT$.  We say that $\bT: \Hsp\longrightarrow\Hsp$ is self-adjoint if
$\bT^* = \bT$, and that a bounded linear self-adjoint operator $\bT$ is
positive if $\langle f, \bT f\rangle\geq 0$ for all $f\in\Hsp$.  We
write $\bT\succeq 0$ if $\bT$ is positive, and $\bT_1\succeq\bT_2$ if
$\bT_1-\bT_2$ is positive.  If $\bT$ is positive, we denote by
$\bT^{1/2}$ the unique positive operator $\B$ satisfying $\B^2 = \bT$.
We always denote by $\I$ the identity operator.  A bounded linear
operator $\bT: \Hsp_1\longrightarrow\Hsp_2$ is compact if for every
bounded sequence $\{f_n\}\in\Hsp_1$, $\{\bT f_n\}$ has a convergent
subsequence in $\Hsp_2$.
% A consequence of compactness is that eigendecomposition of $\bT$ is possible. In particular, if $\bT: \Hsp\longrightarrow\Hsp$ is self-adjoint and compact, then there is a complete orthonormal basis system (CONS) $\{g_n\}_{n=1}^\infty$ of $\Hsp$ such that $\bT g_n = \lambda_ng_n$ and $\lambda_n\longrightarrow 0$ as $n\longrightarrow\infty$.  
For other standard functional analysis concepts such as trace class
and Hilbert--Schmidt operators, see \cite{ReedSimon}.  Finally, we
denote by $[p]$ the set~$\{1, \ldots, p\}$.

\subsection{Main Assumptions}\label{mainAssumptions}

We now state the main assumptions used to establish the consistency
result.

\subsubsection{Reproducing Kernel Hilbert Spaces}

% \iffalse Throughout, we will assume that the population APCs, if
% exists, belongs to $\bHsp$, in which case the population APC problem
% can be restated as
% \begin{equation} \label{popCons}
% \min_{\bPhi\in\bHsp}\Var\sum_{j=1}^p\phi_j\qquad\text{subject to }\sum_{j=1}^p\Var\phi_j = 1.
% \end{equation}
% On the other hand, the kernelized sample APC problem is to solve
% \begin{equation} \label{sampleCons}
% \min_{\bPhi\in\bHsp}\Varhat\sum_{j=1}^p\phi_j + \sum_{j=1}^p\epsilon_n\|\phi_j\|_{k_j^1}^2\qquad\text{subject to }\sum_{j=1}^p\Varhat\phi_j + \sum_{j=1}^p\epsilon_n\|\phi_j\|_{k_j^1}^2 = 1,
% \end{equation}
% where, for simplicity, we parametrize $\alpha_j$ in
% \eqref{statePenSample} by $\alpha_j := \alpha_j^{(n)} = \epsilon_n$
% for $1\leq j\leq p$, and $\epsilon_n$ is a sequence of positive
% numbers such that $\lim_{n\rightarrow\infty}\epsilon_n = 0$.  Our goal
% is to show that $\hPhi = (\hphi_1, \ldots, \hphi_p)$, the solution of
% \eqref{sampleCons}, is a consistent estimator of $\bPhi^* = (\phi_1^*,
% \ldots, \phi_p^*)$, the solution of \eqref{popCons}, as
% $n\longrightarrow\infty$.  \fi

% Our first set of assumptions is concerned with the search space of the APC problem.
% In the APC problem, we consider 
We impose the following assumptions on the kernelized APC
search space $\bHsp = \Hsp_1\times\cdots\times\Hsp_p$:
\begin{Assumption}\label{assumeRKHS}
  For $1\leq j\leq p$, let $\Hsp_j = \Hsp_j^0\oplus\Hsp_j^1$ be an
  RKHS with reproducing kernel $k_j = k_j^0+k_j^1$ consisting of
  real-valued functions with domain $\mX_j$, where
\begin{enumerate}[(a)]\itemsep 0.1em
\item $\mX_j$ is a compact metric space;
\item $\Hsp_j^0\subset C(\mX_j)$ with $\dim(\Hsp_j) = m_j<\infty$;
\item $k_j^1(x, x')$ is jointly continuous on $\mX_j\times\mX_j$; 
\item $P_j(dx_j)$ is a Borel probability measure fully
    supported on $\mX_j$.
%\item $\|f\|^2_{k_j^0} := \Var[f(X_j)]>0$ for all $f\in\Hsp_j^0, f\not\equiv 0$.
\end{enumerate}
\end{Assumption}

The following is an immediate consequence of
Assumption~\ref{assumeRKHS} (see, e.g., Lemmas~\ref{incont} and~\ref{inL2}).
\begin{Lemma}
\label{LemmaRKHScont}
Under Assumption~\ref{assumeRKHS}, $E[k_j(X_j, X_j)] < \sup_{x\in\mX_j}k_j(x, x) < \infty$ and $\Hsp_j\subset C(\mX_j)\subset L^2(\mX_j, dP_j)$, for $1\leq j\leq p$. 
\end{Lemma}

\subsubsection{Cross-Covariance Operators}

Following \cite{Fukumizu2007}, we define the mean element
$m_j\in\Hsp_j$ with respect to a random variable $X_j$ as
\[
\langle \phi_j, m_j\rangle_{k_j} = E[\langle \phi_j, k_{X_j}\rangle_{k_j}] = E[\phi_j(X_j)] \qquad\forall\phi_j\in\Hsp_j.
\]
On the other hand, we define the cross-covariance operator of $(X_i,
X_j)$ as a bounded linear operator from $\Hsp_j$ to $\Hsp_i$ given by
\begin{align*}
\langle\phi_i, \C_{ij}\phi_j\rangle_{k_i} &= E[\langle\phi_i, k_{X_i}-m_i\rangle_{k_i}\langle\phi_j, k_{X_j}-m_j\rangle_{k_j}] \\
&= \Cov[\phi_i(X_i), \phi_j(X_j)] \qquad\qquad\qquad\qquad\qquad\forall \phi_i\in\Hsp_i, \phi_j\in\Hsp_j.
\end{align*}
The existence and uniqueness of both $m_j$ and $\C_{ij}$ are proved by
the Riesz Representation Theorem.  It is immediate that
$\C_{ij} = \C_{ji}^*$.  In particular, when $i = j$, the self-adjoint
operator $\C_{jj}$ is called the covariance operator.  It can be
verified that $\C_{jj}$ is positive and trace-class \citep{Baker1973}.
Under Assumption~\ref{assumeRKHS}, $\Var[\phi_j(X_j)] = 0$ if and only
if $\phi_j$ is a constant function.  Hence, the null space is
$\mN(\C_{jj}) = \Hsp_j\cap\reals$.

Let $\{(X_{\ell 1}, \ldots, X_{\ell p}): 1\leq \ell \leq n\}$ be i.i.d.~random
vectors on $\mX_1\times\cdots\times\mX_p$ with joint distribution
$P_{1:p}(dx_1, \ldots, dx_p)$.  The empirical cross-covariance operator $\hatC_{ij}$ is
defined as the cross-covariance operator with respect to the empirical
distribution $\frac{1}{n}\sum_{\ell=1}^n\delta_{X_{\ell
    i}}\delta_{X_{\ell j}}$, in which case
\begin{align*}
\langle\phi_i, \hatC_{ij}\phi_j\rangle_{k_i} &= \frac{1}{n}\sum_{\ell = 1}^n\bigg\langle\phi_i, k_{X_{\ell i}}- \frac{1}{n}\sum_{a=1}^nk_{X_{ai}}\bigg\rangle_{k_i}\bigg\langle\phi_j, k_{X_{\ell j}}- \frac{1}{n}\sum_{b=1}^nk_{X_{bj}}\bigg\rangle_{k_j} \\
&= \widehat{\Cov}(\phi_i, \phi_j), \qquad\qquad\qquad\qquad\qquad\qquad\forall \phi_i\in\Hsp_i, \phi_j\in\Hsp_j.
\end{align*}
Since $\mR(\hatC_{ij})$ and $\mN(\hatC_{ij})^\perp$ are included in
$\Span\{k_{X_{\ell i}}- \frac{1}{n}\sum_{a=1}^nk_{X_{ai}}:
1\leq\ell\leq n\}$ and $\Span\{k_{X_{\ell j}}-
\frac{1}{n}\sum_{b=1}^nk_{X_{bj}}: 1\leq\ell\leq n\}$, respectively,
% it is clear that 
$\hatC_{ij}$ is of finite rank.

It is known \citep[Theorem 1]{Baker1973} that $\C_{ij}$ has a
representation
\begin{equation}
\label{Vij}
\C_{ij} = \C_{ii}^{1/2}\bV_{ij}\C_{jj}^{1/2},
\end{equation}
where $\bV_{ij}: \Hsp_j\longrightarrow\Hsp_i$ is a unique bounded
linear operator with $\|\bV_{ij}\|\leq 1$.  Moreover, $\bV_{ij} = \Q_i\bV_{ij}\Q_j$
where $\Q_i$ is the orthogonal projection of $\Hsp_i$ onto $\overline{\mR(\C_{ii})}$.
This implies that $\mR(\bV_{ij}) \subseteq \overline{\mR(\C_{ii})}$.

% \iffalse
% \textcolor{ZurichRed}{Note that in general an operator may not be
%   invertible.  However, under Assumption~\ref{assumeRKHS}(a)$-$(e),
%   $\mN(\C_{jj}) = \{0\}$, and so the inverse of $\C_{jj}$ is
%   well-defined.  We can thus express $\bV_{ij}$ as
%   $\C_{ii}^{-1/2}\C_{ij}\C_{jj}^{-1/2}$.  We can similarly define
%   $\hV_{ij}$ as the operator satisfying $\hatC_{ij} =
%   (\hatC_{ii}+\epsilon_n\bP_i^1)^{1/2}\hV_{ij}(\hatC_{jj}+\epsilon_n\bP_j^1)^{1/2}$. Should
%   we talk about the existence of kernelized sample solution??}  
% \fi

The following assumption will be used to establish the
existence and uniqueness of $\bPhi^*$ in $\bHsp$.
\begin{Assumption}\label{compact}
  The operator $\bV_{ij}$ is compact for $1\leq i, j\leq p$, $i\neq
  j$.  Moreover, the smallest eigenvalue of the operator $\bV-\I$,
  where $\bV = (\bV_{ij})_{i,j\in [p]}$, $\bV_{jj} = \I$ for any $j\in
  [p]$, is simple.
\end{Assumption}

\begin{Remark}
{\rm
  The following condition used in \citet[Theorem 3]{Fukumizu2007} is
  sufficient for $\bV_{ij}$ to be Hilbert-Schmidt, which in turn
  implies compactness: suppose that $X_i, X_j$ have joint density
  $f_{i, j}(x_i, x_j)$ and marginal densities $f_i(x_i), f_j(x_j)$,
  then $\bV_{ij}$ is Hilbert-Schmidt provided
  \begin{equation}
  \label{HilbertSchmidt}
  \iint\frac{f_{i, j}^2(x_i, x_j)}{f_i(x_i)f_j(x_j)}\ dx_idx_j < \infty.
  \end{equation}
  Condition \eqref{HilbertSchmidt} is also sufficient to guarantee the existence of 
  population APC $\bPhi^{**}$ in $\bH = H_1\times\cdots\times H_p$, where $H_j = L^2(\mX_j, dP_j)$.  
  Under the assumption that $\Hsp_j$ is dense in $L^2(\mX_j, dP_j)$, we have $\bPhi^{**} = \bPhi^*$,
  up to an almost sure constant function.
}
  % \textcolor{ZurichRed}{The existence of solution in
  % \eqref{sampleCons} follows since $\hV_{ij}$ is finite rank and
  % hence compact.}
\end{Remark}

Due to the presence of the spaces $\Hsp_j^0$ of unpenalized
functions,
% In the case where we want to distinguish a space $\Hsp_j^0$ of unpenalized functions, i.e. $\dim(\Hsp_j^0) = m_j>0$, 
we need an additional assumption.  Noting that $\Hsp_j =
\Hsp_j^0\oplus\Hsp_j^1$, we can zoom in on the covariance operator
$\C_{jj}$ and see that it induces the (restricted) covariance
operators $\C_{j0, j0}$ and $\C_{j1, j1}$ on $\Hsp_j^0$ and
$\Hsp_j^1$, and the cross-covariance operator $\C_{j0, j1}:
\Hsp_j^1\longrightarrow\Hsp_j^0$.  Similar as before, by Theorem 1 of
\cite{Baker1973}, $\C_{j0, j1}$ has the representation
\[
\C_{j0, j1} = \C_{j0, j0}^{1/2}\bV_{j0, j1}\C_{j1, j1}^{1/2},
\]
where $\bV_{j0, j1}: \Hsp_j^1\longrightarrow\Hsp_j^0$ is a unique
bounded linear operator with $\|\bV_{j0, j1}\|\leq 1$.  The following
assumption will be used to establish the existence of solutions to a
kernelized population APC problem, which is a key step in the proof
as we will detail in Section \ref{proofOutline}.

\begin{Assumption}\label{invertible}
  For $1\leq j\leq p$, $\C_{j0, j0}\succeq b\I$ and $\|\bV_{j0,
    j1}\|\leq\sqrt{1-a/b}$ for some absolute positive constants $a<b$.
\end{Assumption}

\subsection{Statement of Main Theorem}

For simplicity, we set the penalty parameters $\alpha_j$ in the
kernelized sample APC problem \eqref{kernSampleAPC} at a common level
$\epsilon_n$.  The following result establishes the consistency of
kernelized sample APCs.

% With all the assumptions introduced, we are now ready to state our main theorem.  Let $\hPhi = (\hphi_1, \ldots, \hphi_p)$ and $\bPhi^* = (\phi_1^*, \ldots, \phi_p^*)$ be the solution of \eqref{popCons} and \eqref{sampleCons}, respectively.  Define 
% \[\lambda_1 = \Var\sum\phi_j^*, \qquad \hlambda_1 = \Varhat\sum\hphi_j + \sum\epsilon_n\|\hphi_j\|_{k_j^1}^2.\]

\begin{Theorem}\label{MainThm2}
  Suppose that Assumptions~\ref{assumeRKHS}, \ref{compact} and
  \ref{invertible} hold, and $f_j^*\in\mR(\C_{jj})$, for
  $1\leq j\leq p$.  Let $(\epsilon_n)_{n=1}^\infty$ be a sequence of
  positive numbers such that
\begin{equation}
  \lim_{n\rightarrow\infty}\epsilon_n = 0, \qquad \lim_{n\rightarrow\infty}\frac{n^{-1/3}}{\epsilon_n} = 0. \label{condeps}
\end{equation}
Then, for any $\delta>0$, there exists $N(\delta)$ such that for all
$n\geq N(\delta)$, with probability greater than $1-\delta$, the
kernelized sample APC problem \eqref{kernSampleAPC} has a unique
solution $\hPhi = (\hphi_1, \ldots, \hphi_p)$ which satisfies
\begin{align}
  \sum_{j=1}^p \Var \big[\hphi_j(X_j)-\phi_j^*(X_j) \big] & <\delta, \quad \mbox{and} \label{apcevecconv}\\
  |\hlambda_1-\lambda_1| & < \delta, \label{apcevalconv}
\end{align}
where $\lambda_1 = \Var(\sum\phi_j^*)$ and $\hlambda_1 = \Varhat(\sum\hphi_j)
+\sum\epsilon_n\|\hphi_j\|_{k_j^1}^2$.
\end{Theorem}

Note that in \eqref{apcevecconv}, for each $1\leq j\leq p$, the
variance is taken with respect to the future observation $X_j$ and the
observed data encoded by $(\ \hat\ \ )$ is not integrated out,
implying that $\sum_{j=1}^p \Var[\hphi_j(X_j)-\phi_j^*(X_j)]$ is a
random quantity.  In essence, the convergence in \eqref{apcevecconv}
says that the difference between the kernelized sample APC and the
population APC converges to a constant in probability, while the
convergence in \eqref{apcevalconv} says that the optimal value of the
kernelized sample criterion converges to the optimal value of the
population criterion.

% While in the theorems above we restrict our attention to the smallest APC, the consistency of subsequent kernelized sample APCs is readily verified, under \textcolor{ZurichRed}{additional mild conditions.}

\subsection{Outline of Proof}
\label{proofOutline}

% To connect population APCs and kernelized sample APCs, we need to
% study the asymptotic properties of the kernelized sample APCs.
% This naturally leads us to consider the regularized population APC
% problem, which amounts to solving
In what follows, we lay out the key steps in proving Theorem
\ref{MainThm2}; the detailed arguments are deferred to Appendix
\ref{consproof}.  To study the convergence of $\hPhi$ to $\bPhi^*$, we
show that both converge to the kernelized population APC defined as the
solution to~\eqref{statePenSample}.

%Under the assumption of Theorem \ref{MainThm2}, the population APC
%$\bPhi^*$ solves
%\begin{equation} 
%\label{popCons}
%  \min_{\bPhi\in\bHsp}\Var(\sum_{j=1}^p\phi_j) 
%  \quad\text{subject to}\quad
%  \sum_{j=1}^p\Var(\phi_j) = 1,
%\end{equation}
%while $\hPhi$, when exists, solves
%\begin{equation} 
%\label{sampleCons}
%  \min_{\bPhi\in\bHsp}\Varhat(\sum_{j=1}^p\phi_j) +
%  \sum_{j=1}^p\epsilon_n\|\phi_j\|_{k_j^1}^2
%  \quad\text{subject to}\quad
%  \sum_{j=1}^p\Varhat(\phi_j) +
%  \sum_{j=1}^p\epsilon_n\|\phi_j\|_{k_j^1}^2 = 1,
%\end{equation}
%To study the convergence of $\hPhi$ to $\bPhi^*$, we show that both
%converge to kernelized population APC $\tPhi$ defined as the solution to
%\begin{equation} 
%\label{penpopCons}
%  \min_{\bPhi\in\bHsp}\Var(\sum_{j=1}^p\phi_j) +
%  \sum_{j=1}^p\epsilon_n\|\phi_j\|_{k_j^1}^2
%  \quad\text{subject to}\quad
%  \sum_{j=1}^p\Var(\phi_j) + \sum_{j=1}^p\epsilon_n\|\phi_j\|_{k_j^1}^2
%  = 1.
%\end{equation}

\subsubsection{Quadratic Forms in RKHS Norm}

As a first step in the proof, we rewrite \eqref{popCons},
\eqref{statePenSample} and \eqref{kernSampleAPC} in terms of quadratic forms
with respect to the RKHS norm $\langle\cdot, \cdot\rangle_k$.
% We will consider the general case where $\Hsp_j$ contains a null
% space $\Hsp_j^0$ of functions that are not penalized, for $1\leq
% j\leq p$.
To this end, using the cross-covariance operators introduced in
Section~\ref{mainAssumptions} and setting all $\alpha_j$ to a common
level $\epsilon_n$, we can rewrite
% the population APC, regularized population APC and regularized
% sample APC problems in
\eqref{popCons}, \eqref{statePenSample} and \eqref{kernSampleAPC} as follows:
\begin{subequations}
\begin{align}
  \min_{\bPhi\in\bHsp}\ &\langle\bPhi, \C\bPhi\rangle_k &&\text{subject to}\qquad \langle\bPhi, \text{diag}(\C)\bPhi\rangle_k = 1, \label{re1}\\
  \min_{\bPhi\in\bHsp}\ &\langle\bPhi, (\C+\J^{(n)})\bPhi\rangle_k &&\text{subject to}\qquad \langle\bPhi, (\text{diag}(\C)+\J^{(n)})\bPhi\rangle_k = 1, \label{re2}\\
  \min_{\bPhi\in\bHsp}\ &\langle\bPhi, (\hatC+\J^{(n)})\bPhi\rangle_k &&\text{subject to}\qquad \langle\bPhi,
  (\text{diag}(\hatC)+\J^{(n)})\bPhi\rangle_k = 1, \label{re3}
\end{align}
\end{subequations}
where
\begin{align*}
\C  = (\C_{ij})_{i,j\in [p]},\quad
\hatC  = (\hatC_{ij})_{i,j\in [p]} \quad \text{and} \quad
\J^{(n)} = \text{diag}(\epsilon_n \bP_j^1)_{j\in [p]},
\end{align*}
% \begin{align*}
% \C &= &&\begin{pmatrix} \C_{11} & \C_{12} & \cdots & \C_{1p}\\ \C_{21} & \C_{22} & \cdots & \C_{2p} \\ 
% \vdots & \vdots & \ddots & \vdots \\ \C_{p1} & \C_{p2} & \cdots & \C_{pp} \end{pmatrix}, 
% \quad&\hatC = &&\begin{pmatrix} \hatC_{11} & \hatC_{12} & \cdots & \hatC_{1p}\\ \hatC_{21} & \hatC_{22} & \cdots & \hatC_{2p} \\ 
% \vdots & \vdots & \ddots & \vdots \\ \hatC_{p1} & \hatC_{p2} & \cdots & \hatC_{pp} \end{pmatrix}, \\
% \J^{(n)} &= &&\begin{pmatrix} \epsilon_n\bP_1^1 & \0 & \cdots & \0\\ \0 & \epsilon_n\bP_2^1 & \cdots & \0 \\ 
% \vdots & \vdots & \ddots & \vdots \\ \0 & \0 & \cdots & \epsilon_n\bP_p^1 \end{pmatrix},
% \end{align*}
with $\bP_j^1: \Hsp_j\longrightarrow\Hsp_j^1$ the orthogonal
projection of $\Hsp_j$ onto its closed subspace~$\Hsp_j^1$.

Note that the constraints in \eqref{re1} -- \eqref{re3} are different.
To resolve this issue we introduce the following changes of
variables:
\begin{subequations}
\begin{align}
f_j &= \C_{jj}^{1/2}\phi_j, \label{re7}\\
f_j &= (\C_{jj}+\epsilon_n\bP_j^1)^{1/2}\phi_j, \label{re8}\\
f_j &= (\Chat_{jj}^{(n)}+\epsilon_n\bP_j^1)^{1/2}\phi_j, \label{re9}
\end{align}
\end{subequations}
for $1\leq j\leq p$ in \eqref{re1} -- \eqref{re3}, respectively.
Thus, \eqref{re1} -- \eqref{re3} can be further rewritten as
\begin{subequations}
\begin{align} 
\min_{\f\in\bHsp}\ &\langle\f, \bV\f\rangle_k &&\hspace{-2cm}\text{subject to}\qquad\qquad \langle\f, \f\rangle_k = 1, \label{re4}\\
\min_{\f\in\bHsp}\ &\langle\f, \tV\f\rangle_k &&\hspace{-2cm}\text{subject to}\qquad\qquad \langle\f, \f\rangle_k = 1, \label{re5}\\
\min_{\f\in\bHsp}\ &\langle\f, \hV\f\rangle_k &&\hspace{-2cm}\text{subject to}\qquad\qquad \langle\f, \f\rangle_k = 1, \label{re6}
\end{align}
\end{subequations}
respectively, with
\begin{align}
\bV = (\bV_{ij})_{i,j\in [p]}, \quad
\tV = (\tV_{ij})_{i,j\in [p]}, \quad
\hV = (\hV_{ij})_{i,j\in [p]},
\label{Vmatrix}
% \bV &= &&\begin{pmatrix} \I & \bV_{12} & \cdots & \bV_{1p}\\ \bV_{21} & \I & \cdots & \bV_{2p} \\ 
% \vdots & \vdots & \ddots & \vdots \\ \bV_{p1} & \bV_{p2} & \cdots & \I \end{pmatrix}, \quad \tV = \begin{pmatrix} \I & \tV_{12} & \cdots & \tV_{1p}\\ \tV_{21} & \I & \cdots & \tV_{2p} \label{Vmatrix}\\ 
% \vdots & \vdots & \ddots & \vdots \\ \tV_{p1} & \tV_{p2} & \cdots & \I \end{pmatrix}, \\
% \hV &= &&\begin{pmatrix} \I & \hV_{12} & \cdots & \hV_{1p}\\ \hV_{21} & \I & \cdots & \hV_{2p} \\ 
% \vdots & \vdots & \ddots & \vdots \\ \hV_{p1} & \hV_{p2} & \cdots & \I \end{pmatrix}, \nonumber
\end{align}
% \begin{align}
% \bV &= &&\begin{pmatrix} \I & \bV_{12} & \cdots & \bV_{1p}\\ \bV_{21} & \I & \cdots & \bV_{2p} \\ 
% \vdots & \vdots & \ddots & \vdots \\ \bV_{p1} & \bV_{p2} & \cdots & \I \end{pmatrix}, \quad \tV = \begin{pmatrix} \I & \tV_{12} & \cdots & \tV_{1p}\\ \tV_{21} & \I & \cdots & \tV_{2p} \label{Vmatrix}\\ 
% \vdots & \vdots & \ddots & \vdots \\ \tV_{p1} & \tV_{p2} & \cdots & \I \end{pmatrix}, \\
% \hV &= &&\begin{pmatrix} \I & \hV_{12} & \cdots & \hV_{1p}\\ \hV_{21} & \I & \cdots & \hV_{2p} \\ 
% \vdots & \vdots & \ddots & \vdots \\ \hV_{p1} & \hV_{p2} & \cdots & \I \end{pmatrix}, \nonumber
% \end{align}
where $\bV_{jj} = \tV_{jj} = \hV_{jj} = \I$ for any $j\in [p]$ and
\begin{subequations}
\begin{align}
\bV_{ij} &= \C_{ii}^{-1/2}\C_{ij}\C_{jj}^{-1/2}, \label{invertorig}\\
\tV_{ij} &= (\C_{ii}+\epsilon_n\bP_i^1)^{-1/2}\C_{ij}(\C_{jj}+\epsilon_n\bP_j^1)^{-1/2}, \label{inverttilde}\\
\hV_{ij} &= (\Chat_{ii}^{(n)}+\epsilon_n\bP_i^1)^{-1/2}\Chat_{ij}^{(n)}(\Chat_{jj}^{(n)}+\epsilon_n\bP_j^1)^{-1/2}, \label{inverthat}
\end{align}
\end{subequations}
for $1\leq i, j\leq p$, $i\neq j$.  By \eqref{Vij}, $\bV_{ij}$ is
uniquely defined.  With slight abuse of notation, we write $\bV_{ij}$
as in \eqref{invertorig} even when $\C_{ii}^{-1/2}$ and
$\C_{jj}^{-1/2}$ are not appropriately defined as operators.  Note,
however, that we do need to ensure the operators
$(\C_{jj}+\epsilon_n\bP_j^1)^{-1/2}$ and
$(\hatC_{jj}+\epsilon_n\bP_j^1)^{-1/2}$ in \eqref{inverttilde} and
\eqref{inverthat} are well-defined with high probability.  This is
guaranteed by the following lemma, the proof of which is given in
Lemmas~\ref{posdef} and \ref{boundprobLemma}.
% \nb{insert a lemma}
\begin{Lemma}
\label{lmm:quad-form}
Under the conditions of Theorem \ref{MainThm2}, for sufficiently large
values of $n$, $\C_{jj}+\epsilon_n\bP_j^1\succeq\epsilon_n\I$ for
$1\leq j\leq p$, and with probability at least
$1-2dp\epsilon_n^{-1}n^{-1/2}$,
\begin{align*}
% P\bigg(
\hatC_{jj}+\epsilon_n\bP_j^1\succeq\frac{\epsilon_n}{2}\I,\qquad\text{for }1\leq j\leq p,
% \bigg) \geq 1-2dp\epsilon_n^{-1}n^{-1/2},
\end{align*}
where $d$ is a constant not depending on~$n$.
\end{Lemma}

Next, we show that the solutions to \eqref{re4}$-$\eqref{re6} exist.
% To ensure that the existence of solutions to \eqref{re4} --
% \eqref{re6},
To this end, we are to show that the operators $\bV - \I,~ \tV-\I$ and
$\hV-\I$ are all compact with high probability.  Note that for any
$i\neq j$, the compactness assumption on $\bV_{ij}$ readily shows that
$\bV-\I$ is compact.  Moreover, %under Assumption~\ref{compact},
using the fact that the product of a bounded linear operator and a
compact operator is compact, and that both $\C_{ii}^{1/2}$ and
$(\C_{ii}+\epsilon_n\bP_i^1)^{-1/2}$ are bounded and $\bV_{ij}$ is
compact,
% $\tV_{ij} =
% (\C_{ii}+\epsilon_n\bP_i^1)^{-1/2}\C_{ii}^{1/2}\bV_{ij}\C_{jj}^{1/2}(\C_{jj}+\epsilon_n\bP_j^1)^{-1/2}$,
we see that $\tV_{ij}$ is also compact.  Last but not least, on the
event that it is well-defined, $\hV_{ij}$ is compact since it is of
finite-rank.  In summary, we have the following result.

% Recall that we adopted the compactness assumption on $\bV_{ij}$ to ensure the existence of population APCs.  Under Assumption~\ref{compact}, using the fact that the product of a bounded linear operator and a compact operator is compact, and that $\tV_{ij} = (\C_{ii}+\epsilon_n\bP_i^1)^{-1/2}\C_{ii}^{1/2}\bV_{ij}\C_{jj}^{1/2}(\C_{jj}+\epsilon_n\bP_j^1)^{-1/2}$, we see that $\tV_{ij}$ is also compact.  Since \textcolor{ZurichRed}{$\hV_{ij}$ is of finite-rank}, it is immediate that it is compact.  

\begin{Corollary}
  % Under Assumptions~\ref{assumeRKHS}(a)$-$(e), \ref{compact} and
  % \ref{invertible},
  On the event that the conclusions of Lemma \ref{lmm:quad-form} hold,
  under Assumption~\ref{compact}, the operators $\bV-\I, \tV-\I$ and
  $\hV-\I$ are well-defined and compact.
\end{Corollary}

Note that compactness implies the spectra of $\bV-\I, \tV-\I$ and
$\hV-\I$ are countable with $0$ the only possible accumulation point.
Consequently, the spectra of $\bV, \tV$ and $\hV$ are countable with
$+1$ as the only possible accumulation point.  Assuming that the
smallest eigenvalue of $\bV$ has multiplicity one, the solutions to
\eqref{re4}$-$\eqref{re6} can be obtained as the eigenfunctions
$\f^*, \tilde{\f}^{(n)}$ and $\hf$ corresponding to the smallest eigenvalues
of $\bV, \tV$ and $\hV$, respectively.  In particular, $\f^*$,
$\tilde{\f}^{(n)}$ and $\hf$ are unique (with high probability).  We can
then obtain the solutions to our APC problems by inverse transforming
$\f^*, \tilde{\f}^{(n)}, \hf$ following \eqref{re7}$-$\eqref{re9}.  By
definition, $\bV, \tV$ and $\hV$ are all positive, so their
eigenvalues are bounded below by 0.
% Since $+1$ is the only possible accumulation point of eigenvalues,
% we will consider $+1$ as the dividing line between smallest APCs
% (eigenfunctions corresponding to eigenvalues between 0 and $+1$) and
% largest APCs (eigenfunctions corresponding to eigenvalues larger
% than $+1$.  In fact, we will see in Theorem~\ref{restateAPCwithS}
% that the largest eigenvalues is always upper bounded by $p$, the
% number of variables).

In summary, we rewrite the APC problems as in
\eqref{re1}$-$\eqref{re3} and \eqref{re4}$-$\eqref{re6}, and we know
that under Assumptions~\ref{assumeRKHS}, \ref{compact} and
\ref{invertible}, the solutions to \eqref{re1}$-$\eqref{re3} and
\eqref{re4}$-$\eqref{re6} exist.  
% , we are ready to to establish the consistency of kernelized sample
% APCs.

\subsubsection{Convergence}

We now turn to establishing the consistency of kernelized sample
APCs.  We start by proving the following two key lemmas, which will
lend support to the proof of our main theorems.  The first lemma
% Lemma~\ref{SampleError} 
deals with the difference between $\hV_{ij}$ and $\tV_{ij}$ which
constitutes the stochastic error.
% while Lemma~\ref{ApproxError} deals with the difference between
% $\tV_{ij}$ and $\bV_{ij}$ which constitutes the approximation error.

\begin{Lemma}\label{SampleError}
  % Suppose that Assumptions~\ref{assumeRKHS}(a)$-$(e), \ref{compact},
  % and \ref{invertible} hold, and let $(\epsilon_n)_{n=1}^\infty$ be
  % a sequence of positive numbers such that \eqref{condeps} holds.
  % % $\lim_{n\rightarrow\infty}\epsilon_n=0$.
  % Then, for the i.i.d.\! sample $(X_{1i}, X_{1j}), \ldots, (X_{ni},
  % X_{nj})$,
  Under the conditions of Theorem \ref{MainThm2}, for sufficiently
  large values of $n$, with probability greater than
  $1-2(d_i+d_j)\epsilon^{-1}n^{-1/2}$, we have
  \begin{equation}
    \|\hV_{ij} - \tV_{ij}\| \leq C\epsilon_n^{-3/2}n^{-1/2}, \label{hVtV}
  \end{equation}
  where $C, d_i$ and $d_j$ are constants not depending on $n$.
\end{Lemma}

% while Lemma~\ref{ApproxError} 
The second lemma deals with the \emph{deterministic} difference
between $\tV_{ij}$ and $\bV_{ij}$ which can be viewed as approximation
error.

\begin{Lemma}\label{ApproxError}
  Suppose that Assumptions~\ref{assumeRKHS}, \ref{compact}
  and \ref{invertible} hold, and let $(\epsilon_n)_{n=1}^\infty$ be a
  sequence of positive numbers such that
  $\lim_{n\rightarrow\infty}\epsilon_n=0$.  Then
\begin{equation}
  \|\tV_{ij} - \bV_{ij}\| \longrightarrow 0\quad (\text{as }n\longrightarrow\infty). \label{approx}
\end{equation}
\end{Lemma}

Combining Lemma~\ref{SampleError} and \ref{ApproxError} leads to the
following result:

\begin{Theorem}\label{MainThm1}
  Suppose that Assumptions~\ref{assumeRKHS}, \ref{compact}
  and \ref{invertible} hold, and that the eigenspace which attain the
  eigenvalue problem
  \[
  \min_{\f\in\bHsp}\langle\f, \bV\f\rangle_k \qquad\text{subject to
  }\langle\f, \f\rangle_k = 1
  \]
  is one-dimensional, spanned by $\f^* = (f_1^*, \ldots, f_p^*)$ with $\langle
  \f^*,\f^*\rangle_k = 1$.  Let
  $(\epsilon_n)_{n=1}^\infty$ be a sequence of positive numbers which
  satisfies \eqref{condeps}.  Let $\hf$ be the unit eigenfunction for
  the smallest eigenvalue of $\hV$.  Then,
  \[
  |\langle\hf, \f^*\rangle_k| \longrightarrow 1
  \]
  in probability, as $n\longrightarrow\infty$.
\end{Theorem}

Finally,
% One can imagine that 
Theorem~\ref{MainThm2} follows from Theorem~\ref{MainThm1} once we
transform $\f^*$ and $\hf$ back to $\bPhi^*$ and $\hPhi$ using the
relationship in \eqref{re7} and \eqref{re9}. For details of the proof
for the two key lemmas and main theorems, see
Appendices~\ref{mainLemmaProof} and \ref{mainThmProof}.

% While in the theorems above we restrict our attention to the eigenfunctions corresponding to the smallest eigenvalue of $\hV$ (and hence, the smallest APCs), the convergence of eigenspaces corresponding to the $\ell^{th}$ smallest eigenvalue (and hence, the $\ell^{th}$ smallest APCs) can be obtained by extending Lemma~\ref{evecconv} in Appendix~\ref{FunctionalFacts}.

%%%%%%%%%%%%%%%%%%%%%%%%%%%%%%%%%%%%%%%%%%%%%%%%%%%

%!TEX root = Paper.tex

% Estimation and Computation
\section{Estimation and Computation} 
\label{compute}

In this section, we motivate an iterative method for computing
kernelized APCs.  This involves the use of power algorithm, an
iterative algorithm for extracting the first few largest (or smallest)
eigenfunctions of a bounded linear operator.  In addition to detailing
out the algorithm, we provide rigorous theoretical justification of 
the use of power algorithm in the RKHS framework.

% We begin with a brief description of the principle underlying the
% power algorithm.  
Consider a matrix $\bM$ with the eigen-decomposition
$\bM=\sum_{i=1}^m\lambda_i\bM_i$, where $\bM_i = \bu_i\bu_i^T$ and the
eigenvalues $\lambda_1>\lambda_2>\cdots>\lambda_m$ are distinct.  The
power algorithm allows us to compute the eigenvector $\bu_1$
corresponding to the largest eigenvalue $\lambda_1$ (see, e.g., \cite{Golub2013}) by forming normalized powers
$\bM^n \bu_0 / \|\bM^n \bu_0\|$ which can be shown to converge to
$\bu_1$ as long as $\bu_0$ is not orthogonal to $\bu_1$.

% starting from an
% initial vector $\bu_0$ with $\bu_0^T\bu_1\neq 0$.  By repeatedly
% premultiplying $\bu_0$ by $\bM$ and normalizing the resulting vector,
% we obtain, after $n$ iterations,
% \[
 % \frac{\bM^n\bu_0}{\|\bM^n\bu_0\|} = \frac{\sum_{i=1}^m\lambda_i^n\bM_i\bu_0}{\sqrt{\sum_{i=1}^m\lambda_i^{2n}\|\bM_i\bu_0\|^2}} = \frac{\sum_{i=1}^m\frac{\lambda_i^n}{\lambda_1^n}\bM_i\bu_0}{\sqrt{\sum_{i=1}^m\frac{\lambda_i^{2n}}{\lambda_1^{2n}}\|\bM_i\bu_0\|^2}} \stackrel{n\rightarrow\infty}\longrightarrow \frac{\bM_1\bu_0}{\|\bM_1\bu_0\|} = % \pm\bu_1.\]
% The first equality follows from the fact that if $\lambda$ is an eigenvalue of $\bM$ with corresponding eigenvector $\bu$, then $\lambda^n$ is an eigenvalue of $\bM^n$ with corresponding eigenvector $\bu$, whereas the last equality follows from the fact that $\bM_1$ projects $\bu_0$ to the direction of $\pm\bu_1$.

To compute the eigenvector $\bu_m$ corresponding to the smallest
eigenvalue $\lambda_m$, the spectrum needs to be flipped and shifted
by replacing $\bM$ with $\gamma\I-\bM$ in the power algorithm. If
$0\leq \lambda_1<\lambda_m\leq B$ for some $B>0$, then using
$\gamma = (B+1)/2$, we have
\begin{align*}
-\frac{B-1}{2} \leq &\gamma-\lambda_i \leq \frac{B-1}{2} &&\qquad\text{if }1\leq\lambda_i\leq B,\\
\frac{B-1}{2} \leq &\gamma-\lambda_i \leq \frac{B-1}{2}+1 &&\qquad\text{if }0\leq\lambda_i\leq 1.
\end{align*}
In this case, the large eigenvalues of $\bM$, $\{\lambda: \lambda>1\}$, are mapped to an interval centered at 0, while the small eigenvalues $\{\lambda: \lambda<1\}$ are affixed to the right end of this interval.

% The Smoothing Operator
\subsection{Eigencharacterization of Kernelized APCs}
To relate power algorithm to kernelized APCs, we first show that the kernelized APC problem \eqref{statePenSample} can be reformulated as an eigenvalue problem with respect to a new inner product $\langle\cdot, \cdot\rangle_\star$ defined on $\bHsp$ (which is initially endowed with $\langle\cdot, \cdot\rangle_k$).  As a consequence, kernelized APC can be obtained as the eigenfunction corresponding to the smallest eigenvalue of an operator $\tS$ defined on $(\bHsp, \langle\cdot, \cdot\rangle_\star)$.  Then, computation of kernelized sample APC reduces to an application of power algorithm on an empirical version of $\tS$.

\begin{Theorem}\label{RKHSstar}
Suppose that Assumption~\ref{assumeRKHS} holds, and $\Hsp_j^0$ excludes constants.  Then $\Hsp_j$ is a reproducing kernel Hilbert space with respect to the inner product
\[\langle f, g\rangle_{\star_j} := \Cov[f(X_j), g(X_j)] + \alpha_j\langle f, g\rangle_{k_j^1}, \qquad f, g\in\Hsp_j, \alpha_j>0.\]
\end{Theorem}

From Theorem~\ref{RKHSstar}, we see that it is reasonable to endow $\bHsp = \Hsp_1\times\cdots\times\Hsp_p$ with the new inner product
\[\langle\bPhi, \bPsi\rangle_\star := \sum_{j=1}^p\langle\phi_j, \psi_j\rangle_{\star_j},\]
where $\bPhi = (\phi_1, \ldots, \phi_p), \bPsi = (\psi_1, \ldots, \psi_p)$ are both elements of $\bHsp$.  

We now introduce the smoothing operator $\bS_{ij}$, defined through a ``generalized" regularized population regression problem: 
\begin{align}
\bS_{ij}: (\Hsp_j, \langle\cdot, \cdot\rangle_{\star_j}) &\longrightarrow (\Hsp_i, \langle\cdot, \cdot\rangle_{\star_i}), \label{smoothij}\\
\phi_j &\longmapsto\argmin_{f\in\Hsp_i} \left\{\Var[\phi_j(X_j)-f(X_i)] + \alpha_i \|f\|_{k^1_i}^2\right\}. \nonumber
\end{align}
Indeed, the problem reduces to the population version of the usual regularized regression problem, when $\phi_j$ and $f$ are both required to have mean zero.  With the establishment of existence and uniqueness of solution to the problem, the smoothing operator $\bS_{ij}$ in \eqref{smoothij} mapping $\phi_j$ to its ``smoothed" version in $\Hsp_i$ is well-defined.  In addition, it enjoys some nice properties:

\begin{Theorem} \label{milestone}
Suppose that $(\Hsp_j, \langle\cdot, \cdot\rangle_{k_j})$ is a reproducing kernel Hilbert space satisfying Assumption~\ref{assumeRKHS} with $\Hsp_j^0$ excluding constants, and $\alpha_j>0$, for $1\leq j\leq p$.  Then $\bS_{ij}$ is well-defined, bounded, linear, compact, and
\begin{equation}
\langle \phi_i, \bS_{ij}\phi_j\rangle_{\star_i} = \Cov[\phi_i(X_i), \phi_j(X_j)], \qquad\forall\phi_i\in\Hsp_i, \phi_j\in\Hsp_j.
\end{equation}
Moreover, 
\begin{equation}
\|\bS_{ij}\phi_j\|_{\star_i} \leq \big(\Var[\phi_j(X_j)]\big)^{1/2} \leq \|\phi_j\|_{\star_j}, \qquad\forall\phi_j\in\Hsp_j. \label{contraction}
\end{equation}
\end{Theorem}

Theorem~\ref{milestone} says that the operator $\bS_{ij}$ is not only well-defined, but acts as the ``covariance operator" with respect to $\langle\cdot, \cdot\rangle_{\star_i}$ (whereas the covariance operator $\C_{ij}$ in Section~\ref{Consistency} is defined with respect to to $\langle\cdot, \cdot\rangle_{k_i}$).  In addition, it is a contraction operation by \eqref{contraction}.  We are now ready to restate the kernelized APC problem as an eigenvalue problem with respect to the inner product $\langle\cdot, \cdot\rangle_\star$.  

\begin{Theorem}\label{restateAPCwithS}
Let $\bHsp = \Hsp_1\times\cdots\times\Hsp_p$, where $\Hsp_j$ is a reproducing kernel Hilbert space with respect to $\langle\cdot, \cdot\rangle_{\star_j}$, for $1\leq j\leq p$. Then the kernelized APC problem \eqref{statePenSample} can be restated as
\begin{equation}
\min_{\bPhi\in\bHsp}\langle\bPhi, \tS\bPhi\rangle_\star
\qquad\text{subject to}\qquad
\langle\bPhi, \bPhi\rangle_\star = 1, \label{SmoothAPC}
\end{equation}
where $\tS: \bHsp\longrightarrow\bHsp$ is defined by the component mapping
\begin{equation}
[\tS\bPhi]_i = \sum_{j\neq i}\bS_{ij}\phi_j + \phi_i, \label{compmap}
\end{equation}
and $\bS_{ij}$ is the smoothing operator as defined in \eqref{smoothij}.  Moreover, $\tS$ is self-adjoint, positive, and bounded above by $p$.
\end{Theorem}

We see that the smoothing operators $\bS_{ij}$'s play an important role in the kernelized APC problem.  This is reminiscent of the role of orthogonal projection operators $\bP_{ij}$'s in the original APC problem in \cite{DonnellBuja1994}.

To address the issue of existence of eigenvalues, and hence the solution of \eqref{SmoothAPC} when $\bHsp$ is infinite-dimensional, we recall the usual compactness condition.  We know that $\bS_{ij}$ is compact from Theorem~\ref{milestone}.  Although this does not imply the compactness of $\tS$, we have the following result:
\begin{Theorem}\label{SminusIcompact}
The operator $\tS-\I: (\bHsp, \langle\cdot, \cdot\rangle_\star)\longrightarrow(\bHsp, \langle\cdot, \cdot\rangle_\star)$,
\[\tS -\I = \begin{pmatrix} \0 & \bS_{12} & \cdots & \bS_{1p}\\ \bS_{21} & \0 & \cdots & \bS_{2p} \\ 
\vdots & \vdots & \ddots & \vdots \\ \bS_{p1} & \bS_{p2} & \cdots & \0 \end{pmatrix},\] 
where $\bS_{ij}: (\Hsp_j, \langle\cdot, \cdot\rangle_{\star_j}) \longrightarrow (\Hsp_i, \langle\cdot, \cdot\rangle_{\star_i})$ is as given in \eqref{smoothij}, is compact.
\end{Theorem}

The compactness attribute implies that $\tS-\I$ has an eigendecomposition with eigenvalues that can only accumulate at 0, which in turn implies that the only possible accumulation point of the eigenvalues of $\tS$ is $+1$.  To this end, we see that kernelized APC can be obtained as the eigenfunction corresponding to the smallest eigenvalue of $\tS$.  In addition, we see that $+1$ is a natural dividing line between smallest and largest kernelized APCs.  

\subsection{Power Algorithm for Kernelized APCs}
Applying the knowledge that kernelized APCs are the eigenfunctions of $\tS$ from a population standpoint, we execute the power algorithm on $\gamma\I-\tS$ to solve for the (smallest) kernelized APC.  The pseudocode is given below.  Here $\gamma$ is taken to be $(p+1)/2$ since the spectrum of $\tS$ is bounded above by $p$, as claimed in Theorem~\ref{restateAPCwithS}.  We see that solving for kernelized APC reduces to iterative smoothing of each component $\phi_j$ against $X_i$, for $j\neq i$.  

\begin{algorithm}[htp]
\caption{Computation of kernelized APC}
\begin{algorithmic}
\STATE Let $\gamma=(p+1)/2$.  Initialize $t = 0$, $\bPhi^{[0]} = (\phi_1^{[0]}, \phi_2^{[0]}, \ldots, \phi_p^{[0]})$.
\REPEAT
\FOR{$i=1, \ldots, p$}
\STATE $\phi_i\leftarrow\gamma\phi_i^{[t]}-(\sum_{j\neq i}\bS_{ij}\phi_j^{[t]}+\phi_i^{[t]})$ 
\COMMENT{Update steps}
\ENDFOR
\STATE Standardize with $c = (\sum\|\phi_i\|_{\star_i}^2)^{-1/2}$\textcolor{ZurichRed}
\STATE $(\phi_1^{[t+1]}, \phi_2^{[t+1]}, \ldots, \phi_p^{[t+1]})\leftarrow(c\phi_1, c\phi_2, \ldots, c\phi_p)$
\STATE $t \leftarrow t+1$
\UNTIL{$\Var\sum\phi_i^{[t]}+\sum\alpha_i\|\phi_i^{[t]}\|_{k_i^1}^2$ converges}
\end{algorithmic}
\label{PowerAlgo}
\end{algorithm}

To compute the $k^{th}$ smallest kernelized APCs for $k>1$, we just need to add a series of Gram-Schmidt steps 
\[\phi_i\leftarrow\phi_i-\bigg(\sum_{j=1}^p\langle\phi_{\ell, j}, \phi_j^{[t]}\rangle_{\star_j}\bigg)\phi_{\ell, i}, \qquad 1\leq\ell\leq k-1\]
following the update steps in Algorithm 1, to ensure that the orthogonality requirements \eqref{orthogonal} are satisfied.  Here $\bPhi_\ell = (\phi_{\ell, 1}, \cdots, \phi_{\ell, p})$ is the previously obtained $\ell^{th}$ smallest kernelized APC, where $1\leq\ell\leq k-1$. 

The power algorithm is guaranteed to converge under mild conditions:
\begin{Proposition}\label{powerconv}
Suppose that the smallest eigenvalue of $\tS$ is of multiplicity one with corresponding unit eigenfunction $\tilde{\bPhi}$.  If the power algorithm is initialized with $\bPhi^{[0]}$ that has a nontrivial projection onto $\tilde{\bPhi}$, then the power algorithm converges.
\end{Proposition}

For implementation details see Appendix~\ref{sec:power-details}.

%%%%%%%%%%%%%%%%%%%%%%%%%%%%%%%%%%%%%%%%%%%%%%%%%%%

%!TEX root = Paper.tex

% Concluding Remarks 
\section{Concluding Remarks}
\label{conclusion}

APCs are a useful tool for exploring additive degeneracy in data.  In
this paper, we propose the estimation of APCs using the shrinkage
regularization approach through kernelizing, and we establish the
consistency of the resulting kernelized sample APCs.

It would be interesting to generalize our study of APCs in several
directions.  Due to the nonparametric nature of APC estimation, so far
we have implicitly assumed that the sample size $n$ is large relative
to the total number of variables $p$.  It would be interesting to
extend APCs to the high-dimensional setting where $p$ can be
comparable to $n$.  It would then be natural to impose additional
structure such as sparsity in a flavor similar to the sparse additive
models proposed by \cite{Ravikumar2009} in the regression framework.
It would also be interesting to study the largest APCs and to examine
whether it provides meaningful interpretation through dimensionality
reduction as in conventional PCA.  

Estimation of APCs is non-trivial due to its unsupervised learning
nature.  We have left open the problems of kernel choice and how to
optimally and differentially select smoothing parameters for different
variables within an APC and also across different APCs.

%%%%%%%%%%%%%%%%%%%%%%%%%%%%%%%%%%%%%%%%%%%%%%%%%%%

% Acknowledgements should go at the end, before appendices and references
\vskip 0.2in
\noindent{\bf Acknowledgements: }{We would like to acknowledge support for this project from the
  National Science Foundation under NSF grant DMS-1310795 to A.B. and
  NSF career grant DMS-1352060 to Z.M. as well as a grant from the
  Simons Foundation (SFARI) to A.B.  We also thank Ming Yuan for
  valuable discussions.}


\vskip 0.2in
\bibliographystyle{chicago}
\bibliography{Ref.bib}

\newpage
\appendix
\newpage

%%%%%%%%%%%%%%%%%%%%%%%%%%%%%%%%%%%%%%%%%%%%%%%%%%%%%

%!TEX root = Paper.tex

\section{Consistency Proof of Section~\ref{Consistency}}
\label{consproof}

In this section, we give the consistency proof for kernelized sample APCs.  We first consider some useful facts from functional analysis in Section~\ref{FunctionalFacts}.  We then present in Section~\ref{PrepFacts} some basic properties of the operator $\C_{jj}+\epsilon_n\bP_j^1$ and the square root of its inverse (which exists under mild conditions), which forms the building blocks of the proof of two key lemmas and main theorems in Section~\ref{mainLemmaProof} and Section~\ref{mainThmProof}.
 
% Some Facts from Functional Analysis
\subsection{Some Facts from Functional Analysis}\label{FunctionalFacts}
% Lemma
\begin{Lemma}\label{lambdamin}
Let $A, B$ be positive, self-adjoint operators on a Hilbert space.  Then
\[\|A^{1/2}-B^{1/2}\| \leq \frac{\|A-B\|}{\lambda_{\text{\rm min}}(A^{1/2}) + \lambda_{\text{\rm min}}(B^{1/2})}.\]
\end{Lemma}
\begin{proof}
The proof essentially follows from \cite{Schmitt1992}.  Let $D = B-A$ and $X = B^{1/2}-A^{1/2}$.  It is straightforward to verify that
\[XA^{1/2}+B^{1/2}X = D.\]
It then follows that
\begin{equation}
X = E_2XE_1 + F, \label{eqalg}
\end{equation}
where
\begin{align*}
E_1 &= (qI+A^{1/2})^{-1}(qI-A^{1/2}), \\
E_2 &=  (qI+B^{1/2})^{-1}(qI-B^{1/2}), \\
F &= 2q(qI + B^{1/2})^{-1}D(qI+A^{1/2})^{-1},
\end{align*}
for all $q>0$.  Applying triangle inequality to \eqref{eqalg} gives 
\[\|X\| \leq \|E_2\|\|X\|\|E_1\| + \|F\|,\]
or, equivalently,
\begin{equation}
\|X\| \leq \frac{\|F\|}{1-\|E_1\|\|E_2\|}. \label{eve}
\end{equation}

Our goal now is to find upper bounds for $\|E_1\|\|E_2\|$ and $\|F\|$.  An upper bound for $\|F\|$ is given by
\begin{align}
\|F\| &\leq 2q\|D\|\|(qI+A^{1/2})^{-1}\|\|(qI + B^{1/2})^{-1}\| \nonumber\\
&\leq \frac{2q\|D\|}{[q+\lambda_{\text{min}}(A^{1/2})][q+\lambda_{\text{min}}(B^{1/2})]}. \label{firstlambdamin}
\end{align}
We now get an upper bound for $\|E_1\|^2$.  First note that
\begin{align*}
\frac{\|E_1x\|^2}{\|x\|^2} &= \frac{\|(qI+A^{1/2})^{-1}(qI-A^{1/2})x\|^2}{\|x\|^2}, \qquad\text{let }y = (qI+A^{1/2})^{-1}x \nonumber\\
&= \frac{\|(qI+A^{1/2})^{-1}(qI-A^{1/2})(qI+A^{1/2})y\|^2}{\|(qI+A^{1/2})y\|^2} \\
&= \frac{\|(qI-A^{1/2})y\|^2}{\|(qI+A^{1/2})y\|^2} \\
&= \frac{q^2\|y\|^2 - 2q\langle y, A^{1/2}y\rangle + \|A^{1/2}y\|^2}{q^2\|y\|^2 + 2q\langle y, A^{1/2}y\rangle + \|A^{1/2}y\|^2},
\end{align*}
where the third equality is due to commutativity between $qI-A^{1/2}$ and $qI+A^{1/2}$.  It then follows that
\begin{align}
\|E_1\|^2 &=\max_{x\neq 0}\frac{\|E_1x\|^2}{\|x\|^2} \nonumber\\
&= \max_{y\neq 0}\frac{q^2\|y\|^2 - 2q\langle y, A^{1/2}y\rangle + \|A^{1/2}y\|^2}{q^2\|y\|^2 + 2q\langle y, A^{1/2}y\rangle + \|A^{1/2}y\|^2} \nonumber\\
&\leq \max_{y\neq 0}\frac{q^2 - 2q\lambda_{\text{min}}(A^{1/2}) + \|A^{1/2}y\|^2/\|y\|^2}{q^2 + 2q\lambda_{\text{min}}(A^{1/2}) + \|A^{1/2}y\|^2/\|y\|^2} \nonumber\\
&= \max_{y\neq 0}\bigg[1-\frac{4q\lambda_{\text{min}}(A^{1/2})}{q^2 + 2q\lambda_{\text{min}}(A^{1/2}) + \|A^{1/2}y\|^2/\|y\|^2}\bigg] \nonumber\\
&= 1-\frac{4q\lambda_{\text{min}}(A^{1/2})}{q^2 + 2q\lambda_{\text{min}}(A^{1/2}) + \lambda^2_{\text{max}}(A^{1/2})}\bigg] \nonumber\\
&= 1 - \frac{4\lambda_{\text{min}}(A^{1/2})}{q} + O\bigg(\frac{1}{q^2}\bigg), \label{secondlambdamin}
\end{align}
the inequality holds because the rational expression is monotonically decreasing in $\langle y, A^{1/2}y\rangle$.  Similarly, an upper bound for $\|E_2\|^2$ is
\begin{equation}
\|E_2\|^2 \leq 1 - \frac{4\lambda_{\text{min}}(B^{1/2})}{q} + O\bigg(\frac{1}{q^2}\bigg). \label{thirdlambdamin}
\end{equation}
So we see that $\|E_1\|\|E_2\|<1$ when $q$ is sufficiently large.  Combining \eqref{eve}, \eqref{firstlambdamin}, \eqref{secondlambdamin} and \eqref{thirdlambdamin} gives
\begin{align*}
\|X\| &\leq \frac{\|F\|}{1-\|E_1\|\|E_2\|} \ \leq\ \frac{\|F\|}{1-\half(\|E_1\|^2 + \|E_2\|^2)} \\ 
&\leq \frac{2q\|D\|}{[q+\lambda_{\text{min}}(A^{1/2})][q+\lambda_{\text{min}}(B^{1/2})]}\cdot\frac{1}{2q^{-1}[\lambda_{\text{min}}(A^{1/2})+\lambda_{\text{min}}(B^{1/2})+O(q^{-1})]} \\
&\leq \frac{\|D\|}{\lambda_{\text{min}}(A^{1/2})+\lambda_{\text{min}}(B^{1/2})+O(q^{-1})}.
\end{align*}
Sending $q$ to infinity gives the desired bound.
\end{proof}

The following lemmas are from Appendix B of \cite{Fukumizu2007}, and correspond to Lemma 8 and 9 in \cite{Fukumizu2007}, respectively.
% Lemma
\begin{Lemma}\label{Fukumizu8}
Let $A$ and $B$ be positive, self-adjoint operators on a Hilbert space, with $0\preceq A\preceq\lambda\I$ and $0\preceq B\preceq\lambda\I$ for some $\lambda>0$.  Then
\[\|A^{3/2}-B^{3/2}\| \leq 3\lambda^{1/2}\|A-B\|.\]
\end{Lemma}

% Lemma
\begin{Lemma}\label{Fukumizu9}
Let $\Hsp_1, \Hsp_2$ be Hilbert spaces, and let $\Hsp$ be a dense linear subspace of $\Hsp_2$.  Suppose that $A_n$ and $A$ are bounded operators on $\Hsp_2$, and $B: \Hsp_1\longrightarrow\Hsp_2$ is compact.  If
\[A_nx\longrightarrow Ax\]
for all $x\in\Hsp$, and
\[\sup_n\|A_n\| <\infty,\]
then $\|A_nB-AB\|\longrightarrow 0$.
\end{Lemma}

The following lemma is adapted from Lemma 10 of \cite{Fukumizu2007}, but proved for the case of the smallest eigenfunction (rather than the largest eigenfunction).
% Lemma
\begin{Lemma}\label{evecconv}
Let $A$ be a compact operator on a Hilbert space $\Hsp$, and $A_n (n\in\N)$ be bounded  operators on $\Hsp$ such that $A_n$ converges to $A$ in norm.  Assume that the eigenspace of $A$ corresponding to the smallest eigenvalue is one-dimensional and spanned by a unit function $f_1$, and the minimum of the spectrum of $A_n$ is attained by a unit eigenfunction $\hfnob_1$.  Then
\[|\langle\hfnob_1, f_1\rangle | \longrightarrow1\qquad\text{as }n\longrightarrow\infty.\] 
\end{Lemma}
\begin{proof}
Since $A$ is compact, the eigendecomposition
\[A = \sum_{i=1}^\infty\lambda_i\langle f_i, \cdot\rangle f_i\]
holds, where $\{\lambda_i\}$ is the eigenvalues and $\{f_i\}$ is the corresponding eigenfunctions so that $\{f_i\}$ forms a complete orthonormal basis system of $\Hsp$.  For convenience, we denote by $\lambda_1$ and $\lambda_2$ the smallest and second smallest eigenvalue of $A$, respectively.  By assumption $\lambda_1<\lambda_2$.

Let $\delta_n = |\langle \hfnob_1, f_1\rangle|$.  We have
\begin{align}
\langle\hfnob_1, A\hfnob_1\rangle &= \lambda_1\langle\hfnob_1, f_1\rangle^2 + \sum_{i=2}^\infty\lambda_i\langle\hfnob_1, f_i\rangle^2 \nonumber\\
&\geq \lambda_1\langle\hfnob_1, f_1\rangle^2 + \lambda_2\sum_{i=2}^\infty\langle\hfnob_1, f_i\rangle^2 \nonumber\\
&= \lambda_1\delta_n^2 + \lambda_2(1-\delta_n^2) \label{delta}\\
&\geq \lambda_1. \nonumber
\end{align}
On the other hand, since $A_n$ converges to $A$ in norm, we have
\[\sup_{\|f\|\leq 1}\|A_nf-Af\| \longrightarrow 0,\]
hence
\[\sup_{\|f\|\leq 1}\|\langle f, A_nf\rangle-\langle f, Af\rangle\| \leq \sup_{\|f\|\leq 1}\|A_nf-Af\| \longrightarrow 0.\]
This in turn implies that
\begin{equation}
\lambda_1 = \langle f_1, Af_1\rangle \leq \langle\hfnob_1, A\hfnob_1\rangle = \langle\hfnob_1, A_n\hfnob_1\rangle+o(1) = \hlambda_1+o(1), \label{gogeq}
\end{equation}
and
\begin{equation}
\hlambda_1 = \langle\hfnob_1, A_n\hfnob_1\rangle \leq \langle f_1, A_nf_1\rangle = \langle f_1, Af_1\rangle + o(1) = \lambda_1+o(1). \label{goleq}
\end{equation}
Combining \eqref{gogeq} and \eqref{goleq} gives
\[\hlambda_1 \longrightarrow \lambda_1,\]
and it follows that $\langle\hfnob_1, A\hfnob_1\rangle$ must converge to $\lambda_1$.  Applying this and $\lambda_1 < \lambda_2$ to \eqref{delta} gives $\delta_n^2\longrightarrow 1$.
\end{proof}

% Proof of Preparatory Lemmas
\subsection{Proofs of Supporting Lemmas} \label{PrepFacts}
In this subsection, we will consider some lemmas that will be directly useful for establishing the proofs in Sections~\ref{mainLemmaProof} and \ref{mainThmProof}.  The following lemma corresponds to Lemma 5 in \cite{Fukumizu2007}, and bounds the Hilbert-Schmidt norm of the difference between the empirical covariance operator and the (population) covariance operator.
% Lemma
\begin{Lemma}\label{HS}
The cross-covariance operator $\C_{ij}$ is Hilbert-Schmidt, and
\[E\|\hatC_{ij} - \C_{ij}\|_{\text{HS}} = O(n^{-1/2}),\]
where $\|\cdot\|_{\text{HS}}$ denotes the Hilbert-Schmidt norm of a Hilbert-Schmidt operator.
\end{Lemma}

In the following, we first consider a sufficient condition for the operator $\C_{jj}+\epsilon\bP_j^1$ to be bounded below by $\epsilon$.  Once this is done, $\Vright$ is well-defined and is bounded above by $\epsilon^{-1/2}$.  To simplify notation, the subscript $j$ is omitted.

Let $\Hsp = \Hsp^0\oplus\Hsp^1$ be an RKHS with reproducing kernel $k = k^0 + k^1$.  Suppose that each $f\in\Hsp$ is a real-valued function with domain $\mX$, and the random variable $X$ also takes value in $\mX$.  Following \cite{Baker1973} and \cite{Fukumizu2007}, under the condition $E[k(X, X)]<\infty$, there exists a unique covariance operator $\C$ on $\Hsp$ such that
\[\Cov[f(X), g(X)] = \langle f, \C g\rangle_k, \qquad\forall f, g\in\Hsp.\]
Moreover, $\C$ induces the covariance operators $\C_{00}$ and $\C_{11}$ on $\Hsp^0$ and $\Hsp^1$, and the cross-covariance operator $\C_{01}: \Hsp^1\longrightarrow\Hsp^0$.  Also, $\C_{01}$ has the representation
\[\C_{01} = \C_{00}^{1/2}\bV_{01}\C_{11}^{1/2},\]
where $\bV_{01}: \Hsp^1\longrightarrow\Hsp^0$ is a unique bounded operator with $\|\bV_{01}\|\leq 1$.

% Lemma
\begin{Lemma}
\label{posdef}
Let $\Hsp = \Hsp^0\oplus\Hsp^1$ be a reproducing kernel Hilbert space with reproducing kernel $k = k^0 + k^1$, such that $E[k(X, X)]<\infty$, and let $0<a<b$ be absolute constants.  Suppose $\Hsp^0$ is finite-dimensional with $\C_{00}\succeq b\I$ and $\|\bV_{01}\|\leq\sqrt{1-a/b}$.  Then for any $\epsilon\in (0, a)$, 
\[\C-\epsilon\bP^0\succeq \0, \qquad \C+\epsilon\bP^1\succeq\epsilon\I.\]
\end{Lemma}
\begin{proof}
For any $f\in\Hsp$, there is a unique decomposition $f = f^0 + f^1$, where $f^0\in\Hsp^0$ and $f^1\in\Hsp^1$.  Thus, it is straightforward to verify that
\[\langle f, (\C-\epsilon\bP^0)f\rangle_k = \langle f^0, (\C_{00}-\epsilon\I)f^0\rangle_{k^0} + \langle f^1, \C_{11}f^1\rangle_{k^1} + 2\langle f^0, \C_{01}f^1\rangle_{k^0}.\]
The desired claim is true if we can show that 
\begin{equation}
|\langle f^0, \C_{01}f^1\rangle_{k^0}| \leq \|(\C_{00}-\epsilon\I)^{1/2}f^0\|_{k^0}\|\C_{11}^{1/2}f^1\|_{k^1}. \label{boundposdef}
\end{equation}
To this end, using the fact that $\C_{01} = \C_{00}^{1/2}\bV_{01}\C_{11}^{1/2}$, we get
\begin{align*}
|\langle f^0, \C_{01}f^1\rangle_{k^0}| &= |\langle f^0, \C_{00}^{1/2}\bV_{01}\C_{11}^{1/2}f^1\rangle_{k^0}| \\
&= |\langle (\C_{00}-\epsilon\I)^{1/2}f^0, (\C_{00}-\epsilon\I)^{-1/2}\C_{00}^{1/2}\bV_{01}\C_{11}^{1/2}f^1\rangle_{k^0}| \\
&\leq \|(\C_{00}-\epsilon\I)^{1/2}f^0\|_{k^0}\|(\C_{00}-\epsilon\I)^{-1/2}\C_{00}^{1/2}\bV_{01}\C_{11}^{1/2}f^1\|_{k^0} \\
&\leq \|(\C_{00}-\epsilon\I)^{1/2}f^0\|_{k^0}\|\C_{11}^{1/2}f^1\|_{k^1}\|(\C_{00}-\epsilon\I)^{-1/2}\C_{00}^{1/2}\bV_{01}\|.
\end{align*}
Therefore, to establish \eqref{boundposdef} we only need to show that 
\[\|(\C_{00}-\epsilon\I)^{-1/2}\C_{00}^{1/2}\bV_{01}\| \leq 1.\]
To this end, let $\Dim(\Hsp^0) = m$, then we have $\C_{00} = \sum_{i=1}^m\lambda_i\langle\psi_i, \cdot\rangle_{k^0}\psi_i$, where $\lambda_i\geq b$ and the $\psi_i$'s are of unit norm.  Hence,
\[(\C_{00}-\epsilon\I)^{-1/2}\C_{00}^{1/2} = \sum_{i=1}^m\bigg(\frac{\lambda_i}{\lambda_i-\epsilon}\bigg)^{1/2}\langle\psi_i, \cdot\rangle_{k^0}\psi_i = \sum_{i=1}^m\bigg(\frac{1}{1-\epsilon/\lambda_i}\bigg)^{1/2}\langle\psi_i, \cdot\rangle_{k^0}\psi_i,\]
and so 
\begin{align*}
\|(\C_{00}-\epsilon\I)^{-1/2}\C_{00}^{1/2}\bV_{01}\| &\leq \|(\C_{00}-\epsilon\I)^{-1/2}\C_{00}^{1/2}\|\|\bV_{01}\| \\
&\leq\max_{1\leq i\leq m}\bigg(\frac{1}{1-\epsilon/\lambda_i}\bigg)^{1/2}\|\bV_{01}\| \\
&\leq \bigg(\frac{1}{1-\epsilon/b}\bigg)^{1/2}\|\bV_{01}\| \\
&\leq \bigg(\frac{1}{1-a/b}\bigg)^{1/2}\|\bV_{01}\| \\
&\leq 1.
\end{align*}
This completes the proof.
\end{proof}
\begin{Remark}
{\rm
The condition that $\Hsp^0$ is finite-dimensional seems unavoidable.  Since $\C_{00}$ has to be trace class as in the definition of covariance operator, the additional condition $\C_{00}\succeq b\I$ only makes sense when $\Dim(\Hsp^0)<\infty$.
}
\end{Remark}

% Lemma
With $\Vright$ now well-defined, we are ready to derive more of its properties, which will be useful later on.
\begin{Lemma}\label{facts}
Suppose that Assumptions~\ref{assumeRKHS} and \ref{invertible} hold.  Then for $\epsilon$ sufficiently small, the operators $(\Vrightno)^{-1}$ and $\Vright$ are self-adjoint and bounded above by $\epsilon^{-1}$ and $\epsilon^{-1/2}$, respectively.  In addition,  
\begin{align*}
&\|\C_{jj}^{1/2}-(\Vrightno)^{1/2}\| \leq \epsilon^{1/2}, \\
&\|\Vright\C_{jj}^{1/2}\|  = \|\C_{jj}^{1/2}\Vright\| \leq 2, \\
&\|\Vleft\Vmid\Vright\| \leq 4.
\end{align*}
\end{Lemma}
\begin{proof}
Since the adjoint of the inverse of an operator is equal to the inverse of the adjoint of the operator, and $\C_{jj}$ and $\bP_j^1$ are self-adjoint, we get
\[[(\Vrightno)^{-1}]^* = [(\Vrightno)^*]^{-1} = (\Vrightno)^{-1}.\]
Now given a positive self-adjoint operator $A$, if $B$ is the unique square root of $A$, i.e. $B^2 = A$, then $(B^*)^2 = (B^*)(B^*) = (B^2)^* = A^* = A$, so by uniqueness of the square root of a positive self-adjoint operator we have $B = B^*$.  This implies that $\Vright$ is self-adjoint. 

Under Assumptions~\ref{assumeRKHS} and \ref{invertible}, we have $\Vrightno\succeq\epsilon\I$ for $\epsilon$ sufficiently small, by Lemma~\ref{posdef}.  It follows immediately that $(\Vrightno)^{-1}$ and $\Vright$ are bounded above by $\epsilon^{-1}$ and $\epsilon^{-1/2}$, respectively.  

Applying the inequality
\[\|A^{1/2}-B^{1/2}\| \leq \frac{\|A-B\|}{\lambda_{\text{min}}(A^{1/2})+\lambda_{\text{min}}(B^{1/2})}\]
from Lemma~\ref{lambdamin} to $A = \C_{jj}$ and $B = \Vrightno$,
we get
\[\|\C_{jj}^{1/2}-(\Vrightno)^{1/2}\| \leq \frac{\epsilon}{0+\epsilon^{1/2}} = \epsilon^{1/2}.\]
It follows that
\begin{align*}
\|\C_{jj}^{1/2}\Vright - \I\| &= \|[\C_{jj}^{1/2}-(\Vrightno)^{1/2}]\Vright\| \\
&\leq \|\C_{jj}^{1/2}-(\Vrightno)^{1/2}\|\|\Vright\| \\
&\leq \epsilon^{1/2}\cdot\epsilon^{-1/2} \\
&= 1.
\end{align*}
By triangle inequality, we have
\[\|\C_{jj}^{1/2}\Vright\| \leq 2.\]
Since $\Vright\C_{jj}^{1/2}$ is the adjoint of $\C_{jj}^{1/2}\Vright$, we immediately get $\|\Vright\C_{jj}^{1/2}\| = \|\C_{jj}^{1/2}\Vright\| \leq 2$.  From $\|\bV\| \leq 1$, we get
\begin{align*}
&\|\Vleft\Vmid\Vright\| \\
&= \|\Vleft\C_{ii}^{1/2}\bV_{ij}\C_{jj}^{1/2}\Vright\| \\
&\leq \|\Vleft\C_{ii}^{1/2}\|\|\bV_{ij}\|\|\C_{jj}^{1/2}\Vright\| \\
&\leq 4.
\end{align*}
\end{proof}

% Lemma
\begin{Lemma}\label{condEmp}
Conditioned on the event 
\[E_n = \bigg\{\hatC_{ii}+\epsilon_n\bP_i^1\succeq\frac{\epsilon_n}{2}\I,\ \hatC_{jj}+\epsilon_n\bP_j^1\succeq\frac{\epsilon_n}{2}\I\bigg\},\]
we have
\begin{align*}
&\|(\hatC_{jj}+\epsilon_n\bP_j^1)^{-1/2}(\hatC_{jj})^{1/2}\| = \|(\hatC_{jj})^{1/2}(\hatC_{jj}+\epsilon_n\bP_j^1)^{-1/2}\| \leq 3,\\
&\|(\hatC_{ii}+\epsilon_n\bP_i^1)^{-1/2}\hatC_{ij}(\hatC_{jj}+\epsilon_n\bP_j^1)^{-1/2}\| \leq 9.
\end{align*}
\end{Lemma}
\begin{proof}
The proof is essentially the same as that in Lemma~\ref{facts} except that the event
\[\big\{\Vleftno\succeq\epsilon\I \text{ and }\Vrightno\succeq\epsilon\I\big\}\]
is now replaced by $E_n$.
\end{proof}

% Lemma
\begin{Lemma}\label{boundprobLemma}
Let $\epsilon_n$ be a sequence of positive numbers such that 
\[\lim_{n\rightarrow\infty}\epsilon_n = 0, \qquad \lim_{n\rightarrow\infty}\frac{n^{-1/2}}{\epsilon_n} = 0.\] 
Then, under Assumptions~\ref{assumeRKHS} and \ref{invertible},
\begin{align*}
& P\bigg(\hatC_{ii}+\epsilon_n\bP_i^1\succeq\frac{\epsilon_n}{2}\I,\ \hatC_{jj}+\epsilon_n\bP_j^1\succeq\frac{\epsilon_n}{2}\I\bigg) \geq 1-\delta, \\
& P\bigg(\hatC_{jj}+\epsilon_n\bP_j^1\succeq\frac{\epsilon_n}{2}\I\quad\text{for }1\leq j\leq p\bigg) \geq 1-\delta',
\end{align*}
where $\delta = 2(d_i+d_j)\epsilon^{-1}n^{-1/2}$, $\delta' = 2dp\epsilon^{-1}n^{-1/2}$, and $d_i, d_j, d$ are constants that do not depend on $n$ when $n$ is sufficiently large. 
\end{Lemma}
\begin{proof}
To simplify notation, we write $\epsilon$ for $\epsilon_n$ and $\hat{\C}_{jj}$ for $\hatC_{jj}$.  Then by Lemma~\ref{HS},
\begin{align}
P(\|\Chat_{jj}-\C_{jj}\| > \epsilon) &\leq P(\|\hat{\C}_{jj}-\C_{jj}\|_{\text{HS}} > \epsilon) \nonumber\\
&\leq \epsilon^{-1}E\|\hat{\C}_{jj}-\C_{jj}\|_{\text{HS}} \nonumber\\
&\leq d_j\epsilon^{-1}n^{-1/2}, \label{boundprob}
\end{align}
where $d_j$ is some constant that does not depend on $n$ when $n$ is sufficiently large.

By Lemma~\ref{posdef}, under Assumptions~\ref{assumeRKHS} and \ref{invertible}, we have $\C_{jj}-\epsilon\bP_j^0\succeq \0$, so 
\begin{align*}
\Vhatrightno &= \bigg[\Chat_{jj}-\epsilon\bP_j^0+\frac{\epsilon}{2}\I\bigg] + \frac{\epsilon}{2}\I \\
&= \bigg[\hat{\C}_{jj}-\C_{jj}+\frac{\epsilon}{2}\I\bigg] + \bigg[\C_{jj}-\epsilon\bP_j^0\bigg] + \frac{\epsilon}{2}\I \\
&\succeq \bigg[\Chat_{jj} - \C_{jj} + \frac{\epsilon}{2}\I\bigg] + \frac{\epsilon}{2}\I.
\end{align*}
It follows that
\begin{align*}
P\bigg(\Vhatrightno\succeq\frac{\epsilon}{2}\I\bigg) &\geq P\bigg(\hat{\C}_{jj}-\C_{jj}+\frac{\epsilon}{2}\I \succeq \0\bigg) \\
&= P\bigg(\C_{jj}-\hat{\C}_{jj} \preceq \frac{\epsilon}{2}\I\bigg) \\
&\geq P\bigg(\|\C_{jj}-\hat{\C}_{jj}\| \leq \frac{\epsilon}{2}\bigg) \\
&\geq 1 - 2d_j\epsilon^{-1}n^{-1/2},
\end{align*}
the last inequality is due to \eqref{boundprob}.  Hence
\begin{align*}
&P\bigg(\Vhatleftno\succeq\frac{\epsilon}{2}\I,\ \Vhatrightno\succeq\frac{\epsilon}{2}\I\bigg) \\
&= 1 - P\bigg(\Vhatleftno\preceq\frac{\epsilon}{2}\I\quad\text{or}\quad\Vhatrightno\preceq\frac{\epsilon}{2}\I\bigg) \\
&\geq 1 - P\bigg(\Vhatleftno\preceq\frac{\epsilon}{2}\I\bigg) - P\bigg(\Vhatrightno\preceq\frac{\epsilon}{2}\I\bigg) \\
&\geq 1 - \delta, 
\end{align*}
where $\delta = 2(d_i+d_j)\epsilon^{-1}n^{-1/2}$.  Let $d = \max\{d_1, \ldots, d_p\}$ and apply union bound yields $\delta' = 2dp\epsilon^{-1}n^{-1/2}$.
\end{proof}

% Lemma
\begin{Lemma} \label{invsqrtform}
Suppose that Assumptions~\ref{assumeRKHS} and \ref{invertible} hold.  Let $f$ be an eigenfunction of $\C_{ii}$ corresponding to a nonzero eigenvalue $\lambda$, with the unique decomposition $f = f^0+f^1$, where $f^0\in\Hsp_i^0$ and $f^1\in\Hsp_i^1$.  Then, for $\epsilon$ sufficiently small,
\begin{equation}
\Vleft f = \lambda^{-1} (\Vleftno)^{1/2}f - \lambda^{-1}\epsilon\Vleft f^1. \label{invsqrt}
\end{equation}
\end{Lemma}
\begin{proof}
We have
\[(\Vleftno) f = \lambda f + \epsilon f^1.\]
Applying $(\Vleftno)^{-1}$ on both sides and rearranging terms, we get
\[(\Vleftno)^{-1}f = \lambda^{-1}f - \lambda^{-1}\epsilon(\Vleftno)^{-1}f^1.\]
Applying $(\Vleftno)^{1/2}$ on both sides, we get
\[\Vleft f = \lambda^{-1} (\Vleftno)^{1/2}f - \lambda^{-1}\epsilon\Vleft f^1.\]
\end{proof}

% Proof of the Main Lemmas
\subsection{Proofs of Main Lemmas} \label{mainLemmaProof}
We are now ready to prove the two key lemmas given in Section~\ref{Consistency}.
% Proof of Sample Error Lemma
\subsection*{Proof of Lemma~\ref{SampleError}}
\begin{proof}
Throughout, we condition on the event 
\begin{equation}
E_n = \bigg\{\hatC_{ii}+\epsilon_n\bP_i^1\succeq\frac{\epsilon_n}{2}\I,\ \hatC_{jj}+\epsilon_n\bP_j^1\succeq\frac{\epsilon_n}{2}\I\bigg\}. \label{En}
\end{equation}
It follows from Lemma~\ref{boundprobLemma} that $P(E_n) \geq 1-\delta$ with $\delta = 2(d_i+d_j)\epsilon_n^{-1}n^{-1/2}$, and $d_i, d_j$ are constants that do not depend on $n$ when $n$ is sufficiently large. .

To simplify notation, we write $\epsilon$ for $\epsilon_n$ and $\hat{\C}_{ij}$ for $\hat{\C}_{ij}^{(n)}$ hereinafter.  First, note that the operator on the left hand side of \eqref{hVtV} can be decomposed as 
\begin{align}
&\hV_{ij}-\tV_{ij} \nonumber\\
&= \Vhatleft\Vhatmid\Vhatright \nonumber\\
&\qquad- \Vleft\Vmid\Vright \nonumber\\
&= \{\Vhatleft-\Vleft\}\Vhatmid\Vhatright \label{part1}\\
&\qquad + \Vleft\{\Vhatmid-\Vmid\}\Vhatright \label{part2}\\
&\qquad + \Vleft\Vmid\{\Vhatright-\Vright\}. \label{part3}
\end{align}
We now bound each term on the right hand side of the equation above separately.  From the equality
\begin{align*}
A^{-1/2} - B^{-1/2} &= A^{-1/2}(B^{3/2}-A^{3/2})B^{-3/2}+(A-B)B^{-3/2} \\
&= [A^{-1/2}(B^{3/2}-A^{3/2})+A-B]B^{-3/2},
\end{align*}
plug in $A = \Vleftno$, $B = \Vhatleftno$, and \eqref{part1} becomes
\begin{align}
&- (A^{-1/2}-B^{-1/2})\Vhatmid\Vhatright \nonumber\\
&= -[A^{-1/2}(B^{3/2}-A^{3/2})+A-B]B^{-3/2}\Vhatmid\Vhatright \nonumber\\
&= -[A^{-1/2}(B^{3/2}-A^{3/2})+A-B]\cdot(\Vhatleftno)^{-3/2}\Vhatmid\Vhatright. \label{prodsample}
\end{align}
Now, we bound each term in the product in \eqref{prodsample}.  Under Assumptions~\ref{assumeRKHS} and \ref{invertible}, by Lemma~\ref{facts}, we have $\|\Vleft\| \leq \epsilon^{-1/2}$, so
\begin{align}
&\|A^{-1/2}(B^{3/2}-A^{3/2})+A-B\| \nonumber\\
&= \|\Vleft\{(\Vhatleftno)^{3/2}- (\Vleftno)^{3/2}\} + \C_{ii}-\hat{\C}_{ii}\| \nonumber\\
&\leq \|\Vleft\|\cdot\|(\Vhatleftno)^{3/2}- (\Vleftno)^{3/2}\| + \|\C_{ii}-\hat{\C}_{ii}\| \nonumber\\
&\leq \epsilon^{-1/2}\cdot 3\lambda^{1/2}\|\C_{ii}-\hat{\C}_{ii}\| + \|\C_{ii}-\hat{\C}_{ii}\| \nonumber\\
&\qquad\text{where }\lambda = \max\{\|\Vhatleftno\|, \|\Vleftno\|\}, \text{ by Lemma~\ref{Fukumizu8}} \nonumber\\
&= \bigg(3\lambda^{1/2}\epsilon^{-1/2}+1\bigg)\|\C_{ii}-\hat{\C}_{ii}\| \nonumber\\
&= O(\epsilon^{-1/2}n^{-1/2}), \qquad\text{by Lemma~\ref{HS}}. \label{part1a}
\end{align} 
To bound the other term in \eqref{prodsample}, it follows from Lemma~\ref{condEmp} that conditioned on \eqref{En}, we have $\|(\Vhatleftno)^{-1}\| \leq 2\epsilon^{-1}$ and $\|(\Vhatleftno)^{-1/2}\Vhatmid\Vhatright\| \leq 9$.  Therefore,
\begin{align}
&\|(\Vhatleftno)^{-3/2}\Vhatmid\Vhatright\| \nonumber\\
&= \|(\Vhatleftno)^{-1}\|\|(\Vhatleftno)^{-1/2}\Vhatmid\Vhatright\| \nonumber\\
&\leq 18\epsilon^{-1}, \label{part1b}
\end{align}
Combining \eqref{part1a} and \eqref{part1b}, we see that
the norm of \eqref{part1} is of order $O(\epsilon^{-3/2}n^{-1/2})$.  A similar bound can be obtained for \eqref{part3}.  An upper bound for \eqref{part2} is given by $2\epsilon^{-1}\|\Vmid-\Vhatmid\| = O(\epsilon^{-1}n^{-1/2})$.  
\end{proof}

We now turn to the proof of Lemma~\ref{ApproxError}.
% Proof of Approximation Error Lemma
\subsection*{Proof of Lemma~\ref{ApproxError}}
\begin{proof}
To simplify notation, we write $\epsilon$ for $\epsilon_n$.  First, we rewrite the left hand side of \eqref{approx} as 
\begin{align}
\|\tV_{ij} - \bV_{ij}\| &= \|\Vleft\Vmid\Vright - \C_{ii}^{-1/2}\C_{ij}\Vmid^{-1/2}\| \nonumber\\
&\leq \|\{\Vleft - \C_{ii}^{-1/2}\}\C_{ij}\Vright\| \label{approx1}\\
&\qquad + \|\C_{ii}^{-1/2}\Vmid\{\Vright - \C_{jj}^{-1/2}\}\|. \label{approx2}
\end{align}
Applying Lemma~\ref{facts}, an upper bound for \eqref{approx1} is given by
\begin{align*}
&\|\{\Vleft - \C_{ii}^{-1/2}\}\C_{ij}\Vright\| \\
&= \|\{\Vleft - \C_{ii}^{-1/2}\}\C_{ii}^{1/2}\bV_{ij}\C_{jj}^{1/2}\Vright\| \\
&\leq 2\|\{\Vleft - \C_{ii}^{-1/2}\}\C_{ii}^{1/2}\bV_{ij}\| \\
&= 2\|\{\Vleft\C_{ii}^{1/2} - \I\}\bV_{ij}\|.
\end{align*}
Our goal is to show that
\begin{equation}
\|\{\Vleft\C_{ii}^{1/2} - \I\}\bV_{ij}\| \longrightarrow 0\text{ as }n\longrightarrow\infty. \label{showthis}
\end{equation}
Let $f$ be a unit eigenfunction of $\C_{ii}$ corresponding to a nonzero eigenvalue $\lambda$.  It follows from Lemma~\ref{invsqrtform} that
\begin{align}
&\|\{\Vleft\C_{ii}^{1/2} - \I\}\C_{ii}f\|_{k_i} \nonumber\\
&= \lambda \|\{\Vleft\C_{ii}^{1/2} - \I\}f\|_{k_i} \nonumber\\
&= \lambda \bigg\|\lambda^{1/2}\bigg\{\lambda^{-1} (\Vleftno)^{1/2}f - \lambda^{-1}\epsilon\Vleft f^1\bigg\} - f\bigg\|_{k_i} \nonumber\\
&\leq \lambda\bigg[\|\lambda^{-1/2}(\Vleftno)^{1/2}f - f\|_{k_i} + \epsilon\lambda^{-1/2}\|\Vleft f^1\|_{k_i}\bigg] \nonumber\\
&= \lambda\|\lambda^{-1/2}(\Vleftno)^{1/2}f - f\|_{k_i} + \epsilon\lambda^{1/2}\|\Vleft f^1\|_{k_i}. \label{approxbound}
\end{align}
To bound the first term in \eqref{approxbound}, note that from Lemma~\ref{facts},
\begin{align*}
|\langle f, (\Vleftno)^{1/2}f\rangle_{k_i} - \langle f, \C_{ii}^{1/2}f\rangle_{k_i}| &= |\langle f, [(\Vleftno)^{1/2}-\C_{ii}^{1/2}]f\rangle_{k_i}| \\
&\leq \|(\Vleftno)^{1/2}-\C_{ii}^{1/2}\| \\
&\leq \epsilon^{1/2}, 
\end{align*}
so we have
\begin{align*}
&\|\lambda^{-1/2}(\Vleftno)^{1/2}f - f\|_{k_i}^2 \\
&= \lambda^{-1}\langle f, (\Vleftno)f\rangle_{k_i} -2\lambda^{-1/2}\langle f, (\Vleftno)^{1/2}f\rangle_{k_i} + \|f\|_{k_i}^2 \\
&\leq \lambda^{-1}(\lambda + \epsilon\|f^1\|^2)-2\lambda^{-1/2}[\langle f, \C_{ii}^{1/2}f\rangle_{k_i}-\epsilon^{1/2}] + 1 \\
&= 2 + \lambda^{-1}\epsilon\|f^1\|^2 -2\lambda^{-1/2}\cdot\lambda^{1/2} + 2\lambda^{-1/2}\epsilon^{1/2}\\
&\leq \lambda^{-1}\epsilon + 2\lambda^{-1/2}\epsilon^{1/2},
\end{align*}
which gives
\begin{align*}
\|\lambda^{-1/2}(\Vleftno)^{1/2}f - f\|_{k_i} &\leq [\lambda^{-1}\epsilon + 2\lambda^{-1/2}\epsilon^{1/2}]^{1/2} \\
&\leq \lambda^{-1/2}\epsilon^{1/2} + \sqrt{2}\lambda^{-1/4}\epsilon^{1/4} \\
&\leq \lambda^{-1/2}\epsilon^{1/2} + 2\lambda^{-1/4}\epsilon^{1/4}.
\end{align*}
It follows that an upper bound for the expression on the left hand side of \eqref{approxbound} is
\begin{align}
\|\{\Vleft\C_{ii}^{1/2} - \I\}\C_{ii}f\|_{k_i} &\leq \lambda\cdot(\lambda^{-1/2}\epsilon^{1/2} + 2\lambda^{-1/4}\epsilon^{1/4}) + \epsilon\lambda^{1/2}\cdot\epsilon^{-1/2}\nonumber\\
&= 2(\lambda\epsilon)^{1/2} + 2\lambda^{3/4}\epsilon^{1/4}. \label{approxbound2}
\end{align}
As pointed out in Section~\ref{mainAssumptions}, $\mR(\bV_{ij}) \subseteq \overline{\mR(\C_{ii}})$.  Now, let $v$ be an arbitrary element in $\mR(\bV_{ij})\cap\mR(\C_{ii})$, so that there exists $u\in\Hsp_i$ such that $v = \C_{ii}u$.  We then have $u = \sum_{\ell=1}^\infty\langle u, f_\ell\rangle_{k_i}f_\ell$, where $f_\ell$ is a unit eigenvector of $\C_{ii}$ corresponding to nonzero eigenvalue $\lambda_\ell$ (i.e.,\! $\{f_\ell\}$ and the eigenfunctions corresponding to zero eigenvalue forms an orthonormal basis system of $\Hsp_i$), and
\begin{align*}
&\|\{\Vleft\C_{ii}^{1/2} - \I\}v\|_{k_i} \\
&= \|\{\Vleft\C_{ii}^{1/2} - \I\}\C_{ii}u\|_{k_i} \\
&\leq \sum_{\ell=1}^\infty |\langle u, f_\ell\rangle_{k_i}|\|\{\Vleft\C_{ii}^{1/2} - \I\}\C_{ii}f_\ell\|_{k_i} \\
&\leq \sum_{\ell=1}^\infty |\langle u, f_\ell\rangle_{k_i}| \cdot [2(\lambda_\ell\epsilon)^{1/2}+2\lambda_\ell^{3/4}\epsilon^{1/4}] \\
&\leq \bigg(\sum_{\ell=1}^\infty \langle u, f_\ell\rangle_{k_i}^2\bigg)^{1/2}\bigg(\sum_{\ell=1}^\infty4\lambda_\ell\epsilon\bigg)^{1/2} \\
&\qquad+ 2\epsilon^{1/4}\bigg\{\sum_{\ell: \lambda_\ell\geq 1} |\langle u, f_\ell\rangle_{k_i}|\lambda_\ell + \sum_{\ell: 0<\lambda_\ell<1} |\langle u, f_\ell\rangle_{k_i}|\lambda_\ell^{1/2} \bigg\}\\
&\leq 2\epsilon^{1/2}\|u\|_{k_i}\bigg(\sum_{\ell=1}^\infty\lambda_\ell\bigg)^{1/2} + 2\epsilon^{1/4}\bigg\{\|u\|_{k_i}\sum_{\ell: \lambda_\ell\geq 1} \lambda_\ell + \|u\|_{k_i}\bigg(\sum_{\ell: 0<\lambda_\ell<1} \lambda_\ell\bigg)^{1/2} \bigg\} \\
&= O(\epsilon^{1/4}),
\end{align*}
where the second inequality follows from \eqref{approxbound2}, the third and fourth inequality follows from Cauchy-Schwarz inequality, and the last equality follows from the fact that $\C_{ii}$ is trace class and hence $\sum_{\ell=1}^\infty\lambda_\ell<\infty$.  We therefore conclude that
\begin{equation}
\{\Vleft\C_{ii}^{1/2} - \I\}v \longrightarrow 0\text{ for all }v\in\mR(\bV_{ij})\cap\mR(\C_{ii})\text{ as }n\longrightarrow\infty. \label{gotozero}
\end{equation}
Coupled with that fact that $\bV_{ij}$ is compact, Lemma~\ref{Fukumizu9} implies that \eqref{showthis} holds. 

To show the convergence of \eqref{approx2}, first note that we can rewrite \eqref{approx2} as
\begin{align*}
\|\C_{ii}^{-1/2}\Vmid\{\Vright - \C_{jj}^{-1/2}\}\| &= \|\bV_{ij}\C_{jj}^{1/2}\{\Vright - \C_{jj}^{-1/2}\}\| \\
&= \|\bV_{ij}\{\C_{jj}^{1/2}\Vright - \I\}\|.
\end{align*}
Using the fact that $\Vright$ is self-adjoint, we get
\begin{align*}
\|\bV_{ij}\{\C_{jj}^{1/2}\Vright - \I\}\| &= \|(\bV_{ij}\{\C_{jj}^{1/2}\Vright - \I\})^*\| \\
&= \|\{\Vright\C_{jj}^{1/2} - \I\}\bV_{ij}^*\|.
\end{align*}
Since $\bV_{ij}$ is compact implies that $\bV_{ij}^*$ is compact, it follows similarly from our proof for \eqref{gotozero} and an application of Lemma~\ref{Fukumizu9} that 
\[\|\{\Vright\C_{jj}^{1/2} - \I\}\bV_{ij}^*\| \longrightarrow 0.\] 

\end{proof}

% Proof of the Main Theorems
\subsection{Proofs of Main Theorems} 
\label{mainThmProof}
% Proof of Theorem 1
\subsection*{Proof of Theorem~\ref{MainThm1}}
\begin{proof}
It follows from Lemmas~\ref{SampleError} and \ref{ApproxError} that we have
\[\|\hV_{ij}-\bV_{ij}\| \longrightarrow 0\]
in probability when $n\longrightarrow\infty$.  From the expression in \eqref{Vmatrix}, we see that 
\[\bV = \sum_{i=1}^p\sum_{j\neq i}\bP_i\bV\bP_j + \I,\]
where $\bP_j$ is the orthogonal projection from $\bHsp$ onto $\Hsp_j$ and $\I:\bHsp\longrightarrow\bHsp$ is the identity operator.  Since there is a one-to-one correspondence between $\bV_{ij}$ and $\bP_i\bV\bP_j$ (for $i\neq j$), using a similar decomposition for $\hV$, we get
\begin{align*}
\|\hV - \bV\| &\leq \sum_{i=1}^p\sum_{j\neq i}\|\hV_{ij}-\bV_{ij}\| \longrightarrow 0
\end{align*}
in probability as $n\longrightarrow\infty$.  It then follows from Lemma~\ref{evecconv} that
\[|\langle\hf, \f^*\rangle_k| \longrightarrow 1\]
in probability as $n\longrightarrow\infty$.
\end{proof}

% Proof of Theorem 2
\subsection*{Proof of Theorem~\ref{MainThm2}}
\begin{proof}
Throughout, we condition on the event 
\begin{equation}
F_n = \bigg\{\hatC_{jj}+\epsilon_n\bP_j^1\succeq\frac{\epsilon_n}{2}\I\quad\text{for }1\leq j\leq p\bigg\},
\end{equation}
where $P(F_n)\geq 1-\delta$ with $\delta = 2dp\epsilon_n^{-1}n^{-1/2}$ from Lemma~\ref{boundprobLemma}, and $d$ is a constant that does not depend on $n$ when $n$ is sufficiently large. 

To simplify notation, we write $\epsilon$ for $\epsilon_n$ and $\hat{\C}_{ij}$ for $\hat{\C}_{ij}^{(n)}$.  Since 
\[\sum_{j=1}^p\Var[\hphi_j-\phi_j^*] = \sum_{j=1}^p\langle\hphi_j-\phi_j^*, \C_{jj}(\hphi_j-\phi_j^*)\rangle_{k_j} = \sum_{j=1}^p \|\C_{jj}^{1/2}(\hphi_j-\phi_j^*)\|_{k_j}^2,\] 
it suffices to show that $\|\C_{jj}^{1/2}(\hphi_j-\phi_j^*)\|_{k_j}\longrightarrow 0$ in probability as $n\longrightarrow\infty$.  Using the fact that $\hat{f}^{(n)}_ j= (\Chat_{jj}+\epsilon\bP_j^1)^{1/2}\hphi_j$, $f_j^* = \C_{jj}^{1/2}\phi_j^*$, we have
\begin{align}
\|\C_{jj}^{1/2}(\hphi_j-\phi_j^*)\|_{k_j} &= \|\C_{jj}^{1/2}\Vhatright\hat{f}^{(n)}_j-f_j^*\|_{k_j} \nonumber\\
&\leq \|\C_{jj}^{1/2}\{(\Chat_{jj}+\epsilon\bP_j^1)^{-1/2}-(\C_{jj}+\epsilon\bP_j^1)^{-1/2}\}\hat{f}^{(n)}_j\|_{k_j} \label{firstterm}\\
&\qquad + \|\C_{jj}^{1/2}\Vright(\hat{f}^{(n)}_j-f_j^*)\|_{k_j} \label{secondterm} \\
&\qquad + \|\C_{jj}^{1/2}\Vright f_j^*-f_j^*\|_{k_j}. \label{thirdterm}
\end{align}

An upper bound for \eqref{firstterm} can be obtained using similar argument as that in the bound for \eqref{part1}.  First, consider the equality
\begin{align*}
A^{-1/2} - B^{-1/2} &= A^{-3/2}(B^{3/2}-A^{3/2})B^{-1/2}+A^{-3/2}(A-B) \\
&= A^{-3/2}[(B^{3/2}-A^{3/2})B^{-1/2}+A-B],
\end{align*}
and plug in $A = \Vrightno$ and $B = \Vhatrightno$, in which case the operator in \eqref{firstterm} becomes
\begin{align}
& -\C_{jj}^{1/2}(A^{-1/2}-B^{-1/2}) \nonumber\\
&= -\C_{jj}^{1/2}A^{-3/2}[(B^{3/2}-A^{3/2})B^{-1/2}+A-B] \nonumber\\
&= -\C_{jj}^{1/2}(\Vrightno)^{-3/2}[(B^{3/2}-A^{3/2})B^{-1/2}+A-B] \label{prodfinal}
\end{align}
Now we bound each term in the product \eqref{prodfinal}.  Since
\begin{align}
\|\C_{jj}^{1/2}(\Vrightno)^{-3/2}\| &\leq \|\C_{jj}^{1/2}(\Vrightno)^{-1/2}\|\|(\Vrightno)^{-1}\| \nonumber\\
&\leq 2\epsilon^{-1}, \label{firsttermdone1}
\end{align}
and
\begin{align}
&\|(B^{3/2}-A^{3/2})B^{-1/2}+A-B\| \nonumber\\
&= \|\{(\Vhatrightno)^{3/2}-(\Vrightno)^{3/2}\}\Vhatright + \C_{jj}-\hat{\C}_{jj}\| \nonumber\\
&\leq \|(\Vhatrightno)^{3/2}-(\Vrightno)^{3/2}\|\|\Vhatright\| + \|\C_{jj}-\hat{\C}_{jj}\| \nonumber\\
&\leq 3\lambda^{1/2}\|\hat{\C}_{jj}-\C_{jj}\|\cdot\sqrt{2}\epsilon^{-1/2} + \|\C_{jj}-\hat{\C}_{jj}\| \nonumber\\
&\qquad\text{where }\lambda = \max\{\|\Vhatrightno\|, \|\Vrightno\|\}, \text{ by Lemma~\ref{Fukumizu8}} \nonumber \\
&= \bigg(3\sqrt{2}\lambda^{1/2}\epsilon^{-1/2} + 1\bigg)\|\C_{jj}-\hat{\C}_{jj}\| \nonumber \\
&= O(\epsilon^{-1/2}n^{-1/2}), \label{firsttermdone2}
\end{align}
combining \eqref{firsttermdone1} and \eqref{firsttermdone2} gives an upper bound of order $O(\epsilon^{-3/2}n^{-1/2})$ on \eqref{firstterm}.

For \eqref{secondterm}, it follows from Theorem~\ref{MainThm1} that
\begin{align*}
\|\C_{jj}^{1/2}(\C_{jj}+\epsilon\bP_j^1)^{-1/2}(\hat{f}^{(n)}_j-f_j^*)\|_{k_j} &\leq \|\C_{jj}^{1/2}(\C_{jj}+\epsilon\bP_j^1)^{-1/2}\|\|\hat{f}^{(n)}_j-f_j^*\|_{k_j} \\
&\leq 2\|\hat{f}^{(n)}_j-f_j^*\|_{k_j} \longrightarrow 0
\end{align*}
in probability as $n\longrightarrow\infty$.  Finally, for \eqref{thirdterm}, using the fact that $f_j^* = \C_{jj}^{1/2}\phi_j^*$, we get
\begin{align*}
&\|\C_{jj}^{1/2}\Vright f_j^*-f_j^*\|_{k_j} \\
&= \|\C_{jj}^{1/2}\Vright\C_{jj}^{1/2}\phi_j^*-\C_{jj}^{1/2}\phi_j^*\|_{k_j} \\
&= \|\C_{jj}^{1/2}\Vright\{\C_{jj}^{1/2}-(\Vrightno)^{1/2}\}\phi_j^*\|_{k_j} \\
&\leq  \|\C_{jj}^{1/2}\Vright\|\|\C_{jj}^{1/2}-(\Vrightno)^{1/2}\|\|\phi_j^*\|_{k_j} \\
&\leq 2\|\phi_j^*\|_{k_j}\epsilon^{1/2}\longrightarrow 0,
\end{align*}
the second inequality follows from Lemma~\ref{facts}.
\end{proof}

%%%%%%%%%%%%%%%%%%%%%%%%%%%%%%%%%%%%%%%%%%%%%%%%%%%%%

%!TEX root = Paper.tex

\newpage
% Supporting Proof of Theorems in Section~\ref{powerAlg}
\section{Supporting Proofs for Theorems in Section~\ref{compute}}
\label{powerproof}

This section contains proofs for theorems in Section~\ref{compute}.  We present some facts about RKHS in Section~\ref{propRKHS}, followed by proofs of theorems in Section~\ref{proofSecAlg}. 

\subsection{Properties of Reproducing Kernel Hilbert Spaces}
\label{propRKHS}
The following properties of RKHS's will be useful for proving theorems in Section~\ref{compute}.

\begin{Lemma}\label{incont}
Let $\kxx$ be a kernel and $\Hsp$ its associated reproducing kernel Hilbert space. Then $\Hsp\subset\cont$ if and only if $\sup_{x\in\mX}k(x, x) <\infty$ and $k(x, \cdot)$ is continuous on $\mX$ for all $x\in\mX$.  Moreover, the inclusion $\Hsp\hookrightarrow\cont$ is continuous.
\end{Lemma}
\begin{proof}
If $\Hsp\subset\cont$, then $k(x, \cdot)\in\cont$.  Moreover, for each $f\in\Hsp$, we have $|\langle k_x, f\rangle_k| = |f(x)| \leq \|f\|_\infty$ for all $x\in\mX$. The principle of uniform boundedness implies that there exists $M<\infty$ such that $\|k_x\|_k\leq M$ for all $x\in\mX$.  It follows that $\sup_{x\in\mX}k(x, x)\leq M^2<\infty$.

Conversely, assume that $k(x, x)\leq M^2$ and $k(x, \cdot)\in\cont$ for all $x\in\mX$.  Given $f\in\Hsp$, we have
\begin{equation}
|f(x)| = |\langle f, k_x\rangle_k| \leq \|f\|_k\|k_x\|_k = \|f\|_k\sqrt{k(x, x)} \leq M\|f\|_k, \qquad\forall x\in\mX. \label{unifconv}
\end{equation}
Hence, convergence in $\Hsp$ implies uniform convergence, so the closure of $\Span{\{k(x, \cdot): x\in\mX\}}$ (with respect to $\|\cdot\|_k$) is contained in $\cont$, i.e. $\Hsp\subset\cont$.  The continuity of inclusion follows from $\|f\|_\infty\leq M\|f\|_k$.
\end{proof}

\begin{Corollary}\label{corunifconv}
Let $\{f_n\}\subset\Hsp$, then $\|f_n-f\|_k\longrightarrow 0$ for some $f\in\Hsp$ implies that $\|f_n-f\|_\infty\longrightarrow 0$.
\end{Corollary}
\begin{proof}
Immediate from \eqref{unifconv}.
\end{proof}

\begin{Lemma}\label{inL2}
Let $\kxx$ be a kernel and $\Hsp$ its associated reproducing kernel Hilbert space. Let $X$ be a random variable taking values in $\mX$ with induced probability measure $P$ on $\mX$.  Then $E[k(X, X)] < \infty$ implies that $\Hsp\subset L^2(\mX, dP)$, and the inclusion $\Hsp\hookrightarrow L^2(\mX, dP)$ is continuous.
\end{Lemma}
\begin{proof}
The claim follows from
\[E[f^2(X)] = E[\langle f, k(X, \cdot)\rangle_k^2] \leq E[\|f\|^2_k\|k(X, \cdot)\|_k^2] = \|f\|_k^2E[k(X, X)], \qquad\forall f\in\Hsp.\]
\end{proof}

The readers are referred to \cite{ReedSimon}, \cite{Scholkopf2002} and \cite{Aronszajn1950} for more details on functional analysis, kernel methods and reproducing kernels, respectively.

\subsection{Proofs}
\label{proofSecAlg}
When Assumption~\ref{assumeRKHS} holds and $\Hsp_j^0$ excludes constants, for $f\in\Hsp_j^0$, $\|f\|_{p_j}^2 := \Var[f(X_j)] = 0$ if and only if $f\equiv 0$.  By Remark~\ref{RemarkConstructRKHS}, this implies that $\Hsp_j^0$ can be turned into an RKHS with inner product $\langle\cdot, \cdot\rangle_{P_j}$.  Without loss of generality, in the following proofs, we define $\langle\cdot, \cdot\rangle_{k_j^0} := \langle\cdot, \cdot\rangle_{P_j}$.

% Proof of H endowed with star norm is an RKHS
\subsection*{Proof of Theorem~\ref{RKHSstar}}
\begin{proof}
To simplify notation, we omit the subscript $j$.  Also, we define the semi-norm
\[\langle f, g\rangle_{P} := \Cov[f(X), g(X)], \qquad f, g\in L^2(\mX, dP),\]
where $X$ is a random variable taking values in $\mX$ with probability distribution $P$.

We first check that $\langle\cdot, \cdot\rangle_\star$ defines an appropriate inner product on $\Hsp$.  Clearly, $\langle\cdot, \cdot\rangle_\star$ is symmetric, bilinear and positive semi-definite.  When $f = f^0+f^1$ with $f^0\in\Hsp^0, f^1\in\Hsp^1$ satisfies $\|f\|_\star = 0$, we have $\|f^1\|_{k^1} = 0$ and $\Var[f(X)] = 0$.  This in turn implies that $f^1 \equiv 0$, and $f^0$ is a constant function almost surely.  Since $\Hsp^0\subset C(\mX)$ and $\Hsp^0$ excludes constant, this means that $f^0 \equiv 0$.  Therefore, $f \equiv 0$ if and only if $\|f\|_\star = 0$.

To see that $\Hsp$ is a Hilbert space with respect to $\langle\cdot, \cdot\rangle_\star$, we need to show that if $\{f_n\}$ is a Cauchy sequence in $\Hsp$ with respect to 
$\|\cdot\|_\star$, then $\{f_n\}$ converges in the norm $\|\cdot\|_\star$ to some $f^*\in\Hsp$.
So suppose that $\{f_n\}$ is Cauchy with respect to $\|\cdot\|_\star$, then it is also Cauchy with respect to $\|\cdot\|_{k^1}$ and $\|\cdot\|_{P}$.  We decompose $f_n = f_n^0 + f_n^1$, with $f_n^0\in\Hsp^0$ and $f_n^1\in\Hsp^1$.  Since $\|f_n\|_{k^1} = \|f_n^1\|_{k^1}$, and $\Hsp^1$ is itself an RKHS with respect to $\langle\cdot, \cdot\rangle_{k^1}$, we see that $\{f_n^1\}$ is a Cauchy sequence in $\Hsp^1$, so there exists a unique $f^{*1}\in\Hsp^1$ such that 
\begin{equation}
\|f_n^1-f^{*1}\|_{k^1} \longrightarrow 0. \label{convf1k}
\end{equation}
On the other hand, by Corollary~\ref{corunifconv} we know that $\|f_n^1-f^{*1}\|_{k^1}\longrightarrow 0$ implies $\|f_n^1-f^{*1}\|_\infty\longrightarrow 0$.  So
\begin{equation}
\|f_n^1-f^{*1}\|_{P} \leq 2\|f_n^1-f^{*1}\|_\infty \longrightarrow 0. \label{convf1}
\end{equation}

Now we consider $\{f_n^0\}$, and we want to show that there exists $f^{*0}\in\Hsp^0$ such that $\|f_n^0-f^{*0}\|_{P}\longrightarrow 0$.  We know that $\{f_n\}$ and $\{f_n^1\}$ are Cauchy with respect to $\|\cdot\|_{P}$, so for every $\epsilon>0$, there exists $N(\epsilon)$ such that if $m\geq n\geq N(\epsilon)$, we have
\[\|f_n^0-f_m^0 + f_n^1-f_m^1\|_{P} =\|f_n-f_m\|_{P} <\frac{\epsilon}{2}, \qquad\text{and}\qquad\|f_n^1-f_m^1\|_{P}<\frac{\epsilon}{2},\]
which implies that 
\[\|f_n^0 - f_m^0\|_{P} \leq \|f_n^1 - f_m^1\|_{P} + \|f_n - f_m\|_{P} < \epsilon,\]
so $\{f_n^0\}$ is also Cauchy with respect to $\|\cdot\|_{P}$.  Since $\Hsp^0$ is finite dimensional and $\langle\cdot, \cdot\rangle_{P}$ induces a (strict) norm on $\Hsp^0$, $\Hsp^0$ is complete with respect to $\|\cdot\|_{P}$, so there exists $f^{*0}\in\Hsp^0$ such that 
\begin{equation}
\|f_n^0-f^{*0}\|_{P}\longrightarrow 0. \label{convf0}
\end{equation}

Finally, let $f^* = f^{*0} + f^{*1}$, we have $f^*\in\Hsp$.  Combining \eqref{convf1k}, \eqref{convf1} and \eqref{convf0}, we see that
\begin{align*}
\|f_n-f^*\|_\star^2 &= \|f_n-f^*\|_{P}^2 + \alpha\|f_n-f^*\|_{k^1}^2 \\
&\leq 2\|f_n^0-f^{*0}\|_{P}^2 + 2\|f_n^1-f^{*1}\|_{P}^2 + \alpha\|f_n^1-f^{*1}\|_{k^1}^2 \longrightarrow 0.
\end{align*}
Since every Cauchy sequence converges in $(\Hsp, \langle\cdot, \cdot\rangle_\star)$, we see that $\Hsp$ is a Hilbert space with respect to the inner product $\langle\cdot, \cdot\rangle_\star$.

To check the reproducing property of $\Hsp$ with respect to $\langle\cdot, \cdot\rangle_\star$, we need to show that the evaluation functionals $\delta_x(f) = f(x)$ are bounded for all $x\in\mX$.  Suppose $f = f^0 + f^1, f^0\in\Hsp^0, f^1\in\Hsp^1$.  Under Assumption~\ref{assumeRKHS},
\begin{equation}
|f^1(x)| = |\langle f^1, k^1_x\rangle_{k^1}| \leq \|k_x^1\|_{k^1}\|f^1\|_{k^1} \leq \sqrt{k^1(x, x)}\frac{1}{\sqrt{\alpha}}\|f\|_\star \leq M\|f\|_\star, \label{f1bound}
\end{equation}
where $M = \sup_{x\in\mX}\sqrt{k^1(x, x)}\frac{1}{\sqrt{\alpha}}<\infty$.  On the other hand, since $\langle\cdot, \cdot\rangle_{k^0} := \langle\cdot, \cdot\rangle_{P}$ and $\|f^1\|_{P}\leq 2\|f^1\|_\infty\leq 2M\|f\|_\star$, we have
\begin{align*}
|f^0(x)| &=  |\langle f^0, k^0_x\rangle_{k^0}| \leq \|k_x^0\|_{k^0}\|f^0\|_{k^0} = \|k_x^0\|_{P}\|f^0\|_{P} \\
&\leq \|k_x^0\|_{P}(\|f\|_{P} + \|f^1\|_{P}) \leq \|k_x^0\|_{P}(\|f\|_\star + 2M\|f\|_\star) = d_x\|f\|_\star
\end{align*}
where $d_x = \|k_x^0\|_{P}(2M+1)$.  So $|f(x)| \leq |f^0(x)|+|f^1(x)| \leq C_x\|f\|_\star$ for some constant $C_x<\infty$ for all $x\in\mX$.  This then implies that $\Hsp$ is an RKHS with respect to $\langle\cdot, \cdot\rangle_\star$.  
\end{proof}

% Proof of S_{ij} being well-defined
\subsection*{Proof of Theorem~\ref{milestone}}
\begin{proof}
Throughout, it is understood that $\Hsp_j\subset L^2(\mX_j, dP_j)$ is endowed with $\langle\cdot, \cdot\rangle_{\star_j}$, for $1\leq j\leq p$.  To see that $\bS_{ij}$ is well-defined, we need to show the existence and uniqueness of the solution of the ``generalized" regularized population regression problem.  First, note that given $\phi_j\in\Hsp_j$, the operator $\Cov[\phi_j(X_j), \cdot(X_i)] : \Hsp_i\longrightarrow\R$ is a bounded linear functional on $\Hsp_i$.  By the Riesz Representation Theorem, there exists a unique $h\in\Hsp_i$ such that $\Cov[\phi_j(X_j), f(X_i)] = \langle h, f\rangle_{\star_i}$ for all $f\in\Hsp_i$.  It then follows that
\begin{align*}
&\argmin_{f\in\Hsp_i} \left\{\Var[\phi_j(X_j)-f(X_i)] + \alpha_i\|f\|_{k_i^1}^2\right\}  \\
&= \argmin_{f\in\Hsp_i} \left\{-2\Cov[\phi_j(X_j), f(X_i)] + \Var[f(X_i)] + \alpha_i\|f\|_{k_i^1}^2\right\} \\
&= \argmin_{f\in\Hsp_i} \left\{-2\langle h, f\rangle_{\star_i} + \|f\|_{\star_i}^2\right\} \\
&= h.
\end{align*}
That is, we have $\bS_{ij}\phi_j = h$, where $h$ is unique and satisfies $\Cov[\phi_j(X_j), f(X_i)] = \langle h, f\rangle_{\star_i}$ for all $f\in\Hsp_i$.  Equivalently, 
\[\Cov[\phi_i(X_i), \phi_j(X_j)] = \langle \phi_i, \bS_{ij}\phi_j\rangle_{\star_i}, \qquad\forall\phi_i\in\Hsp_i, \phi_j\in\Hsp_j.\]
To check the properties of $\bS_{ij}$, we first note that we can decompose the operation of $\bS_{ij}$ as follows:
\begin{flalign*}
\bS_{ij}: & \ \Hsp_j 
&& \hspace{-2cm}\stackrel{\I}\longrightarrow L^2(\mX_j, dP_j) 
&& \hspace{-2cm}\stackrel{\bT}\longrightarrow \Hsp_i^* 
&& \hspace{-2cm}\stackrel{\bR}\longrightarrow \Hsp_i, \\
& \ \phi_j 
&& \hspace{-2cm}\stackrel{\I}\longmapsto  \phi_j 
&& \hspace{-2cm}\stackrel{\bT}\longmapsto \Cov[\phi_j(X_j), \cdot(X_i)] 
&& \hspace{-2cm}\stackrel{\bR}\longmapsto h.
\end{flalign*}
Here $\Hsp_i^*$ denotes the dual space of $\Hsp_i$, and consists of bounded linear functionals defined on $\Hsp_i$.  On the other hand, $\I: \Hsp_j\longrightarrow L^2(\mX_j, dP_j)$ denotes the inclusion of $\Hsp_j$ into $L^2(\mX_j, dP_j)$, $\bT: L^2(\mX_j, dP_j)\longrightarrow\Hsp_i^*$ denotes the one-to-one correspondence between the function $\phi_j\in L^2(\mX_j, dP_j)$ (in fact, its corresponding equivalence class with respect to the squared-norm $\Var[\cdot(X_j)]$) and the bounded linear functional $\Cov[\phi_j(X_j), \cdot(X_i)] : \Hsp_i\longrightarrow\R$, and $\bR: \Hsp_i^*\longrightarrow\Hsp_i$ is the isomorphism between $\Cov[\phi_j(X_j), \cdot(X_i)]\in\Hsp_i^*$ and the function $h\in\Hsp_i$.  Therefore, we have $\bS_{ij} = \bR\bT\I$.
\begin{itemize}
\item {\bf Linearity and boundedness:} \\
$\bS_{ij}: \Hsp_j\longrightarrow\Hsp_i$ is a bounded linear operator as long as each of $\bR, \bT, \I$ is.  That $\bR$ and $\I$ are bounded and linear follows immediately.  To check the linearity of $\bT$, for all $f\in\Hsp_i$ and $a, b\in\R$, we have 
\begin{align*}
\bT(a\phi_j+b\psi_j)(f) &= \Cov[a\phi_j(X_j)+b\psi_j(X_j), f(X_i)] \\
&= a\Cov[\phi_j(X_j), f(X_i)] + b\Cov[\psi_j(X_j), f(X_i)] \\
&= a\bT(\phi_j)(f) + b\bT(\psi_j)(f),
\end{align*}
so linearity follows.  We now check that $\bT$ is bounded with operator norm less than or equal to one.
\begin{align*}
\|\bT\| &= \sup_{\phi_j: \Var[\phi_j(X_j)]\leq 1}\|\bT(\phi_j)\| \\
&= \sup_{\phi_j: \Var[\phi_j(X_j)]\leq 1}\ \sup_{f:\|f\|_{\star_i}\leq 1}|\Cov[\phi_j(X_j), f(X_i)]| \\
&\leq \sup_{\phi_j: \Var[\phi_j(X_j)]\leq 1}\ \sup_{f:\|f\|_{\star_i}\leq 1} \big(\Var[\phi_j(X_j)]\Var[f(X_i)]\big)^{1/2} \\
&\leq \sup_{\phi_j: \Var[\phi_j(X_j)]\leq 1}\ \sup_{f:\|f\|_{\star_i}\leq 1}\big(\Var[\phi_j(X_j)]\big)^{1/2}\|f\|_{\star_i} \\
&=1.
\end{align*}
In fact, in the case when $i=j$, the operator norm of $\bT$ is exactly equal to one provided $\Dim(\Hsp_j^0)>0$, since we can pick $\phi_j\in\Hsp_j^0$ with $\Var[\phi_j(X_j)]=1$, in which case $\|\phi_j\|_{\star_j} = (\Var[\phi_j(X_j)])^{1/2} = 1$ and both inequalities above become equalities.

\item {\bf Compactness:} \\
Using the fact that the product of a bounded linear operator and a compact operator is compact (see, e.g. \cite{ReedSimon}, Theorem VI.12(c)), it suffices to show that one of $\bR, \bT$ or $\I$ is compact.  We are off to show the compactness of $\I$, which requires that every bounded sequence $\{f_n\}$ in $(\Hsp_j, \langle\cdot, \cdot\rangle_{\star_j})$ has a convergent subsequence in $L^2(\mX_j, dP_j)$ (endowed with squared-norm $\Var[\cdot(X_j)]$).  

To simplify notation, we omit the subscript $j$.  Also, we define the semi-norm
\[\langle f, g\rangle_{P} := \Cov[f(X), g(X)], \qquad f, g\in L^2(\mX, dP),\]
where $X$ is a random variable taking values in $\mX$ with probability distribution $P$.

This proof idea is to establish that $\{f_n\}\subset\Hsp$ is uniformly bounded and equicontinuous and then invoke the Arzel\`{a}-Ascoli Theorem.  We first recall that $\Hsp$ is also an RKHS with respect to $\langle\cdot, \cdot\rangle_k$.  
If $\|f_n\|_{\star}\leq B$, then by definition we have
\[\|f_n^1\|_{k^1} \leq \frac{\|f_n\|_{\star}}{\sqrt{\alpha}}\leq \frac{B}{\sqrt{\alpha}}\]
and
\begin{align*}
\|f_n^0\|_{k^0} &= \|f_n^0\|_{P} \leq \|f_n^1\|_{P} + \|f_n\|_{P} \leq 2\|f_n^1\|_\infty + \|f_n\|_\star \\
&\leq 2M\|f_n\|_\star + \|f_n\|_\star \leq (2M+1)B,
\end{align*}
where the second to the last inequality follows from \eqref{f1bound}. So 
\[\|f_n\|_{k} = \|f_n^0 + f_n^1\|_k \leq \|f_n^0\|_{k^0} + \|f_n^1\|_{k^1} \leq C,\]
where $C = (2M+1)B+B/\sqrt{\alpha} < \infty$.  

Under Assumption~\ref{assumeRKHS}, $k$ is uniformly continuous on $\mX\times\mX$ and $\sup_{x\in\mX}k(x, x)<\infty$. Since
\[|f_n(x)| = |\langle f_n, k_x\rangle_k| \leq \|f_n\|_k\|k_x\|_k = \|f_n\|_k\sqrt{k(x, x)},\]
and
\[\|f_n\|_\infty \leq  \|f_n\|_k\sqrt{\sup_{x\in\mX} k(x, x)} \leq C\sqrt{\sup_{x\in\mX} k(x, x)} <\infty,\]
$\{f_n\}$ is uniformly bounded.  To check that $\{f_n\}$ is equicontinuous, note that
\begin{equation}
|f_n(x) - f_n(x')| = |\langle f_n, k_x-k_{x'}\rangle_k| \leq \|f_n\|_k\|k_x-k_{x'}\|_k \leq C\|k_x-k_{x'}\|_k,  \label{equi}
\end{equation}
and
\begin{align}
\|k_x-k_{x'}\|_k^2 &= \langle k_x-k_{x'}, k_x-k_{x'}\rangle_k = k(x, x) - 2k(x, x') + k(x', x') \nonumber\\
&\leq 2\sup_{z\in\mX}|k(z, x)-k(z, x')|. \label{equi2}
\end{align}
It then follows from \eqref{equi}, $\eqref{equi2}$ and the uniform continuity of $k$ on $\mX\times\mX$ that $\{f_n\}$ is equicontinuous.  Applying the Arzel\`{a}-Ascoli Theorem, $\{f_n\}$ contains a uniformly convergent subsequence $\{f_{n_k}\}$. Since $\|f_{n_k}\|_{P} \leq 2\|f_{n_k}\|_\infty$, it follows that $\{f_{n_k}\}$ also converges in $L^2(\mX, dP)$. So the inclusion $\I: \Hsp\hookrightarrow L^2(\mX, dP)$ is compact.
\end{itemize}

To show \eqref{contraction}, recall that Riesz Representation Theorem tells us that if $\ell$ is a bounded linear functional on $\Hsp$ with representer $h_\ell\in\Hsp$, i.e. $\ell(f) = \langle h_\ell, f\rangle_\star$ for all $f\in\Hsp$, then $\|\ell\| = \|h_\ell\|_\star$.  In the case where $\ell(f) = \Cov[\phi_j(X_j), f(X_i)] = \langle h, f\rangle_{\star_i}$ for all $f\in\Hsp_i$, we have 
\begin{align*}
\|\bS_{ij}\phi_j\|_{\star_i} &= \|h\|_{\star_i} = \|\Cov[\phi_j(X_j), \cdot(X_i)]\| = \sup_{\|f\|_{\star_i}\leq 1} |\Cov[\phi_j(X_j), f(X_i)]| \\
&\leq \sup_{\|f\|_{\star_i}\leq 1}\big(\Var[\phi_j(X_j)]\Var[f(X_i)]\big)^{1/2} \\
&\leq \sup_{\|f\|_{\star_i}\leq 1}\big(\Var[\phi_j(X_j)]\big)^{1/2}\|f\|_{\star_i} = \big(\Var[\phi_j(X_j)]\big)^{1/2} \leq \|\phi_j\|_{\star_j}.
\end{align*}
\end{proof}

% Proof of restatement of regularized APC problem using the S operator
\subsection*{Proof of Theorem~\ref{restateAPCwithS}}
\begin{proof}
First, note that we can rewrite the optimization criterion in the kernelized population APC problem as
\begin{align*}
&\Var\bigg[\sum_i\phi_i(X_i)\bigg] + \sum_i\alpha_i\|\phi_i\|_{k^1_i}^2 \\
&= \sum_i\Var[\phi_i(X_i)] + \sum_i\alpha_i\|\phi_i\|_{k^1_i}^2 + \sum_i\sum_{j\neq i}\Cov[\phi_i(X_i), \phi_j(X_j)] \\
&= \sum_i\|\phi_i\|_{\star_i}^2 + \sum_i\sum_{j\neq i}\langle\phi_i, \bS_{ij}\phi_j\rangle_{\star_i} \\
&= \sum_i\bigg\langle\phi_i, \sum_{j\neq i}\bS_{ij}\phi_j + \phi_i\bigg\rangle_{\star_i} \\
&= \langle\bPhi, \tS\bPhi\rangle_\star \qquad\geq 0.
\end{align*}
Hence, $\tS$ is positive.  That the constraint $\sum\Var\phi_i(X_i) + \sum\alpha_i\|\phi_i\|_{k^1_i}^2 = \langle\bPhi, \bPhi\rangle_\star$ follows by definition.

To see that $\tS$ is self-adjoint, we need to show that $\langle\bPhi, \tS\bPsi\rangle_\star = \langle\tS\bPhi, \bPsi\rangle_\star$.  Since 
\[\Cov[\phi_i(X_i), \psi_j(X_j)] = \langle\phi_i, \bS_{ij}\psi_j\rangle_{\star_i} = \langle \bS_{ji}\phi_i, \psi_j\rangle_{\star_j}, \qquad\forall\phi_i\in\Hsp_i, \psi_j\in\Hsp_j,\]
giving $\bS_{ij}^* = \bS_{ji}$, we have
\begin{align*}
\langle\bPhi, \tS\bPsi\rangle_\star &= \sum_i\langle\phi_i, [\tS\bPsi]_i\rangle_{\star_i} = \sum_i\bigg\langle\phi_i, \sum_{j\neq i}\bS_{ij}\psi_j + \psi_i\bigg\rangle_{\star_i} \\
&= \sum_i\sum_{j\neq i}\langle\phi_i, \bS_{ij}\psi_j\rangle_{\star_i} + \sum_i\langle\phi_i, \psi_i\rangle_{\star_i} = \sum_j\sum_{i\neq j}\langle \bS_{ji}\phi_i, \psi_j\rangle_{\star_j} + \sum_j\langle\phi_j, \psi_j\rangle_{\star_j} \\
&= \sum_j \bigg\langle \sum_{i\neq j}\bS_{ji}\phi_i + \phi_j, \psi_j\bigg\rangle_{\star_j} = \sum_j\langle [\tS\bPhi]_j, \psi_j\rangle_{\star_j} = \langle\tS\bPhi, \bPsi\rangle_\star,
\end{align*}
and so $\tS$ is self-adjoint.

To check that $\tS$ is bounded above by $p$, by \eqref{contraction}, we have $\|\bS_{ij}\phi_j\|_{\star_i} \leq \|\phi_j\|_{\star_j}$.  Therefore, 
\begin{align*}
\|\tS\bPhi\|_\star^2 &= \sum_{i=1}^p\bigg\|\sum_{j\neq i}\bS_{ij}\phi_j + \phi_i\bigg\|_{\star_i}^2 \leq \sum_{i=1}^p \bigg(\sum_{j\neq i}\|\bS_{ij}\phi_j\|_{\star_i} + \|\phi_i\|_{\star_i}\bigg)^2 \\
&\leq \sum_{i=1}^p\bigg(\sum_{j=1}^p\|\phi_j\|_{\star_j}\bigg)^2 \qquad\text{since $\|\bS_{ij}\phi_j\|_{\star_i} \leq \|\phi_j\|_{\star_j}$} \\
&\leq p\cdot p\sum_{j=1}^p\|\phi_j\|_{\star_j}^2 \qquad\text{since $\bigg(\sum_{j=1}^pa_j\bigg)^2 \leq p\sum_{j=1}^pa_j^2$ if $a_j\geq 0$ for $1\leq j\leq p$} \\
&= p^2 \|\bPhi\|_\star^2,
\end{align*}
so $\|\tS\| = \sup_{\|\bPhi\|_\star\leq 1}\|\tS\bPhi\|_\star \leq p$.
\end{proof}

% Proof of S-I being compact
\subsection*{Proof of Theorem~\ref{SminusIcompact}}
\begin{proof}
We can decompose $\tS - \I$ as $\tS-\I = \sum_{j=1}^p\bT_j$, where $\bT_j: \bHsp\longrightarrow\bHsp$ is defined by
\[\bT_j(\bPhi) = (\bS_{1j}(\phi_j), \ldots, \bS_{j-1,j}(\phi_j), \0, \bS_{j+1,j}(\phi_j), \ldots, \bS_{pj}(\phi_j)).\]
It is sufficient to show that every summand $\bT_j$ is compact.  Let $\bB$ be the unit ball in $\bHsp$, our goal is then to show that $\bT_j(\bB)$ is a relatively compact set in $\bHsp$.  Let $B_j$ be the unit ball in $\Hsp_j$, then since $\bB\subset B_1\times\cdots\times B_p$, compactness of $\bT_j$ follows if $\bT_j(B_1\times\cdots\times B_p)$ is shown to be relatively compact.  By Theorem~\ref{milestone}, $\bS_{ij}(B_j)$ is relatively compact in $\Hsp_j$, so
\[\bT_j(B_1\times\cdots\times B_p) = \bS_{1j}(B_j)\times\cdots\times\bS_{j-1,j}(B_j)\times\{0\}\times\bS_{j+1,j}(B_j)\times\cdots\times\bS_{pj}(B_j)\]
is relatively compact in $\bHsp$, since the norm topology and the product topology coincide. 
\end{proof}

% Proof of convergence of power algorithm
\subsection*{Proof of Proposition~\ref{powerconv}}
\begin{proof}
Let $\bM = \gamma\I-\tS$, where $\gamma = (p+1)/2$.  Then $\tilde{\bPhi}$ is the unit eigenfunction corresponding to the largest eigenvalue $\lambda$ of $\bM$, and it is assumed that $\lambda$ has multiplicity one. By assumption, the power algorithm is initialized with $\bPhi^{[0]}$ that satisfies
\[\bPhi^{[0]} = a_0\tilde{\bPhi}+\bPsi^{[0]}, \qquad\text{where }\bPsi^{[0]}\perp\tilde{\bPhi}\text{ and }a_0>0.\]
Let
\[\bPhi^{[t+1]} = \frac{\bM\bPhi^{[t]}}{\|\bM\bPhi^{[t]}\|_\star},\]
and suppose that
\[\bPhi^{[t]} = a_t\tilde{\bPhi}+\bPsi^{[t]}, \qquad\text{where }\bPsi^{[t]}\perp\tilde{\bPhi}.\]
Then
\[\bPhi^{[t+1]} = \frac{\bM\bPhi^{[t]}}{\|\bM\bPhi^{[t]}\|_\star} = \frac{\bM(a_t\tilde{\bPhi}+\bPsi^{[t]})}{\|\bM\bPhi^{[t]}\|_\star} = \frac{a_t\lambda}{\|\bM\bPhi^{[t]}\|_\star}\tilde{\bPhi} + \frac{\bM \bPsi^{[t]}}{\|\bM\bPhi^{[t]}\|_\star}.\]
Matching the coefficients, we see that
\begin{equation}
a_{t+1} = \frac{a_t\lambda}{\|\bM\bPhi^{[t]}\|_\star},\qquad \bPsi^{[t+1]} = \frac{\bM \bPsi^{[t]}}{\|\bM\bPhi^{[t]}\|_\star}, \label{update}
\end{equation}
and it follows that $a_0>0$ implies $a_t>0$ for all $t\in\N$.  Now note that for $\bPsi\perp\tilde{\bPhi}$,
\begin{equation}
\|\bM\bPsi\|_\star \leq r\|\bPsi\|_\star, \qquad\text{where } r<\lambda, \label{powconv}
\end{equation}
so by \eqref{update} and \eqref{powconv},
\[\frac{\|\bPsi^{[t+1]}\|_\star}{a_{t+1}} = \frac{\|\bM \bPsi^{[t]}\|_\star}{a_t\lambda} \leq \bigg(\frac{r}{\lambda}\bigg)\frac{\|\bPsi^{[t]}\|_\star}{a_t},\]
which in turn implies
\begin{equation}
\frac{\|\bPsi^{[t]}\|_\star}{a_t}\leq \bigg(\frac{r}{\lambda}\bigg)^t\frac{\|\bPsi^{[0]}\|_\star}{a_0}\longrightarrow 0\text{ as }t\longrightarrow\infty. \label{end}
\end{equation}
To show that $\bPhi^{[t]}\longrightarrow\tilde{\bPhi}$, note that $\|\bPhi^{[t]}\| = 1$ implies that
\[\|a_t\tilde{\bPhi} + \bPsi^{[t]}\|_\star = 1 \iff a_t^2 + \|\bPsi^{[t]}\|_\star^2 = 1 \iff 1+\frac{\|\bPsi^{[t]}\|_\star^2}{a_t^2} = \frac{1}{a_t^2}.\]
By \eqref{end} we have $a_t^2\longrightarrow 1$ and $\|\bPsi^{[t]}\|_\star^2\longrightarrow 0$, hence
\[\|\bPhi^{[t]} - \tilde{\bPhi}\|_\star^2 = (1-a_t)^2 + \|\bPsi^{[t]}\|_\star^2 \longrightarrow 0.\]
\end{proof}

%%%%%%%%%%%%%%%%%%%%%%%%%%%%%%%%%%%%%%%%%%%%%%%%%%%%%

%!TEX root = Paper.tex

% Details of the power algorithm
\section{Implementation Details of the Power Algorithm}
\label{sec:power-details}

We justified the use of a smoothing-based power algorithm in computing 
kernelized population APCs in Section~\ref{compute}.  In this section, we give a
detailed description of its empirical implementation.  

We first need to resolve the issue that the function space $\bHsp$ in the 
kernelized sample APC problem \eqref{kernSampleAPC} are (almost always) 
infinite-dimensional, which can pose challenges computationally.  As will be
shown, the solution to \eqref{kernSampleAPC} always lie in a finite-dimensional 
subspace of $\bHsp$.  Consider the smoothing splines problem
\begin{equation}
\min_{f\in\Hsp} \left\{\frac{1}{n}\sum_{i=1}^n(y_i-f(x_i))^2 + \alpha\|f\|_{k^1}^2\right\}, \label{smoothsplines}
\end{equation}
where $\Hsp = \Hsp^0\oplus\Hsp^1$, $\Dim(\Hsp^0) = m < n$, and $\Hsp^1$ is an RKHS with reproducing kernel $k^1$.  It is known that the solution $\hat{f}$ of \eqref{smoothsplines} must lie in a finite-dimensional subspace of $\Hsp$.  Specifically, write $\hat{f} = \hat{f}^0+\hat{f}^1$ with $\hat{f}^0\in\Hsp^0, \hat{f}^1\in\Hsp^1$, then 
\[\hat{f}^1\in\Span\{k^1(x_i, \cdot): 1\leq i\leq n\}.\]  
In essence, this means that to solve \eqref{smoothsplines}, only the representers of the evaluation functionals (projected to $\Hsp^1$) at the locations of the observed data matters.  This is known as the \emph{Representer Theorem for smoothing splines}.  A more general version of this Representer Theorem, adapted to the case of kernelized APCs, states that for any probability measure $P_j(dx_j)$, not necessarily an empirical measure, only the representers of the evaluation functionals at the locations that belong to the support of $P_j$ matters.  

\begin{Theorem}[Representer Theorem for Kernelized APCs] \label{representer}
The solution to the kernelized APC problem \eqref{statePenSample}, 
if exists, is taken on the subspace $\bHsp_P := \Hsp_{P_1}\times\Hsp_{P_2}\times\cdots\times\Hsp_{P_p}$, where
\[\Hsp_{P_j} := \Hsp^0_j\oplus\overline{\Span\{k_j^1(x, \cdot)-m_j^1: x\in\Supp(P_j)\}},\]
and $m_j^1$ is the mean element of $\Hsp_j^1$ with respect to $P_j$.
\end{Theorem}
\begin{proof}
Let $m_j = m_j^0 + m_j^1$ be the mean element of $\Hsp_j$ with respect to $P_j$, where $m_j^0\in\Hsp_j^0$ and $m_j^1\in\Hsp_j^1$. For $\psi_j\in\Hsp_j$, we have
\[\psi_j\perp \Hsp_{P_j}\Rightarrow\psi_j(x) - E\psi_j = 0, \qquad\text{for $x\in\Supp(P_j)$}.\]
This is because $\psi_j(x) - E\psi_j = \langle\psi_j, k_j(x, \cdot)-m_j\rangle_{k_j} = \langle \psi_j, k_j^0(x, \cdot)-m_j^0\rangle_{k_j} + \langle \psi_j, k_j^1(x, \cdot)-m_j^1\rangle_{k_j}$ and $k_j^0(x, \cdot)- m_j^0\in\Hsp_j^0$.
So given $\psi_j\perp \Hsp_{P_j}$ and $\phi_j\in \Hsp_{P_j}$, we have
$\Var\psi_j = 0$ and $\|\phi_j+\psi_j\|_{k_j^1}^2 = \|\phi_j\|_{k_j^1}^2 + \|\psi_j\|_{k_j^1}^2$, which implies that
\[\Var\sum_{j=1}^p(\phi_j+\psi_j) = \Var\sum_{j=1}^p\phi_j, \qquad\sum_{j=1}^p\alpha_j\|\phi_j+\psi_j\|_{k_j^1}^2 \geq \sum_{j=1}^p\alpha_j\|\phi_j\|_{k_j^1}^2,\]
and the inequality is strict when $\psi_i\not\equiv 0$ for some $1\leq i\leq p$.

Now, suppose on the contrary that $(\phi^*_1+\psi^*_1, \ldots, \phi^*_p+\psi^*_p)$ is the optimal solution of the kernelized APC problem, where $\phi^*_j\in\Hsp_{P_j}, \psi^*_j\perp\Hsp_{P_j}$ and $\psi^*_i\not\equiv 0$ for some $1\leq i\leq p$.  Let 
\[\delta = \sum_{j=1}^p\alpha_j\|\phi^*_j+\psi^*_j\|_{k_j^1}^2,\] 
then $(\phi^*_1+\psi^*_1, \ldots, \phi^*_p+\psi^*_p)$ is also an optimal solution of the following optimization problem:
\begin{equation}
\min_{\bPhi\in\bHsp}\Var\sum_{j=1}^p\phi_j + \sum_{j=1}^p\alpha_j\|\phi_j\|_{k_j^1}^2\qquad\text{subject to}\qquad\sum_{j=1}^p\Var\phi_j = 1-\delta. \label{optproof}
\end{equation}
But as argued before we have $\Var\sum(\phi^*_j+\psi^*_j) = \Var(\sum\phi^*_j)$ and $\sum\alpha_j\|\phi^*_j+\psi^*_j\|_{k_j^1}^2 > \sum\alpha_j\|\phi^*_j\|_{k_j^1}^2$.  Also, subject to the constraint that $\sum\Var(\phi^*_j+\psi^*_j) = 1-\delta$, we have $\sum\Var\phi^*_j = 1-\delta$.  This gives the desired contradiction since in this case $(\phi^*_1, \ldots, \phi^*_p)$ is a better solution of \eqref{optproof} comparing to the optimal solution $(\phi^*_1+\psi^*_1, \ldots, \phi^*_p+\psi^*_p)$.  Therefore, we must have $\psi^*_j \equiv 0$ for $1\leq j\leq p$, and the proof is complete.
\end{proof}

Note that in the case where $P_j$ denotes the empirical probability measure with only finitely many values $\{x_{1j}, \ldots, x_{nj}\}$ in its support, Theorem~\ref{representer} specializes to the finite-sample version of the Representer Theorem for kernelized APCs:  
\begin{Corollary}\label{corRSB}
Given data $\x_i = (x_{i1}, \ldots, x_{ip})$, $1\leq i\leq n$, the solution of the kernelized sample APC problem \eqref{kernSampleAPC},
if exists, is taken on the finite-dimensional subspace $\bHsp_n := \Hsp_{n, 1}\times\cdots\times\Hsp_{n, p}$, where
\[\Hsp_{n, j} := \Hsp_j^0\oplus\Span\bigg\{k_j^1(x_{ij}, \cdot) - \frac{1}{n}\sum_{a=1}^nk_j^1(x_{aj}, \cdot): 1\leq i\leq n\bigg\}.\]
\end{Corollary}

One can similarly show that other higher-order kernelized sample APCs, if exists, also lie in the finite-dimensional subspace $\bHsp_n$.

To implement the power algorithm, it follows from Corollary~\ref{corRSB} that it suffices to work with the coefficients of the basis of $\Hsp_{n, i}$.  Specifically, let $\phi_i = \sum_{\ell=1}^n\beta_{\ell i} f_{\ell i} + \sum_{\ell=1}^{m_i}\beta_{n+\ell, i}q_{\ell i}$, where $f_{\ell i} = k_i^1(x_{\ell i}, \cdot) - \frac{1}{n}\sum_{a=1}^nk_i^1(x_{ai}, \cdot)$, $1\leq\ell\leq n$.  Then the update steps $\phi_i \leftarrow\gamma\phi_i^{[t]}-(\sum_{j\neq i}\bS_{ij}\phi_j^{[t]}+\phi_i^{[t]})$ in Algorithm~\ref{PowerAlgo} becomes
\begin{align}
&\beta_{\ell i} \leftarrow (\gamma-1)\beta_{\ell i}^{[t]} - c_{\ell i}, && 1\leq\ell \leq n, \label{powercoeff}\\
&\beta_{n+\ell, i} \leftarrow (\gamma-1)\beta_{n+\ell, i}^{[t]} - d_{\ell i}, &&1\leq\ell \leq m_i, \nonumber
\end{align}
where $\{c_{\ell i}\}_{\ell=1}^n$ and $\{d_{\ell i}\}_{\ell=1}^{m_i}$ are two sets of coefficients obtained from the smoothing step $\sum_{j\neq i}\bS_{ij}\phi_j^{[t]}$, which will be derived shortly.  

Let $\bbeta_i = (\bbeta_{1i}, \ldots, \bbeta_{ni})$, and let $\G_i$ be the $n\times n$ centered Gram matrix associated with $k_i^1$, with $(j, \ell)$ entry
\begin{align}
(\G_i)_{j\ell} &= \langle f_{ji}, f_{\ell i}\rangle_{k_i^1} \label{centeredGram}\\
&= k_i^1(x_{ji}, x_{\ell i}) - \frac{1}{n}\sum_{b=1}^nk_i^1(x_{ji}, x_{bi}) - \frac{1}{n}\sum_{a=1}^nk_i^1(x_{ai}, x_{\ell i}) + \frac{1}{n^2}\sum_{a=1}^n\sum_{b=1}^nk_i^1(x_{ai}, x_{bi}). \nonumber
\end{align}
Then the normalizing constant $c$ in Algorithm~\ref{PowerAlgo} can be obtained upon computation of the variance of the transformed data points $\{\phi_i(x_{\ell i})\}_{\ell=1}^n$ and the penalty term $\|\phi_i\|^2_{k_i^1} = \bbeta_i^T\G_i\bbeta_i$, for $1\leq i\leq p$.

We now consider the smoothing step $\sum_{j\neq i}\bS_{ij}\phi_j$, which by linearity of smoothing is empirically the regularized least squares regression of $\sum_{j\neq i}\phi_j$ against $X_i$.  This amounts to solving the following optimization problem:
\begin{equation}
\min_{f\in\Hsp_i} \left\{\Varhat\bigg[\sum_{j\neq i}\phi_j(X_j)-f(X_i)\bigg] + \alpha_i \|f\|_{k^1_i}^2\right\}, \label{empSmooth}
\end{equation}
where $\Varhat[\sum_{j\neq i}\phi_j(X_j)-f(X_i)]$ evaluates to
\[\frac{1}{n}\sum_{\ell=1}^n\bigg[\sum_{j\neq i}\Big(\phi_j(x_{\ell j}) - \frac{1}{n}\sum_{b=1}^n\phi_j(x_{bj})\Big) - \Big(f(x_{\ell i}) - \frac{1}{n}\sum_{a=1}^nf(x_{ai})\Big)\bigg]^2.\]
We see that \eqref{empSmooth} is essentially the smoothing splines problem \eqref{smoothsplines} (modulo centering), hence it is not surprising that its solution lies in $\Hsp_{n, i}$ as well. 

Following \cite{Wahba1990} (page 11-12), let the closed form solution of \eqref{empSmooth} be
\[f = \sum_{\ell=1}^n c_{\ell i} f_{\ell i} + \sum_{\ell=1}^{m_i} d_{\ell i} q_{\ell i},\]
and \eqref{empSmooth} can be restated as
\begin{equation}
\min_{\bc\in\R^n, \bd\in\R^{m_i}} \left\{\frac{1}{n}\|\y-(\G_i\bc+\Q_i\bd)\|^2 + \alpha_i\bc^T\G_i\bc\right\}, \label{empSmoothAlg}
\end{equation}
where $\bc^T = (c_{1i}, \ldots, c_{ni}), \bd^T = (d_{1i}, \ldots, d_{m_ii}), \y^T = (y_1, \ldots, y_n)$ with $y_\ell = \sum_{j\neq i}[\phi_j(x_{\ell j}) - \frac{1}{n}\sum_{b=1}^n\phi_j(x_{bj})]$ for $1\leq\ell\leq n$, $\G_i$ is as given in \eqref{centeredGram}, and $\Q_i$ is the column-centered version of 
\[\tilde{\Q}_i = \begin{pmatrix}q_{1i}(x_{1i}) & \cdots & q_{m_ii}(x_{1i}) \\\vdots & \vdots & \vdots \\ q_{1i}(x_{ni}) &\cdots & q_{m_ii}(x_{ni}) \end{pmatrix}.\]
It then follows that the solution of \eqref{empSmoothAlg} is
\[\bd = (\Q_i^T\bM_i^{-1}\Q_i)^{-1}\Q_i^T\bM_i^{-1}\y, \qquad \bc = \bM_i^{-1}(\y - \Q_i\bd),\]
where $\bM_i = \G_i + n\alpha_i\I$, $\I$ being the $n\times n$ identity matrix.  Plugging $\bc$ and $\bd$ into \eqref{powercoeff} completes the update steps.

% Reduction to Linear Algebra Problem
\section{A Direct Approach for Computing APCs}
\label{linAlg}

In this section, we give a direct approach for computing APCs.  

From Corollary~\ref{corRSB}, we know that the solution $\hat{\bPhi} = (\hat{\phi}_1, \ldots, \hat{\phi}_p)$ of the kernelized sample APC problem \eqref{kernSampleAPC} lies in the finite-dimensional function space $\bHsp_n = \Hsp_{n, 1}\times\cdots\times\Hsp_{n, p}$.  In the following, we derive the resulting linear algebra problem in terms of the coefficients with respect to the basis of $\Hsp_{n, j}$'s.  We will focus on the case where  there are no null spaces, i.e. $\Hsp_j = \Hsp_j^1$ and $k_j = k_j^1$, for $1\leq j\leq p$.  The case with null spaces requires the use of the additional basis $\{q_{1j}, \ldots, q_{m_jj}\}$ for $\Hsp_j^0$, $1\leq j\leq p$, which is tractable but with slightly more tedious derivation.  We recommend the use of power algorithm described in Section~\ref{compute} when dealing with cases involving null spaces.  The power algorithm is computationally more attractive than the direct linear algebra approach given below, when the interest is only in extracting a few eigenfunctions.

For each $1\leq j\leq p$, we express $\phi_j\in\Hsp_{n, j}$ as $\phi_j = \sum_{i=1}^n\beta_{ij}f_{ij}$, where
\[f_{ij}(\cdot) := k_j^1(x_{ij}, \cdot) - \frac{1}{n}\sum_{a=1}^nk_j^1(x_{aj}, \cdot), \qquad 1\leq i\leq n. \]
Then
\begin{align*}
\sum_{j=1}^p\phi_j &= \sum_{j=1}^p\sum_{i=1}^n\beta_{ij}f_{ij} = \sum_{j=1}^p \bbeta_j^T\f_j \\
&\qquad\text{where }\bbeta_j^T = (\beta_{1j}, \ldots, \beta_{nj}),\ \f_j^T = (f_{1j}, \ldots, f_{nj})\\
&= \bbeta^T\F \qquad\text{where }\bbeta^T = (\bbeta_1^T, \ldots, \bbeta_p^T),\ \F^T = (\f_1^T, \ldots, \f_p^T).
\end{align*}
The penalty term associated with $\phi_j$ evaluates to
\[\|\phi_j\|_{k_j}^2 = \bigg\langle\sum_{i=1}^n\beta_{ij}f_{ij}, \sum_{\ell=1}^n\beta_{\ell j}f_{\ell j}\bigg\rangle_{k_j} = \sum_{i=1}^n\sum_{\ell=1}^n\beta_{ij}\beta_{\ell j}\langle f_{ij}, f_{\ell j}\rangle_{k_j} = \bbeta_j^T\G_j\bbeta_j,\]
where $\G_j$ is the centered Gram matrix associated with $k_j$, with $(i, \ell)$ entry
\begin{align*}
(\G_j)_{i\ell} &= \langle f_{ij}, f_{\ell j}\rangle_{k_j} \\
&= k_j(x_{ij}, x_{\ell j}) - \frac{1}{n}\sum_{b=1}^nk_j(x_{ij}, x_{bj}) - \frac{1}{n}\sum_{a=1}^nk_j(x_{aj}, x_{\ell j}) + \frac{1}{n^2}\sum_{a=1}^n\sum_{b=1}^nk_j(x_{aj}, x_{bj}).
\end{align*}
Therefore, we can rewrite the penalty term as
\[\sum_{j=1}^p\alpha_j\|\phi_j\|_{k_j}^2 = \sum_{j=1}^p\alpha_j\bbeta_j^T\G_j\bbeta_j.\]
The variance term in the criterion evaluates to
\[\Varhat\sum_{j=1}^p\phi_j = \Varhat(\bbeta^T\F) = \frac{1}{n}\bbeta^T\G\G^T\bbeta,\]
where $\G^T = (\G_1, \cdots, \G_p)$.
Meanwhile, the variance term in the constraint is 
\[\sum_{j=1}^p\Varhat\phi_j = \sum_{j=1}^p\Varhat(\bbeta_j^T\f_j) = \frac{1}{n}\sum_{j=1}^p\bbeta_j^T\G_j^2\bbeta_j.\]
Hence, the optimization problem \eqref{kernSampleAPC}, expressed in linear algebra notation, becomes
\begin{align}
&\min_{\bbeta\in\R^{pn}} &&\frac{1}{n}\bbeta^T\G\G^T\bbeta + \bbeta^T\text{diag}(\alpha_1\G_1, \ldots, \alpha_p\G_p)\bbeta \label{beta}\\
&\text{subject to } &&\frac{1}{n}\bbeta^T\text{diag}(\G_1^2, \ldots, \G_p^2)\bbeta + \bbeta^T\text{diag}(\alpha_1\G_1, \ldots, \alpha_p\G_p)\bbeta = 1. \nonumber
\end{align}
Equivalently, we want to solve the following generalized eigenvalue problem:
\begin{align}
&\begin{pmatrix}\G_1^2+n\alpha_1\G_1 & \G_1\G_2 & \cdots & \G_1\G_p\\ \G_2\G_1 & \G_2^2+n\alpha_2\G_2 & \cdots & \G_2\G_p\\ \vdots & \vdots & \ddots & \vdots \\ \G_p\G_1 & \G_p\G_2 & \cdots & \G_p^2 + n\alpha_p\G_p \end{pmatrix}\bbeta \label{solveLinAlg}\\
&\qquad= \lambda\begin{pmatrix}\G_1^2+n\alpha_1\G_1& \0 & \cdots & \0 \\ \0 & \G_2^2+n\alpha_2\G_2 & \ddots & \0 \\ \vdots & \vdots & \ddots & \vdots \\ \0 & \0 & \cdots & \G_p^2+n\alpha_p\G_p\end{pmatrix}\bbeta. \nonumber
\end{align}
Following \cite{Bach2003}, we can approximate the diagonal blocks $\G_j^2+n\alpha_j\G_j$ in \eqref{solveLinAlg} by $(\G_j+\frac{n\alpha_j}{2}\I)^2$. Letting $\bgamma_j = (\G_j+\frac{n\alpha_j}{2}\I)\bbeta_j$ allows the reformulation of the generalized eigenproblem above as an eigenproblem, in which case we just need to perform eigendecomposition on
\[\bR = \begin{pmatrix}\I & \bR_1^T\bR_2 & \cdots & \bR_1^T\bR_p\\ \bR_2^T\bR_1 & \I & \cdots & \bR_2^T\bR_p\\ \vdots & \vdots & \ddots & \vdots \\ \bR_p^T\bR_1 & \bR_p^T\bR_2 & \cdots & \I \end{pmatrix},\]
where $\bR_j = \G_j(\G_j+\frac{n\alpha_j}{2}\I)^{-1}$ and $\I$ is the $n\times n$ identity matrix, to get its eigenvector $\hat{\bgamma} = (\hat{\bgamma}_1, \ldots, \hat{\bgamma}_p)$ (corresponding to the smallest eigenvalue).  The desired (approximate) solution of \eqref{beta} can then be obtained as $\hat{\bbeta}_j = (\G_j+\frac{n\alpha_j}{2}\I)^{-1}\hat{\bgamma}_j$, while the (mean-centered) estimated transform evaluated at the data points is 
\[\hat{\bphi}_j = \G_j\hat{\bbeta}_j = \G_j\Big(\G_j+\frac{n\alpha_j}{2}\I\Big)^{-1}\hat{\bgamma}_j.\]
The second smallest and subsequent higher order kernelized sample APCs can be obtained similarly by extracting the eigenvector corresponding to the second smallest and subsequent smallest eigenvalue of $\bR$.

We remark that the linear algebra problem \eqref{beta} is often numerically ill-conditioned due to low-rankness of $\G_j$, so one has to make adjustment in order to solve for APCs.  This, however, introduces undesirable arbitrariness to the resulting optimization problem.

%%%%%%%%%%%%%%%%%%%%%%%%%%%%%%%%%%%%%%%%%%%%%%%%%%%%%

\end{document}